%% file: main.tex
\documentclass[english]{article}
\usepackage[utf8]{inputenc}
\usepackage[margin=2.5cm]{geometry}
\usepackage{babel}

\usepackage{amssymb}
\usepackage{amsmath}
\usepackage{amsthm}
\usepackage{dsfont} 
\usepackage{bbm}
\usepackage[colorlinks,linkcolor=blue,citecolor=red,urlcolor=blue]{hyperref}
\usepackage{enumitem}
\usepackage{bm}
\usepackage{xspace}

\usepackage[ruled]{algorithm}
\usepackage{algorithmic}

\input{definitions}

\newtheorem{prop}{Proposition}
\newtheorem{assumption}{Assumption}
\newtheorem{definition}{Definition}
\newtheorem{lemma}{Lemma}
\usepackage{authblk}

\usepackage[textwidth=1.8cm, textsize=scriptsize]{todonotes}

\title{Solving Fredholm Integral Equations of the Second Kind\\ via Wasserstein Gradient Flows}
\author[1]{Francesca R. Crucinio\thanks{Corresponding author: francescaromana.crucinio@unito.it}}
\author[2]{Adam M. Johansen}
\affil[1]{ESOMAS, University of Turin, Italy \& Collegio Carlo Alberto, Turin, Italy}
\affil[2]{Department of Statistics, University of Warwick}
\date{ }

\begin{document}
\maketitle

\abstract{Motivated by a recent method for approximate solution of Fredholm equations of the first kind, we develop a corresponding method for a class of Fredholm equations of the \emph{second kind}. In particular, we consider the class of equations for which the solution is a probability measure. The approach centres around specifying a functional whose gradient flow admits a minimizer corresponding to a regularized version of the solution of the underlying equation and using a mean-field particle system to approximately simulate that flow. Theoretical support for the method is presented, along with some illustrative numerical results.}

\graphicspath{{Images/}}

\input{maintext.tex}

\paragraph*{Acknowledgments}

AMJ acknowledges the financial support of the United Kingdom Engineering and Physical Sciences Research Council (EPSRC; grants EP/R034710/1 and EP/T004134/1) and by United Kingdom Research and Innovation (UKRI) via grant number EP/Y014650/1, as part of the ERC Synergy project OCEAN. 

FRC gratefully acknowledges the ``de Castro" Statistics Initative at the \textit{Collegio Carlo Alberto} and the \textit{Fondazione Franca e Diego de Castro}.
FRC is supported by the Gruppo
Nazionale per l'Analisi Matematica, la Probabilità e le loro Applicazioni (GNAMPA-INdAM).

For the purpose of open access, the authors have applied a Creative Commons Attribution (CC BY) licence to any Author Accepted Manuscript version arising from this submission.

\medskip\noindent\textbf{Data access statement:} No new data was created during this research. Julia code to reproduce all examples is available at \\  \url{https://github.com/FrancescaCrucinio/FE2kind_WGF}.

\bibliographystyle{abbrv}  
\bibliography{wgf_secondkind_biblio.bib}
\cleardoublepage
\appendix

\input{subdifferentials}
\input{simods_mkvsde}
\input{siam_supplement}
\input{additional_expe}

\end{document}

%% file: definitions.tex
\newcommand{\1}{\mathbbm{1}}
\def\kker{\mathrm{k}}

\newcommand{\mcb}[1]{\mathcal{B}(#1)}
\def\rset{\mathbb{R}}
\def\nset{\mathbb{N}}
\def\Exp{\mathbb{E}}
\def\rmd{\mathrm{d}}
\def\Fun{\mathcal{F}}
\def\Gun{\mathcal{G}}
\def\Mun{\mathcal{M}}
\def\Lun{\mathcal{L}}
\newcommand{\KL}[2]{\textup{KL}\left( #1 \middle\vert #2 \right)}
\def\Hent{\mathrm{H}}
\def\Pens{\mathcal{P}}
\newcommand{\N}{\mathcal{N}}
\newcommand{\wassersteinD}[1][1=\distance]{\mathbf{W}_{#1}}
\def\msa{\mathsf{A}}
\newcommand{\norm}[1]{\ensuremath{\left\Vert #1 \right\Vert}\xspace}
\def\nset{\mathbb{N}}
\def\nsets{\mathbb{N}^{\star}}
\def\Mtt{\mathtt{M}}
\def\Ctt{\mathtt{C}}
\def\Ltt{\mathtt{L}}
\def\msx{\mathsf{X}}
\newcommand{\ocint}[1]{\left(#1\right]}
\def\rms{\mathrm{s}}
\def\Id{\operatorname{Id}}
\def\rmL{\mathrm{L}}
\def\rmc{\mathrm{C}}
\def\mtt{\mathtt{m}}
\def\ctt{\mathtt{c}}
\DeclareMathOperator{\ise}{ISE}
\newcommand{\ind}{\mathds{1}}
\DeclareMathOperator{\mse}{MSE}

%% file: maintext.tex
\section{Introduction}
Fredholm integral equations of the second kind are defined by
\begin{equation}
  \label{eq:fe}
\pi(x) = \varphi(x) + \lambda \int_{\rset^d} \kker(x, y) \pi( y)\rmd y 
\end{equation}
for any $x \in \rset^d$ and some $\lambda\in\rset$, with $\pi$ an unknown function on
$\rset^d$, $\varphi$ an observed forcing on $\rset^d$ and
$\kker: \ \rset^d\times  \rset^d \to \rset$ a positive integral kernel---where $\pi$ and $\kker$ belong to suitably regular classes.
We focus on the case in which $\pi$ is a positive function with known finite integral, i.e. $\int_{\rset^d} \pi(x)\rmd x<\infty$, so that we can renormalize the left hand side of~\eqref{eq:fe} and consider the normalized version of $\pi$ and define the new forcing as $\varphi(x)/\int_{\rset^d} \pi(x)\rmd x$ (which we denote by $\pi, \varphi$ too with a slight abuse of notation).
Such equations find applications throughout statistics and data science in contexts as diverse as econometrics \cite{dibu2021delayed, carrasco2007linear}, modelling light transport \cite{veach1998robust, veach1997metropolis} and reinforcement learning \cite{dahm2017learning}. In the homogeneous case (i.e. when $\varphi\equiv 0$), solving~\eqref{eq:fe} corresponds to finding the eigenfunction $\pi$ of $\kker$ for a given eigenvalue $\lambda$; the eigen--decomposition of positive kernels $\kker$ is normally referred to as Karhunen--Lo{\`e}ve decomposition (see \cite{daw2022overview} for a recent review) and has applications to Gaussian process regression \cite{NIPS2000_19de10ad}, spatial statistics \cite{cressie2008fixed, sang2012full} and Markov chain Monte Carlo (MCMC) in infinite dimensional spaces \cite{agapiou2018unbiased}. When $\kker$ is the density of a Markov kernel (i.e. $\int \kker(x, y)\rmd y=1$) and $\lambda=1$, solving~\eqref{eq:fe} is equivalent to finding an invariant measure of $\kker$. This is useful when dealing with approximate MCMC kernels \cite{alquier2016noisy, medina2017stability} or tractable approximations to transition densities \cite{beckers2016equilibrium}.

Henceforth we will use the same symbol to denote measures and their Lebesgue densities whenever they possess them---particularly $\pi$ and $\kker$ such that, e.g., $\rmd \pi(y) = \pi(y)\rmd y$. To solve~\eqref{eq:fe} we consider a minimization problem with target functional
\begin{equation}
\label{eq:minimisation}
  \Fun_\alpha(\pi) = \KL{\pi}{\varphi+\lambda \int_{\rset^d} \kker(\cdot, y) \pi( y)\rmd y} + \alpha \KL{\pi}{\pi_0},
\end{equation}
where $\KL{\mu}{\nu}=\int_{\rset^d} \log ((\rmd \mu /\rmd \nu)(y)) \rmd
\mu(y)$ denotes the Kullback--Leibler divergence between $\mu$ and $\nu$, for a given regularization
parameter $\alpha>0$ and reference measure $\pi_0$. Henceforth, we implicitly work under the assumption that the $\pi$ of \eqref{eq:fe} is the Lebesgue density of a probability distribution over a Euclidean space.
Standard approaches normally rely upon discretization of $\pi$ and $\kker$ through basis functions or using numerical integration (see, e.g., \cite{atkinson1992survey} for a survey of classical methods). These methods are particularly effective in low-dimensional settings, but they become challenging when $d$ is not small and in the case in which the solution $\pi$ is defined over an unbounded set.
In this paper we focus on the latter issue, and develop a method to solve Fredholm integral equations of the second kind on unbounded domains.

Approaches based upon loss minimization have mostly focused on maximum entropy solutions \cite{mead1986approximate,islam2020approximating, jin2016solving} and $\mathbb{L}^2$ loss \cite{guan2022solving}. However, the use of the Kullback--Leibler divergence is common in the literature on Fredholm integral equations of the first kind \cite{green1990use, resmerita2007joint, chae2018algorithm, iusem1994new} and in that context a functional closely related to~\eqref{eq:minimisation} has been recently shown to have properties similar to those of Tikhonov regularization \cite{crucinio2022solving}.

We extend the approach of \cite{crucinio2022solving} to equations of the second kind, obtaining a particle approximation of a nonlinear stochastic differential equation (SDE) of McKean--Vlasov type (an MKVSDE) associated with the Wasserstein gradient flow
minimizing~\eqref{eq:minimisation}. We also study the convergence of the minimizers of~\eqref{eq:minimisation} when the regularization parameter $\alpha\to 0$.
This approach is particularly well-suited to integral equations on unbounded domains since is not based on a fixed discretization, but on an adaptive stochastic discretization with diffusive behaviour.

The main difficulty which occurs when adapting the approach of \cite{crucinio2022solving} to the context of Fredholm equations of the second kind is a consequence of the presence of the unknown distribution, $\pi$, on both sides of the first KL divergence in~\eqref{eq:minimisation}: this results in an MKVSDE whose drift depends upon nested expectations w.r.t. $\pi$. As a consequence, the algorithm developed herein has higher cost w.r.t. the number of particles in our particle system than that of \cite{crucinio2022solving} and the fundamental difference in the drift of the MKVSDE means that novel results are required to justify our approach. We derive explicit bounds on the error of the proposed numerical algorithm which can be used to set the number of particles in the system and the 
time discretization step.

Finally, we demonstrate empirically that our method provides good solutions to integral equations of the form~\eqref{eq:fe} on unbounded domains and can be applied more generally than other methods (e.g. those based on the von Neumann decomposition; see our first numerical experiment). We stress that the presence of the regularization term $\KL{\pi}{\pi_0}$ in~\eqref{eq:minimisation} is crucial for obtaining significant improvements in the case in which the integral operator induced by the kernel $\kker$ is not invertible or poorly conditioned---as one might expect, regularization is critical in order to obtain good solutions to this ill-posed inverse problem.

The remainder of the manuscript is organized as follows.  In Section~\ref{sec:functional} we
study the minimization problem in~\eqref{eq:minimisation} and give conditions under which it admits a unique minimizer. We derive the Wasserstein gradient flow and the corresponding McKean--Vlasov SDE in Section~\ref{sec:wgf}. Section~\ref{sec:ips} introduces an interacting particle system approximating the MKVSDE and a numerical scheme which approximates the continuous time dynamics. We also provide explicit error bounds on the approximation.
Finally, in Section~\ref{sec:expe}, we test our method on several examples and compare with other methods commonly used in the literature. 

\subsection{Notation}
\label{sec:notation}
We endow $\rset^d$ with the Borel $\sigma$-field $\mcb{\rset^d}$ with respect to
the Euclidean norm $\norm{\cdot}$; when dealing with a matrix $A$ we consider the spectral norm $\norm{A}:=\sup_{\norm{u}=1}\norm{Au}$ induced by the Euclidean norm of vectors.   
We denote $\rmc(\rset^d, \rset^p)$ the set of continuous functions from $\rset^d$ to $\rset^p$ and $\rmc^n(\rset^d, \rset^p)$ the set of $n$-times differentiable functions
from $\rset^d$ to $\rset^p$ for any $n \in \nsets$.
For all differentiable $f$, we denote by $\nabla f$ its gradient. For all $f \in \rmc^1(\rset^{d_1} \times \dots \times \rset^{d_m}, \rset^p)$ with $m\in\nsets$, we denote by $\nabla_i f$ its gradient w.r.t. component $i \in \{1,\ldots,m\}$.  Furthermore, if $f$ is twice differentiable we denote by $\nabla^2f$ its Hessian and by $\Delta f$ its Laplacian.    
We say that a function $f: \ \msx \to \rset$ (where $\msx$ is a metric space) is coercive if for every $t \in \rset$, $f^{-1}(\ocint{-\infty, t})$ is relatively compact. We denote by $\Pens(\rset^d)$ the set of probability measures over
$\mcb{\rset^d}$, and endow this space with the topology of weak convergence. For any $p \in \nset$, we denote by
$\Pens_p(\rset^d) = \{\pi \in \Pens(\rset^d):\int_{\rset^d}
\norm{x}^p \rmd \pi(x) < +\infty\}$ the set of probability measures over
$\mcb{\rset^d}$ with finite $p$-th moment. We denote the subset of $\Pens(\rset^d)$ of measures which are absolutely continuous w.r.t. Lebesgue by $\Pens^{ac}(\rset^d)$ and let $\Pens^{ac}_p(\rset^d) := \Pens_p(\rset^d) \cap \Pens^{ac}(\rset^d)$. 
For any $\mu,\nu\in\Pens_p(\rset^d)$ we define the
$p$-Wasserstein distance $\wassersteinD[p](\mu, \nu)$ between $\mu$ and $\nu$ by
\begin{equation}
  \label{eq:def_distance_wasser}
    \wassersteinD[p](\mu, \nu) = \left(\inf_{\gamma \in \mathbf{T}(\mu,\nu)} \int_{\rset^d\times \rset^d}  \norm{x -y}^p \rmd \gamma (x,y)\right)^{1/p}
\end{equation}
where
$\mathbf{T}(\mu, \nu)=\{\gamma\in\Pens(\rset^d\times
  \rset^d):\gamma(\msa \times \rset^d) = \mu(\msa),\ \gamma(\rset^d \times
  \msa) = \nu(\msa)\ \forall \msa \in \mcb{\rset^d}\}$ denotes the set of all transport
plans between $\mu$ and $\nu$. In the following, we metrize $\Pens_p(\rset^d)$ with $\wassersteinD[p]$. 
For all $\nu\in\Pens^{ac}(\rset^d)$ we denote by $\Hent(\nu) :=-\int_{\rset^d}\log(\nu(y)) \rmd \nu(y)$ the differential entropy of $\nu$, where with a slight abuse of notation we denote by $\nu$ its density w.r.t. the Lebesgue measure. 
\section{The Minimization Problem}
\label{sec:functional}

We first introduce the minimization problem for~\eqref{eq:minimisation} and give conditions under which it is well-posed, before studying the convergence to the unregularized problem and highlighting connections with maximum entropy methods, a popular technique to solve low-dimensional Fredholm integral equations.

\subsection{Stability of $\Fun_\alpha$}

We study the properties of the functional
\begin{equation*}
  \Fun_\alpha(\pi) = \KL{\pi}{\varphi+\lambda \int_{\rset^d} \kker(\cdot, y) \pi( y) \rmd y} + \alpha \KL{\pi}{\pi_0},
\end{equation*}
for 
$\alpha \geq 0$, and establish that it admits a unique minimizer.

In addition, we consider the following assumption on the kernel
$\kker$ which guarantees finiteness of the corresponding term in $\Fun_\alpha$.
\begin{assumption}
\label{assum:general_kker}
The density of the kernel $\kker$ and the forcing $\varphi$ are such that $\kker\in \rmc^\infty(\rset^d \times \rset^d, [0, +\infty))$, $\varphi \in \rmc^\infty(\rset^d , [0, +\infty))$  and there exists $\Mtt \geq 0$ such that for any $(x ,y)\in\rset^d\times\rset^d$ we have $\kker(x,y) + \norm{\nabla \kker(x,y)} + \norm{\nabla^2 \kker(x,y)}+\norm{(\partial^3_{ij\ell} \kker(x,y))_{i,j,\ell}} \leq \Mtt$, where $\partial^3_{ij\ell}$ denotes the third derivative of $\kker$ w.r.t. component $i, j, \ell\in\{1, 2\}$,  and $\varphi(x) + \norm{\nabla \varphi(x)} + \norm{\nabla^2 \varphi(x)} \leq \Mtt$.
\end{assumption}

Assumption~\ref{assum:general_kker} restricts the class of kernels $\kker$ to smooth kernels with bounded derivatives up to the third order.
This might seem restrictive, however, smooth kernels are particularly challenging while discountinuous or degenerate kernels normally allow for specific solution methodologies (see, e.g., \cite[Ch. 11]{kress2014linear}).
In this section, we will require only that $\kker \in \rmc(\rset^d \times \rset^d, [0, +\infty))$. The control of higher order derivatives is used to establish convergence result for the interacting particle system (in both continuous and discrete time) .

To address stability issues we also
consider the following regularized functional:%
\begin{equation}
\label{eq:G_eta}
\forall \alpha, \eta \geq 0, \pi \in \Pens(\rset^d): \qquad \Fun_\alpha^\eta(\pi) = \Fun^\eta(\pi) +\alpha \KL{\pi}{\pi_0},
\end{equation}%
where $\eta \geq 0$ is a hyperparameter and $\Fun^\eta(\pi):=\KL{\pi}{\varphi+\lambda \int_{\rset^d} \kker(\cdot, y) \pi( y) \rmd y+\eta}$.
In the case $\eta = 0$, $\Fun_\alpha^\eta(\pi)$ coincides with $\Fun_\alpha(\pi)$.
However, establishing the uniqueness of the minimizer of $\Fun_\alpha$ is difficult.  In what follows we study the function $\Fun_\alpha^\eta$ for
$\alpha, \eta > 0$ and show that $\Fun^\eta_\alpha$, restricted to $\Pens^{ac}(\rset^d)$, is coercive in this case. 
Hence, $\Fun_{\alpha}^\eta$ admits a unique minimizer $\pi_{\alpha,\eta}^\star$.

\begin{prop}
\label{prop:convergence_minimum}
Under Assumption~\ref{assum:general_kker}, we have the following
\begin{enumerate}[label=(\alph*)]
  \item \label{item:a} For any $\eta \geq 0$, $\Fun^\eta:\Pens^{ac}(\rset^d)\to \rset$ is lower bounded and convex.
\item \label{item:b}
 For any $\alpha, \eta > 0$, $\Fun_{\alpha}^\eta:\Pens^{ac}(\rset^d)\to \rset$ is proper, strictly convex, coercive and lower semi-continuous. In particular, $\Fun_\alpha^\eta$ admits a unique minimizer $\pi_{\alpha, \eta}^\star \in \Pens^{ac}(\rset^d)$.
\end{enumerate}
\end{prop}
\begin{proof}
See Appendix~\ref{app:functional}.
\end{proof}

The proof of Proposition \ref{prop:convergence_minimum} shows that the condition $\alpha>0$ is sufficient to guarantee uniqueness of the minimizer, in fact we rely on the coercivity of $\KL{\pi}{\pi_0}$ to show that $\Fun_\alpha^\eta$ is coercive. For the remainder of this work we only consider the case $\alpha, \eta>0$.

\subsection{Convergence of Minimizers}

We now consider the limiting behaviour of the regularized functional $\Fun_{\alpha}^\eta$ when the parameters $\alpha, \eta$ converge to 0.
In particular, we show that if both parameters tend to 0 the minimizer of $\Fun_{\alpha}^\eta$ converges to one of the minimizers of the unregularized functional under some regularity conditions.

To control the behaviour of the family of minima
$\inf_{\pi \in \Pens^{ac}(\rset^d)}\{\Fun_{\alpha}^\eta(\pi):\ \alpha, \eta > 0\}$
when $\alpha$ is close to $0$, we make use of the following assumption which
controls the tail behaviour of $\pi_0$ and the tail behaviour of the density $\kker$.

\begin{assumption}
  \label{assum:pi0_second}The following hold:
  \begin{enumerate}[wide, labelindent=0pt, label=(\alph*)]
  \item $\pi_0(x)\propto \exp\left[-U(x)\right]$, where
  $U: \ \rset^d \to \rset$ is such that there exist $\tau, C_1 > 0$ satisfying for any $x \in \rset^d$
  \begin{equation*}
     -C_1 - \tau \norm{x}^2 \leq U(x) \leq C_1 + \tau \norm{x}^2.
    \end{equation*}
  \item there exists $C_2 \geq 0$ such that for any
    $x \in \rset^d$ and $y \in \rset^d$
    \begin{equation*}
      \kker(x,y) \geq C_2^{-1} \exp[-C_2 (1 + \norm{x}^2+\norm{y}^2)]. 
    \end{equation*}
  \end{enumerate}
  \end{assumption}

Under Assumptions \ref{assum:general_kker} and \ref{assum:pi0_second} we can show that any minimizer of $\Fun_{\alpha}^\eta$ for $\alpha>0, \ \eta\geq 0$ has finite second moment (Lemma~\ref{lem:min2}), i.e. $\pi_{\alpha, \eta}^\star \in \Pens_2^{ac}(\rset^d)$.
\\
Let $\mathcal{Q}_2^{ac}(\rset^d)$ denote the set of probability distributions with finite second moment, which admit an essentially bounded density w.r.t. the Lebesgue measure, i.e. $\pi \in \Pens_2^{ac}(\rset^d)$ and $\mathrm{ess}\sup_{x\in \rset^d}\pi(x)<m$ for some $m>0$.
Under the assumption that $\pi_{\alpha, \eta}^\star \in \mathcal{Q}_2^{ac}(\rset^d)$, we have the following results which describe the behaviour of the minimizer as $\eta\to 0$ and as $(\alpha, \eta)\to (0, 0)$ jointly.

\begin{prop}
  \label{prop:alphaeta0}
  
  Assume \ref{assum:general_kker} and \ref{assum:pi0_second} and that $\pi_{\alpha,
    \eta}^{\star}\in \mathcal{Q}_2^{ac}(\rset^d)$ for all $\alpha>0,\ \eta\geq 0$.
  \begin{enumerate}
  \item  If there exist 
  $\pi_\alpha^\star \in \mathcal{Q}_2^{ac}(\rset^d)$,
  $(\eta_n)_{n \in \nset^\star} \in [0, +\infty)^{\nset^\star}$ such that
  $\lim_{n \to +\infty} \eta_n = 0$ and
  $\lim_{n \to +\infty} \wassersteinD[2](\pi_{\alpha,
    \eta_n}^{\star},\pi_\alpha^\star) = 0$, then
  \begin{equation*}
      \pi_\alpha^{\star} \in \arg\min_{\mathcal{Q}_2^{ac}(\rset^d)} \Fun_\alpha(\pi).
  \end{equation*}
    \item If there exist 
  $\pi^\star \in \mathcal{Q}_2^{ac}(\rset^d)$, 
  $(\eta_n)_{n \in \nset^\star} \in [0, +\infty)^{\nset^\star}, (\alpha_n)_{n \in \nset^\star} \in (0, +\infty)^{\nset}$ such that
  $\lim_{n \to +\infty} \eta_n = 0, \lim_{n \to +\infty} \alpha_n = 0$ and
  $\lim_{n \to +\infty} \wassersteinD[2](\pi_{\alpha,\eta_n}^{\star},\pi^\star) = 0$, then
  \begin{equation*}
      \pi^{\star} \in \arg\min_{\mathcal{Q}_2^{ac}(\rset^d)} \Fun(\pi).
  \end{equation*}
  \end{enumerate}
\end{prop}
\begin{proof}
    See Appendix~\ref{app:gamma}.
\end{proof}

\subsection{Maximum Entropy Methods}

Minimizing the functional $\Fun_\alpha$ is related to finding the maximum entropy solution to~\eqref{eq:fe} \cite{jaynes1957information}. In fact, $\Fun_\alpha$ can be seen as the Lagrangian associated with the following primal problem
\begin{equation}
  \arg\min \left\lbrace\KL{\pi}{\pi_0}:\pi \in \Pens(\rset^d), \ \KL{\pi}{\varphi+\lambda \int \kker(\cdot, y)\pi(y)\rmd y} = 0\right\rbrace.
\end{equation}
To highlight the connection with maximum entropy methods, we observe that if we replace the $\textrm{KL}$ penalty with an entropic penalty, we obtain  a functional
\begin{equation}
\label{eq:funct_entropy}
\tilde{\Fun}_{\alpha}(\pi) =  \KL{\pi}{\varphi+\lambda \int_{\rset^d} \kker(\cdot, y) \pi(y)\rmd y} - \alpha \Hent(\pi),
\end{equation}
corresponding to the Lagrangian associated with the following primal problem
\begin{equation}
  \arg\max \left\lbrace\Hent(\pi):\pi \in \Pens_{\mathrm{H}}(\rset^d), \ \KL{\pi}{\varphi+\lambda \int \kker(\cdot, y)\pi(y) \rmd y} = 0\right\rbrace,
\end{equation}
where $\Pens_{\mathrm{H}}(\rset^d)$ is the set of probability distributions with
finite entropy. 
However, $\tilde{\Fun}_{\alpha}$ is not lower bounded and the corresponding minimization problem is not well-defined and therefore is not considered here. 

Maximum entropy solutions to Fredholm integral equations were originally proposed in \cite{mead1986approximate} under moment constraints obtained by selecting different families of basis functions \cite{jin2016solving, islam2020approximating}.
In the one dimensional case, such a solution can be written analytically (see, e.g., \cite[Prop. 3.1]{islam2020approximating}), and these approaches are particularly effective.
However, in higher dimensions maximum entropy methods generally require discretization of the support of the solution $\pi$.

\section{A Wasserstein Gradient Flow for $\Fun_\alpha^\eta$ and an associated McKean--Vlasov SDE}
\label{sec:wgf}
Having established that the regularized functional \eqref{eq:G_eta} admits a unique minimizer, we exploit the connection between minimization of functionals in the space of probability
measures and partial differential equations (PDEs) identified in \cite{jordan1998variational, otto2001geometry} to obtain an SDE whose invariant measure can be related to the minimizer of
$\Fun_\alpha^\eta$.

We start by deriving the Wasserstein gradient flow of $\Fun_\alpha^\eta$ and, after observing that the corresponding PDE is a Fokker--Plank equation, use standard stochastic calculus to obtain a McKean--Vlasov SDE (MKVSDE) whose law satisfies the Fokker--Plank PDE.

\subsection{Wasserstein Gradient Flow}

To derive the Wasserstein gradient flow PDE for $\Fun_\alpha^\eta$ we first derive the Wasserstein subdifferential of the functional \cite[Def.
10.1.1]{ambrosio2008gradient}, to guarantee the existence of this operator we further assume that $\pi_0$ satisfies the following assumption.
\begin{assumption}
  \label{assum:pi0} The following hold:
  \begin{enumerate}[label=(\alph*)]
  \item \label{item:differentiable_pi0} $\pi_0$ admits a density w.r.t. the Lebesgue measure, $\rmd \pi_0(x) = \pi_0(x)\rmd x$, with $\pi_0(x)\propto \exp\left[-U(x)\right]$, where 
  $U: \ \rset^d \to \rset$ belongs to $\rmc^{\infty}(\rset^d , \rset)$.
  \item \label{item:lip_pi0}    There exists $\Ltt \geq 0$ such that
    $\norm{\nabla U(x_1) - \nabla U(x_2)} \leq \Ltt \norm{x_1 - x_2}$, for any $x_1, x_2 \in \rset^d$.
    \item \label{item:dissipativity}
 There exist $\mtt,\ctt > 0$ such that for any
    $x_1, x_2 \in \rset^d$,
    \[ \langle \nabla U(x_1) - \nabla U(x_2), x_1 - x_2 \rangle \geq \mtt
    \norm{x_1 - x_2}^2 - \ctt. \]
  \end{enumerate}
  \end{assumption}

In Appendix~\ref{app:subdifferential} we show that under Assumption~\ref{assum:general_kker} and~\ref{assum:pi0}--\ref{item:differentiable_pi0} the subdifferential of $\Fun_\alpha^\eta$ is:
\begin{align*}
\partial_{\rms}\Fun_\alpha^\eta(\pi)  &=  \left\lbrace x\to -\int\left[\frac{\lambda\nabla_2 \kker(z, x)}{\lambda\pi\left[ \kker(z, \cdot)\right]+\varphi(z)+\eta}+\frac{\lambda\nabla_1 \kker(x, z)+\nabla \varphi(x)}{\lambda\pi\left[ \kker(x, \cdot)\right]+\varphi(x)+\eta}\right]\rmd \pi\left(z\right)\right.\\
&\left.\qquad\qquad+\nabla \log\pi(x)+\alpha\nabla \log(\rmd \pi / \rmd \pi_0)(x)\right\rbrace,
\end{align*}
where we defined $\pi[\kker(z, \cdot)]:=\int_{\rset^d}\kker(z, y)\pi(y)\rmd y$.

The Wasserstein gradient flow \cite[Def.
11.1.1]{ambrosio2008gradient} associated
with $\Fun_\alpha^\eta$ is a family of
probability measures $(\pi_t)_{t \geq 0}$ satisfying (in a weak sense) $\partial_{t}\pi_{t} = \nabla\cdot(\pi_t\ \partial_{\rms}\Fun_\alpha^\eta(\pi))$, which in our case corresponds to:

\resizebox{0.95\textwidth}{!}{\noindent
\hspace*{-0.5cm}\parbox{\textwidth}{\noindent
  \begin{align}
 &\partial_{t}\pi_{t} = \label{eq:PDE_BM}\\
&-\nabla\cdot\left(\pi_{t}\left\lbrace\int\left[\frac{\lambda\nabla_2 \kker(z, x)}{\lambda\pi_t\left[ \kker(z, \cdot)\right]+\varphi(z)+\eta}+\frac{\lambda\nabla_1 \kker(x, z)+\nabla \varphi(x)}{\lambda\pi_t\left[ \kker(x, \cdot)\right] +\varphi(x)+\eta}\right] \rmd \pi_t\left(z\right) -\alpha\nabla U(x)\right\rbrace\right)  \notag\\
& \qquad \qquad +(1+\alpha)\triangle\pi_{t}. \notag
\end{align}
}}

For strongly geodesically convex (i.e. convex along geodesics) functionals the Wasserstein
gradient flow converges geometrically towards the unique
minimizer.  In our setting $\Fun_{\alpha}^\eta$ is not geodesically
convex but only convex; in the following section we will show that $\pi_t$ in \eqref{eq:PDE_BM} converges to a limiting measure, $\pi_{\alpha, \eta}^\star$.

\subsection{McKean--Vlasov SDE}
Since~\eqref{eq:PDE_BM} is a
Fokker--Plank equation, we informally derive the corresponding SDE, see, e.g. \cite{barbu2020nonlinear} for the precise derivation,
\begin{align}
\label{eq:mckean_sde}
\rmd \bm{X}_t = \left\lbrace\int b^\eta(\bm{X}_t, z, \pi_t) \rmd \pi_t(z)-\alpha\nabla  U(\bm{X}_t)\right\rbrace \rmd t+\sqrt{2(1+\alpha)}\rmd \bm{B}_t
\end{align}
where $(\bm{B}_t)_{t \geq 0}$ is a $d$-dimensional Brownian motion, $\pi_t$ is
the law of $(\bm{X}_t)_{t \geq 0}$ and for any $\nu \in \Pens(\rset^d)$ and $(x, z) \in \rset^d\times\rset^d$
\begin{equation}
  \label{eq:drift}
  b^\eta(x, z, \nu) = \frac{\lambda\nabla_2 \kker(z, x)}{\lambda\nu\left[ \kker(z, \cdot)\right]+\varphi(z)+\eta}+\frac{\lambda\nabla_1 \kker(x, z)+\nabla \varphi(x)}{\lambda\nu\left[ \kker(x, \cdot)\right]+\varphi(x)+\eta}.
\end{equation}
The presence of $\eta>0$ in the denominator guarantees that $b^\eta$ is always well defined.

Comparing~\eqref{eq:drift} with the drift of the MKVSDE derived in \cite[Eq. (13)]{crucinio2022solving} for equations of the first kind we find that the dependence of the dynamics~\eqref{eq:mckean_sde} on $\pi$ is considerably more complicated. In particular, the kernel $\kker$ needs to be integrated against $\pi$ twice (one w.r.t. the first component and one w.r.t. the second) and $b^\eta$ in~\eqref{eq:drift} is then integrated w.r.t. to obtain the drift in~\eqref{eq:mckean_sde}. This is not the case for the MKVSDE derived in \cite{crucinio2022solving} in which $\pi$ only appears once in the definition of the drift. The presence of nested integrals w.r.t $\pi$ complicates both the analysis (e.g. establishing the Lipschitz continuity of the drift needed for the proof of Proposition~\ref{prop:existence_uniqueness} below requires extra care) and the numerical approximation: na\"{i}ve approximations of~\eqref{eq:drift} obtained by replacing $\pi$ with an empirical measure $\pi^N$ can lead to a cubic cost in the number of particles as discussed in Section~\ref{sec:tuning}.

The SDE \eqref{eq:mckean_sde} belongs to the class of McKean--Vlasov SDEs as the drift coefficient depends upon not only $\bm{X}_t$ but also its law, $\pi_t$.
Under Assumption~\ref{assum:general_kker} and~\ref{assum:pi0}, if $\eta>0$ the drift of~\eqref{eq:mckean_sde} is
Lipschitz continuous and we can establish existence and uniqueness for the MKVSDE.
\begin{prop}
\label{prop:existence_uniqueness}
Under Assumptions~\ref{assum:general_kker}, \ref{assum:pi0}--\ref{item:differentiable_pi0} and~\ref{assum:pi0}--\ref{item:lip_pi0}, for any $\alpha, \eta > 0$ there exists a unique strong
solution to~\eqref{eq:mckean_sde} for any initial condition $\bm{X}_0$ such that
$\mathcal{L}(\bm{X}_0) \in \mathcal{P}^{ac}_1(\rset^d)$.
\end{prop}
\begin{proof}
See Appendix~\ref{app:proof_eu}.
\end{proof}

The previous proposition is limited to the case where $\eta > 0$.
Indeed, if $\eta=0$ the drift is \emph{not} Lipschitz continuous and the SDE~\eqref{eq:mckean_sde} might be unstable, with solutions existing up to a (possibly small) explosion time.
As highlighted before, showing that the drift of~\eqref{eq:mckean_sde} is Lipschitz continuous requires extra care compared to \cite[Proposition 5]{crucinio2022solving}: not only one needs to show that $b^\eta$ is Lipschitz continuous in $x$ as for \cite[Proposition 5]{crucinio2022solving}, but also that $b^\eta$ is Lipschitz continuous in $z$ due to the presence of the integral w.r.t. $\pi$ in~\eqref{eq:mckean_sde}.

Using recent results from \cite{hu2019mean} we can establish that $\pi_t$ converges to the unique minimizer $\pi_{\alpha, \eta}^\star$ of $\Fun_{\alpha}^\eta$ when $t \to +\infty$.
The presence of an interaction term $b^\eta$ which is not small w.r.t. $\alpha U$ prevents us from using more standard approaches (e.g. \cite{butkovsky2014ergodic,eberle2016reflection,malrieu2001logarithmic,bogachev2019convergence}) and thus from obtaining quantitative convergence rates.

\begin{prop}
  \label{prop:la_convergence_star}
Under Assumptions~\ref{assum:general_kker} and~\ref{assum:pi0}, for any $\alpha, \eta > 0$, we have
  $$\underset{t \to +\infty}{\lim} \wassersteinD[2](\pi_t, \pi_{\alpha,
    \eta}^\star) = 0.$$
\end{prop}
\begin{proof}
See Appendix~\ref{app:proof_invariant}.
\end{proof}

The previous result deals with the behaviour of $\pi_t$ for fixed $\alpha, \eta$. Under our assumptions, \cite[Thm. 4.1]{chizat2022mean} and \cite[Thm. 1]{nitanda2022convex} guarantee that $\Fun^\eta(\pi_t)- \inf \Fun^\eta \to 0$ provided that the parameter $\alpha$ decays at rate $1/\log t$.

\section{Particle System}
\label{sec:ips}
As discussed in the previous section, the drift of the MKVSDE~\eqref{eq:mckean_sde} cannot be computed analytically, since it depends on the law $\pi_t$ of $\bm{X}_t$ at each time step.
Thus, it is not possible to directly approximate~\eqref{eq:mckean_sde} with a time-discretized process. 

A classical approach to circumvent this issue is to consider an interacting particle system $(\bm{X}_t^{1:N})_{t \geq 0} = \{(\bm{X}_t^{i,N})_{t \geq 0}\}_{i=1}^N$ which approximates~\eqref{eq:mckean_sde} for any
$N \in \nset^\star$ and satisfies a classical SDE \cite{mckean1966class, bossy1997stochastic}.
We introduce the particle system $(\bm{X}_t^{1:N})_{t \geq 0}$ which satisfies the
following SDE, for any $i \in \{1, \dots, N\}$, $\bm{X}_0^{i,N} \in \rset^d$:
\begin{equation}
\label{eq:particle}
\rmd \bm{X}_{t}^{i,N}=\left\lbrace\int b^\eta(\bm{X}_t^{i,N}, z, \pi_t^N)\rmd\pi_t^N\left( z\right)-\alpha\nabla U(\bm{X}_t^{i,N})\right\rbrace \rmd t+\sqrt{2(1+\alpha)}\rmd \bm{B}_{t}^{i},
\end{equation}
where $\{(\bm{B}_t^i)_{t \geq 0}\}_{i \in \nset}$ is a family of independent
Brownian motions and $\pi_{t}^{N}$ is the empirical measure
associated with $(\bm{X}_t^{1:N})_{t \geq 0}$, $\pi_{t}^{N}=\frac{1}{N}\sum_{i=1}^{N}\delta_{\bm{X}_{t}^{i,N}}$ for every $t \geq 0$. 

In the following proposition we establish a propagation of chaos result, 
showing that there exists a unique strong solution to~\eqref{eq:particle} and that this
solution approximates~\eqref{eq:mckean_sde} for any finite time horizon. The proof is classical, see, e.g. \cite[Prop. 7]{crucinio2022solving}, and is given in the supplement for completeness (see~\ref{app:poc}).
\begin{prop}
\label{prop:propagation_chaos}
Under Assumptions~\ref{assum:general_kker}, \ref{assum:pi0}--\ref{item:differentiable_pi0} and~\ref{assum:pi0}--\ref{item:lip_pi0} for any $\alpha, \eta > 0$ and $N \in \nset^\star$
there exists a unique strong solution to~\eqref{eq:particle} for any initial
condition $\bm{X}_0^{1:N}$ such that $\mathcal{L}(\bm{X}_0^{1:N}) \in \Pens_1^{ac}((\rset^d)^N)$ and
$\{\bm{X}_0^{i,N}\}_{i=1}^N$ is exchangeable. 
In addition, for any $T \geq 0$ there exists $c_1(T) \geq 0$ such that for any $N \in \nsets$ and $\ell \in \{1, \dots, N\}$ if
$\bm{X}_0^{\ell,N} = \bm{X}_0$ almost surely, then:
\begin{equation*}
\Exp\left[\sup_{t \in [0,T]} \norm{\bm{X}_t - \bm{X}_t^{\ell,N}}^2\right]^{1/2} \leq \frac{c_1(T)}{\sqrt{N}}
\end{equation*}
\end{prop}

To obtain stronger convergence results for our particle system we use the dissipativity condition in Assumption~\ref{assum:pi0}--\ref{item:dissipativity} which allows us to obtain the exponential ergodicity of the particle system using standard tools.
The bounds that we obtain, however, are not uniform w.r.t to the number of particles and, in particular, they depend on a constant $C_N\to \infty$ and a rate $\rho_N\to 1$ as $N \to +\infty$. The proof (see~\ref{app:ergodic}) follows the argument of \cite[Prop. 8]{crucinio2022solving}.
\begin{prop}
\label{prop:particle_ergodic}
Under Assumptions~\ref{assum:general_kker} and \ref{assum:pi0}, for any $\alpha, \eta > 0$ and $N \in \nsets$ there
exist $C_N \geq 0$ and $\rho_N \in [0,1)$ such that for any
$x_1^{1:N}, x_2^{1:N} \in (\rset^d)^N$ and $t \geq 0$
  \begin{equation*}
    \wassersteinD[1](\pi_t^N(x_1^{1:N}), \pi_t^N(x_2^{1:N})) \leq C_N \rho_N^t \norm{x_1^{1:N} - x_2^{1:N}},
  \end{equation*}
  where for any $x^{1:N} \in (\rset^d)^N$, $\pi_t^N(x^{1:N})$ is the
  distribution of $(\bm{X}_t^{1:N})_{t \geq 0}$ with initial condition $x^{1:N}$.
  In particular,~\eqref{eq:particle} admits a unique invariant probability
  measure, $\pi^N \in \Pens_1((\rset^d)^N)$, and for any
  $x^{1:N} \in (\rset^d)^N$ and $t \geq 0$
  \begin{equation*}
      \wassersteinD[1](\pi_t^N(x^{1:N}), \pi^N) \leq C_N \rho_N^t \left(\norm{x^{1:N}} + \int_{\rset^d} \norm{\tilde{x}} \rmd \pi^N(\tilde{x})\right).
  \end{equation*}
\end{prop}

Proposition \ref{prop:propagation_chaos} and \ref{prop:particle_ergodic} both show that the interacting particle system \eqref{eq:particle} is well behaved: Proposition \ref{prop:propagation_chaos} shows that for any finite time horizon $T$ \eqref{eq:particle} well approximates \eqref{eq:mckean_sde}, whereas, Proposition  \ref{prop:particle_ergodic} guarantees that  \eqref{eq:particle} has a unique invariant measure for each finite $N$.
The next result guarantees that the invariant measure $\pi^N$ of  \eqref{eq:particle} converges to $\pi_{\alpha, \eta}^\star$. The proof follows the same argument of \cite[Prop. 9]{crucinio2022solving}, and is given in~\ref{app:invariant} for completness.

\begin{prop}
\label{prop:particles_min}
Under Assumptions~\ref{assum:general_kker} and \ref{assum:pi0}, for any $\alpha, \eta > 0$,
$$\lim_{N \to +\infty} \wassersteinD[1](\pi^{N}, \pi_{\alpha, \eta}^\star) =
0.$$ In addition, $\pi_{\alpha, \eta}^\star$ is the unique invariant probability
measure of~\eqref{eq:mckean_sde}.
\end{prop}

\subsection{Euler--Maruyama Discretization}

Having obtained a particle system \eqref{eq:particle} whose invariant distribution converges to the minimizer $\pi_{\alpha, \eta}^\star$ of \eqref{eq:minimisation}, we now consider the numerical approximation of  \eqref{eq:particle} using an Euler--Maruyama time discretization scheme. For any
$N \in \nset$, we consider the following time discretization
given by $X_0^{1:N} \in (\rset^d)^N$ and
\begin{equation}
  \label{eq:euler}
  \forall n \in \nset, k \in \{1,\ldots,N\}:\qquad
  X_{n+1}^{k,N}  = X_n^{k,N} + \gamma b( X_n^{k,N}, \pi^{N}_n)+ \sqrt{2 \gamma(\alpha+1)} Z_{n+1}^k,
\end{equation}
where $\{Z_n^k\}_{k, n \in \nset}$ is a family of independent standard Gaussian random
variables, $\gamma > 0$ is a stepsize, $b$ is
\begin{equation}
\label{eq:b}
b(x, \pi) := \int_{\rset^d} b^\eta(x, z, \pi) \rmd \pi(z)-\alpha \nabla U(x),
\end{equation}
and
for any $n \in \nset$, we have that
$\pi^{N}_n = (1/N) \sum_{k=1}^N \delta_{X_n^{k,N}}$. 

Since under our assumptions the drift $b$ is Lipschitz continuous in both arguments (see Appendix \ref{app:proof_eu} for a proof), \cite{bossy1997stochastic} guarantees that the Euler--Maruyama scheme \eqref{eq:euler} has strong order of convergence $1/2$: for
any $N \in \nsets$, $\ell \in \{1, \dots, N\}$, $\gamma >0$ and $T \geq 0$, there exists $c_2(T) \geq 0$, independent of $\gamma$, such that for $n_T = \lfloor T/\gamma\rfloor$:
  \begin{equation}
  \label{eq:euler_error}
    \Exp\left[\sup_{n \in \{0, \dots, n_T\}} \norm{\bm{X}_{n \gamma}^{\ell,N}- X_{n}^{\ell,N}}\right] \leq c_2(T) \sqrt{\gamma}.
  \end{equation}
For simplicity, we use an Euler--Maruyama time discretization instead of the tamed one proposed in \cite{crucinio2022solving}; the convergence rate in~\eqref{eq:euler_error} follows from \cite{bossy1997stochastic}, to obtain a stronger order of convergence one would need additional control on the drift of~\eqref{eq:particle} and its Lions derivatives as shown in \cite[Appendix C]{crucinio2022solving}.

For small values of
$\gamma >0$ and large values of $n\gamma$ and $N\in \nset$ we get that
$\pi_n^N=(1/N) \sum_{k=1}^N \delta_{X_n^{k,N}}$ is an approximation of
$\pi_{\alpha, \eta}^\star$. 
To obtain a smooth approximation, we can plug in this particle approximation into~\eqref{eq:fe} and define $\hat{\pi}_{n}^{N}: \ \rset^d \to \rset$ via:
\begin{align}
  \label{eq:kde}
\forall x \in \rset^d: \qquad \hat{\pi}_n^{N}(x) = \varphi(x) + \frac{\lambda}{N}\sum_{k=1}^N \kker(x, X_{n}^{k, N}),
\end{align}
as suggested in \cite{doucet2010solving} as a method for obtaining smooth approximations of the solution of a Fredholm equation of the second kind from a particle approximation. 

Our final algorithm is summarized in Algorithm \ref{alg:second_kind}. For simplicity we assume initialization is carried out by simple random sampling from some initial distribution but other possibilities, including stratified sampling and using quasi Monte Carlo pointsets might further improve performance.
\begin{algorithm}
\caption{FE2kind-WGF: a method for solving Fredholm integral equations of the second kind with Wasserstein gradient flows.}\label{alg:second_kind}
\begin{algorithmic}
\STATE{\textbf{Require:} $N,n_T \in \nset$, $\alpha, \eta, \gamma > 0$, $\mu, \pi_0, \pi_{\mathrm{init}} \in \Pens^{ac}(\rset^d)$.}
\STATE{Draw $\{X_0^{k,N}\}_{k=1}^N$ from $\pi_{\mathrm{init}}^{\otimes N}$}
\FOR{$n=1:n_T$}
\FOR{$k=1:N$}
\STATE{Update $X_{n}^{k,N}  = X_{n-1}^{k,N} + \gamma b(X_{n-1}^{k,N}, \pi^{N}_{n-1}) + \sqrt{2\gamma (\alpha+1)} Z^k_{n} $ as in~\eqref{eq:euler}}
\ENDFOR
\ENDFOR
\RETURN $ \hat{\pi}_n^{N}(x) $ as in~\eqref{eq:kde}
\end{algorithmic}
\end{algorithm}

\subsection{Implementation Guidelines}
\label{sec:tuning}
\paragraph*{Efficient Computation}
Notice that the drift of \eqref{eq:mckean_sde} involves nested expectations w.r.t. $\pi_t$ contrary to the drift of the MKVSDE obtained for equations of the first kind in \cite{crucinio2022solving}. In particular, we can decompose the drift into
\begin{align*}
    b(x, \pi) &= \int b^\eta(x, z, \pi) \rmd \pi(z)-\alpha\nabla  U(x)\\
    &=\int \frac{\lambda\nabla_2 \kker(z, x)}{\lambda\pi\left[ \kker(z, \cdot)\right]+\varphi(z)+\eta} \rmd \pi(z)+\frac{\lambda\pi[\nabla_1 \kker(x, \cdot)]+\nabla \varphi(x)}{\lambda\pi\left[ \kker(x, \cdot)\right]+\varphi(x)+\eta}-\alpha\nabla  U(x).
\end{align*}
Na\"{i}ve implementation of the evaluation of $b$ can lead to $\mathcal{O}(N^3)$ cost; precomputing $\pi^N_n \left[ \kker(X_n^{j,N}, \cdot)\right]$ and storing its value for each $j$ allows computation at $\mathcal{O}(N^2)$ cost.

Particle systems such as this also lend themselves to natural parallel implementation. On appropriate shared-memory architectures one could parallelise both the pre-computation of $\pi^N_n \left[ \kker(X_n^{j,N}, \cdot)\right]$ and then the propagation of the particles themselves leaving an $\mathcal{O}(N)$ cost for direct parallel implementations with width $N$; further performance improvements could be obtained by using fork-and-join techniques to carry out the computations associated with individual particles. 

\paragraph{Choice of $\pi_0$}

The reference measure $\pi_0$ is the mechanism by which regularization is specified and can be chosen to impose particular properties (e.g. degree of smoothness, localized support) on $\pi$ or to favour solutions with such properties.
While the presence of $\pi_0$ might seem arbitrary, \eqref{eq:fe} is an inverse problem, and thus regularization is required to obtain stable results. This manifests itself in the fact that $\Fun_\alpha^\eta$  admits a unique minimizer only when $\alpha>0$ (Proposition~\ref{prop:convergence_minimum}).

The choice of an improper reference measure $\pi_0\propto C$ with $C>0$ might seem a natural pragmatic choice. However, this results in a MKVSDE whose drift does not depend on $\pi_0$
\begin{align}
\label{eq:mckean_sde_entropy}
\rmd \bm{X}_t = \left\lbrace\int b^\eta(\bm{X}_t, z, \pi_t) \rmd \pi_t(z)\right\rbrace \rmd t+\sqrt{2(1+\alpha)}\rmd \bm{B}_t.
\end{align}
This scheme corresponds to constructing the gradient flow for the functional~\eqref{eq:funct_entropy}, that, as previously discussed, does not necessarily lead to a unique minimizer. In addition, our proof of convergence of $\pi_t$ to the unique minimizer (Proposition~\ref{prop:la_convergence_star}) relies on the fact that $\alpha>0$ with some $\pi_0$ satisfying Assumption~\ref{assum:pi0}.

The first experiment shows that the presence of $\pi_0$ stabilizes the minimizer $\pi_{\alpha, \eta}^\star$.
See also \cite[App. D.1]{crucinio2022solving} for a discussion of the impact of $\pi_0$ in a similar context.

\paragraph*{Tuning and algorithmic set up}

The value of $\alpha$ controls the amount of regularization introduced by the cross-entropy penalty, and, as shown by Proposition~\ref{prop:convergence_minimum}, when $\alpha>0$ the functional $\Fun_\alpha^\eta$ is coercive and thus admits a unique minimizer.
The value of $\alpha$ could be fixed a priori, or, to guarantee a good trade-off between accuracy and regularization, it could be selected by cross-validation as suggested in \cite{crucinio2022solving, wahba1977practical, amato1991maximum} for integral equations of the first kind.

The parameter $\eta$ was introduced in~\eqref{eq:G_eta} to deal with the possible instability of the functional $\Fun_{\alpha}^{\eta}$; we did not find performances to be significantly
influenced by this parameter as long as its value is sufficiently small. In practice, in the experiments in Section~\ref{sec:expe} we set $\eta\equiv 0$ and observe that the resulting algorithm is stable. This was also empirically observed in related contexts by \cite{crucinio2022solving} where a tamed Euler discretization is employed, and by \cite{lim2024particle} where a standard Euler discretization is employed.

The values of the number of particles
$N$ and the time discretization step $\gamma$ control the quality of the numerical approximation of~\eqref{eq:mckean_sde}.
Combining the
result in Proposition \ref{prop:propagation_chaos} with \eqref{eq:euler_error} we obtain the following global error estimate: for any $T \geq 0$ we have
\begin{align}
\label{eq:guideline}
  \Exp\left[\sup_{n \in \{0, \dots, n_T\}} \norm{\bm{X}_{n \gamma}^{\star}- X_n^{\ell,N}}\right] \leq \frac{c_1(T)}{\sqrt{N}} +c_2(T) \sqrt{\gamma}.
\end{align}
This global error estimate differs from the one in \cite[Eq. (19)]{crucinio2022solving}: for equations of the second kind, $\pi$ appears on both sides of~\eqref{eq:fe} and thus no observed data is available (leading to the lack of the $\gamma m^{-1/2}$ term in \cite[Eq. (19)]{crucinio2022solving}); 
the different analysis of the time discretization scheme discussed below~\eqref{eq:euler_error} leads to the different rate w.r.t. $\gamma$.

Choosing $N$ amounts to the classical task of selecting an appropriate sample size for Monte Carlo approximations, while the choice of $\gamma$ corresponds to the specification of a timescale on which to discretize a (stochastic) process; hence, one can exploit the vast literature on Monte Carlo methods and discrete time approximations of SDEs to select these values \cite{kloeden1992stochastic}.

Given a fixed computational budget $B$, as the cost of running Algorithm~\ref{alg:second_kind} is $\mathcal{O}(N^2/\gamma)$, we can minimize the r.h.s. of~\eqref{eq:guideline} subject to $B=N^2/\gamma$ and obtain, via a simple Lagrange multiplier argument, that the optimal values of $N, \gamma$ are $N = B^{1/3} [c_1(T) / (2 c_2(T))]^{2/3}$ and $\gamma=N^2/B=(c_1(T) / 2 c_2(T))^{4/3} B^{-1/3}$.
This suggests that, optimally, $N$ and $1/\gamma$ both scale with the cube root of the computational effort available---suggesting setting $\gamma = O(N^{-1})$, which results in an error decaying at the usual Monte Carlo rate of $O(N^{-1/2})$ with an overall cost scaling as $O(N^3)$.

It is straightforward to choose the number of time steps $n_T$ adaptively
by approximating the value of $\Fun_{\alpha}^\eta$ through numerical
integration by replacing the r.h.s. of~\eqref{eq:fe} with~\eqref{eq:kde} and using standard kernel density estimation \cite{silverman1986density} for the l.h.s..
Once the value of $\Fun_{\alpha}^\eta$ stops decreasing, a
minimizer has apparently been reached and the iteration can be stopped.

\section{Experiments}
\label{sec:expe}
\subsection{Simple Gaussian Example}
\label{sec:rj}
We start by considering a simple analytically-tractable integral equation. Consider, for $\lambda \in (0,1)$, ~\eqref{eq:fe} with $\varphi(x)=(1-\lambda)\N\left(x;0,1\right)$ and $\kker(x, y) =\N\left(y;xe^{-\beta},1-e^{-2\beta}\right)$ for some $\beta>0$, with $\N(x;\mu,\sigma^2)$ denoting the density of a normal distribution of mean $\mu$ and variance $\sigma^2$ evaluated at $x$.
Since $\max_x \int_{\rset^d} \lambda k(x, y)\rmd y = \lambda < 1$, \cite[Cor. 2.16]{kress2014linear} guarantees that~\eqref{eq:fe} admits a unique solution, $\pi(x)=\N\left(x;0,1\right)$ for all $0<\lambda<1$. In our experiments we set $\beta = 0.5$.

\subsubsection{Effect of Reference Measure}
As discussed in Section~\ref{sec:tuning}, the reference measure $\pi_0$ can be used to impose certain properties on the regularized solution $\pi_{\alpha,\eta}^\star$. To better understand the role of the reference measure, we compare the approximate solutions obtained with an improper reference measure (i.e. one corresponding to the entropic penalty in~\eqref{eq:funct_entropy}) and several choices of $\pi_0$: the target, $\pi_0=\pi$, a more diffuse reference measure, $\pi_0=\mathcal{N}(0, 2^2)$, and a more concentrated one, $\pi_0=\mathcal{N}(0, 0.1^2)$.
We set $N=100$ and $\gamma = 10^{-2}$ as suggested in Section \ref{sec:tuning} and iterate for $n_T=200$, which empirically seems sufficient to obtain convergence of the value of $\Fun_\alpha^\eta$ (approximated numerically as described in Section \ref{sec:tuning}). We consider $\alpha\in[0,1]$.

We check the accuracy of the reconstructions through their mean, variance and integrated square error
\begin{align}
  \label{eq:ise}
  \ise(\hat{\pi})=\int_{\rset} \{ \pi(x) - \hat{\pi}(x)\}^2 \rmd x,
\end{align}
with $\hat{\pi}$ an estimator of $\pi$ (the integral is approximated by numerical integration).

\begin{figure}
\centering
\begin{tikzpicture}[every node/.append style={font=\normalsize}]
\node (img1) {\includegraphics[width=0.25\textwidth]{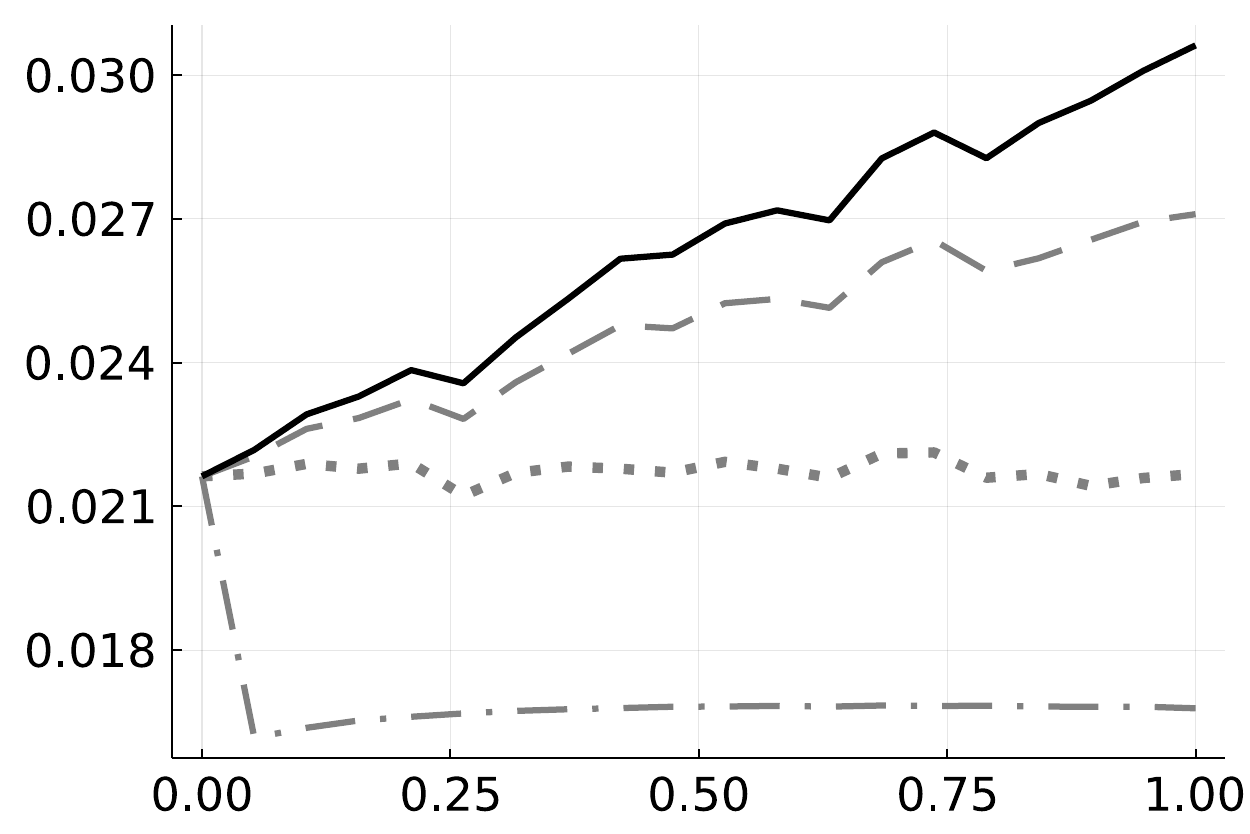}};
\node[below=of img1, node distance = 0, yshift = 1cm] (label1) {$\alpha$};
  \node[left=of img1, node distance = 0, rotate=90, anchor = center, yshift = -0.8cm] {$\ise(\hat{\pi})$};
\node[right=of img1, node distance = 0, xshift = -0.5cm] (img2) {\includegraphics[width=0.25\textwidth]{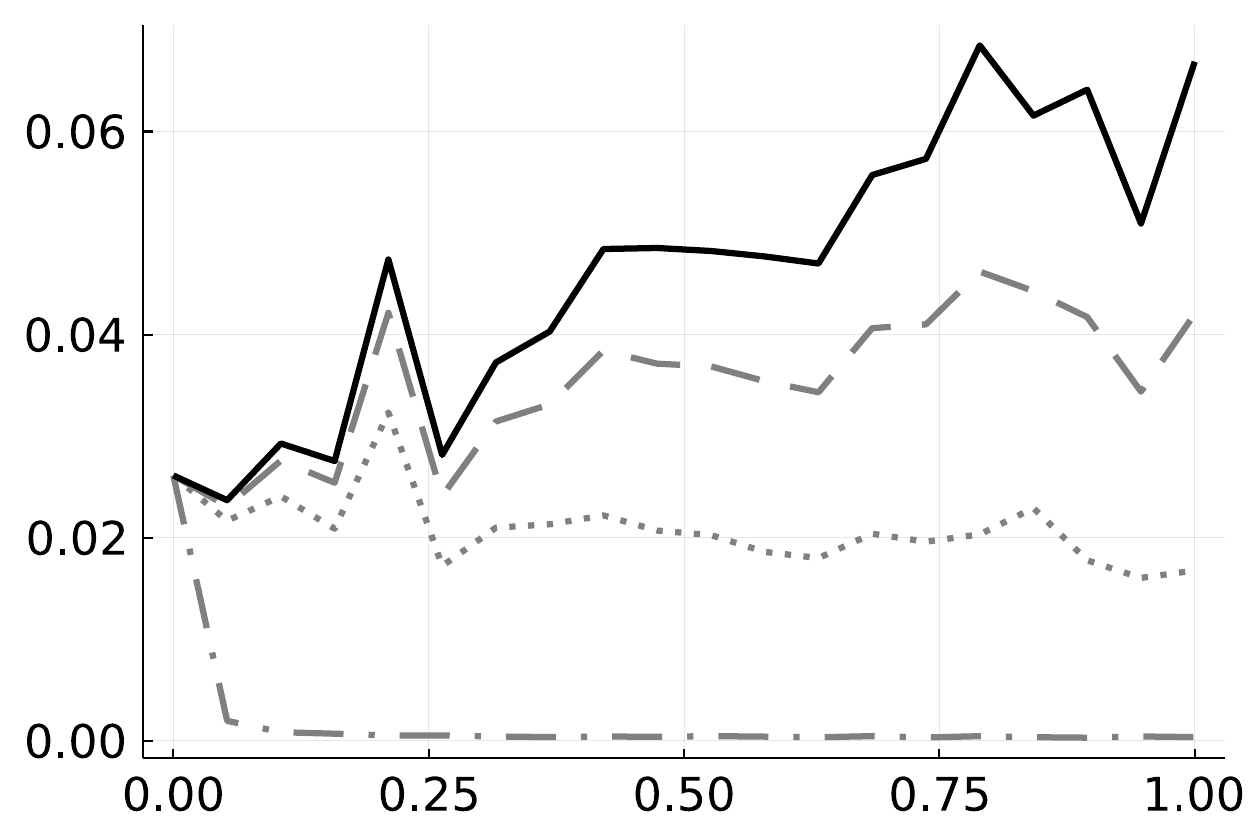}};
  \node[left=of img2, node distance = 0, rotate=90, anchor = center, yshift = -0.8cm] {$\mse[\textrm{mean}]$};
  \node[right=of img2, node distance = 0, xshift = -0.5cm] (img3) {\includegraphics[width=0.25\textwidth]{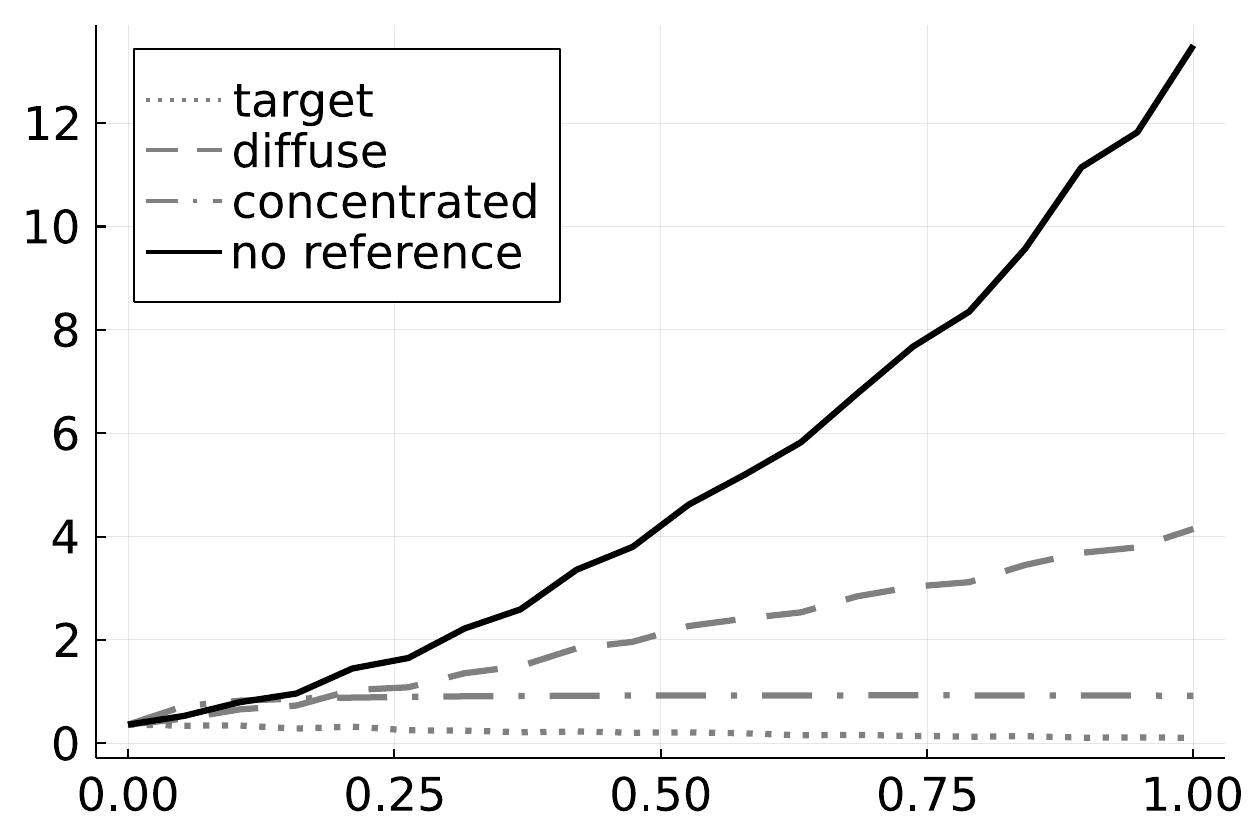}};
\node[left=of img3, node distance = 0, rotate=90, anchor = center, yshift = -0.8cm] {$\mse[\textrm{var}]$};
\node[below=of img2, node distance = 0, yshift = 1cm] (label1) {$\alpha$};
\node[below=of img3, node distance = 0, yshift = 1cm] (label1) {$\alpha$};
  \end{tikzpicture}
\caption{Effect of reference measure on reconstruction accuracy as $\alpha$ increases. We compare $\ise$ and $\mse$ of both mean and variance.}
\label{fig:toy_gaussian_reference}
\end{figure}

Figure~\ref{fig:toy_gaussian_reference} shows the $\ise$ and the mean squared error ($\mse$) for mean and variance of the reconstructions as the value of the regularizing parameter $\alpha$ increases. Large values of $\alpha$ inflate the diffusion coefficient of the MKVSDEs~\eqref{eq:mckean_sde} and~\eqref{eq:mckean_sde_entropy}, but, in the latter, there is no $\nabla U$ term in the drift and thus the diffusion term becomes stronger than the drift term for large $\alpha$s, causing the reconstructed solutions to have larger variance than the solution. A reference measure more diffuse than the target has a similar effect.

It may seem superficially surprising that a concentrated measure seems to give better results than using $\pi_0=\pi$ if one looks at the mean alone, but this simply reflects the fact that  a concentrated measure forces the particles to be close to its mean (in this case 0) and therefore compensates for the effect of the diffusion term in the SDE~\eqref{eq:mckean_sde}. However, a concentrated measure performs worse in terms of variance, especially for large $\alpha$.

\subsubsection{Comparison with Reversible jump MCMC}
The previous experiments show that the presence of the reference measure significantly improves the results provided by FE2kind-WGF.
We now compare our approach with the MCMC algorithm of \cite{doucet2010solving} which provides a Monte Carlo approximation of the von Neumann series representation of the solution $\pi$ (see, e.g., \cite[Sec. 2.4]{kress2014linear}). We implement their reversible jump MCMC (RJ-MCMC) with the following set up: we set the death and birth probabilities to 1/3 and use a random walk proposal with variance $\sigma^2 = 0.1^2$ for the update move and sample from $\pi$ for the birth move.

For FE2kind-WGF we set $\alpha = 0.01$, $N=100$ and $\gamma = 10^{-2}$ and iterate for $n_T=200$ time steps which empirically seems sufficient to obtain convergence of the value of $\Fun_\alpha^\eta$ (approximated numerically as described in Section \ref{sec:tuning}). We compare three reference measures, the target, $\pi_0=\pi$, a more diffuse reference measure, $\pi_0=\mathcal{N}(0, 2^2)$, and a more concentrated one, $\pi_0=\mathcal{N}(0, 0.1^2)$. We do not consider the case in which no reference measure is given further as the experiments in the previous section showed that this approach does not perform well. The initial distribution of the particles is a Gaussian with mean 0 and small variance ($\sigma^2 = 0.1^2$).

We compare the behaviour of the two approaches as $\lambda$ varies in $(0,1)$. 
We set the number of sample paths drawn using RJ-MCMC to $M=2\cdot10^4$ to match the average cost of FE2kind-WGF (roughly 0.25 seconds for both algorithms).
A smooth reconstruction is obtained by using~\eqref{eq:kde} for FE2kind-WGF and using an equivalent strategy for the RJ-MCMC estimator (see \cite[Eq. (44)]{doucet2010solving}).

Figure~\ref{fig:toy_gaussian_comparison} shows the average $\ise$ and the mean squared error of the estimated variance over 100 replicates for $\lambda\in(0, 1)$. As $\lambda\to 1$ recovering the solution becomes more challenging, as the von Neumann series becomes unstable.
However, the results provided by FE2kind-WGF remain stable as $\lambda$ increases for all choices of $\pi_0$, and in fact, can be applied even when $\lambda=1$ (i.e. no forcing term; see Section~\ref{sec:gp}), while RJ-MCMC cannot. Indeed, FE2kind-WGF provides reconstructions as least as good as those of the RJ-MCMC approach for all values of $\lambda$ and has considerably better behaviour as $\lambda$ increases. This is not surprising as the larger the value of $\lambda$, the more terms of the von Neumann expansion that make significant contributions and so the larger the space that the RJ-MCMC algorithm is required to explore in order to give a good approximation overall whereas the FE2kind-WGF approach does not depend upon such an expansion. 

\begin{figure}
\centering
\begin{tikzpicture}[every node/.append style={font=\normalsize}]
\node (img1) {\includegraphics[width=0.4\textwidth]{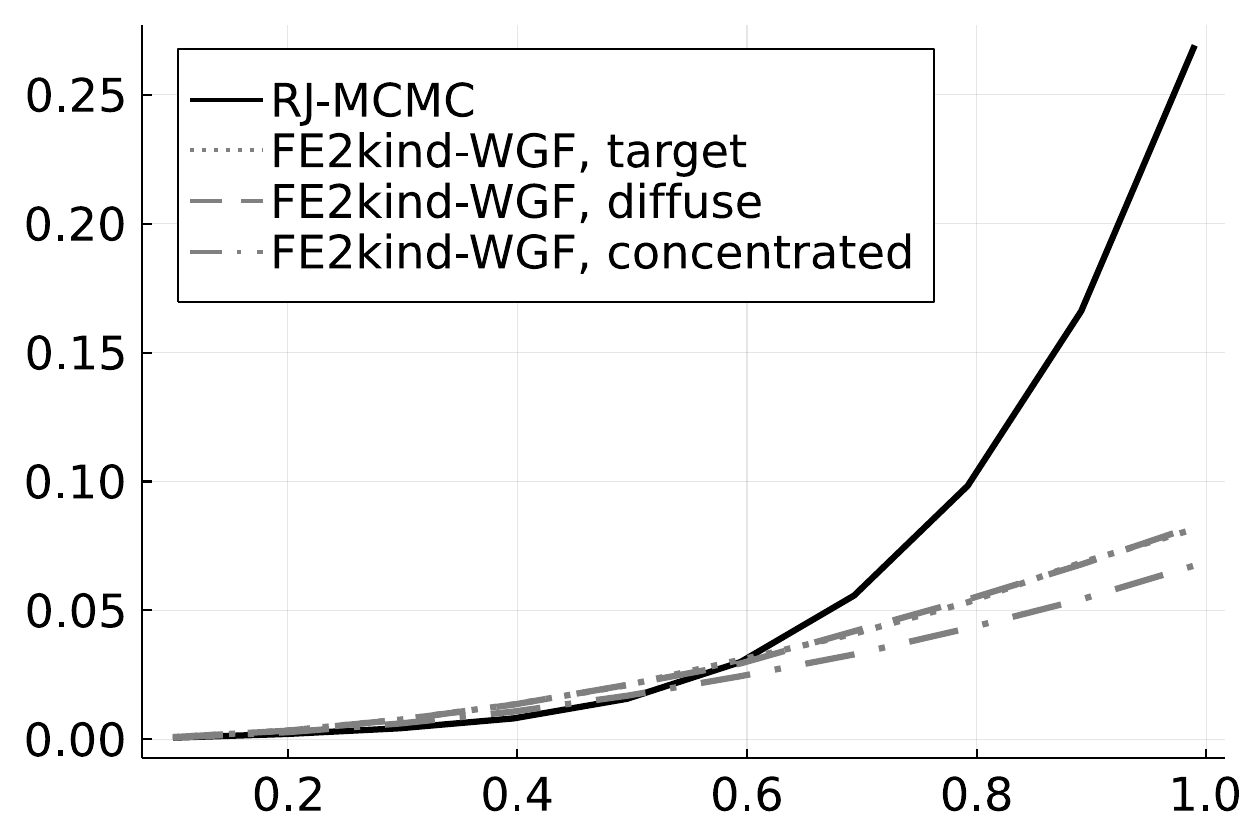}};
\node[below=of img1, node distance = 0, yshift = 1cm] (label1) {$\lambda$};
  \node[left=of img1, node distance = 0, rotate=90, anchor = center, yshift = -0.8cm] {$\ise(\hat{\pi})$};
\node[right=of img1, node distance = 0, xshift = -0.5cm] (img2) {\includegraphics[width=0.4\textwidth]{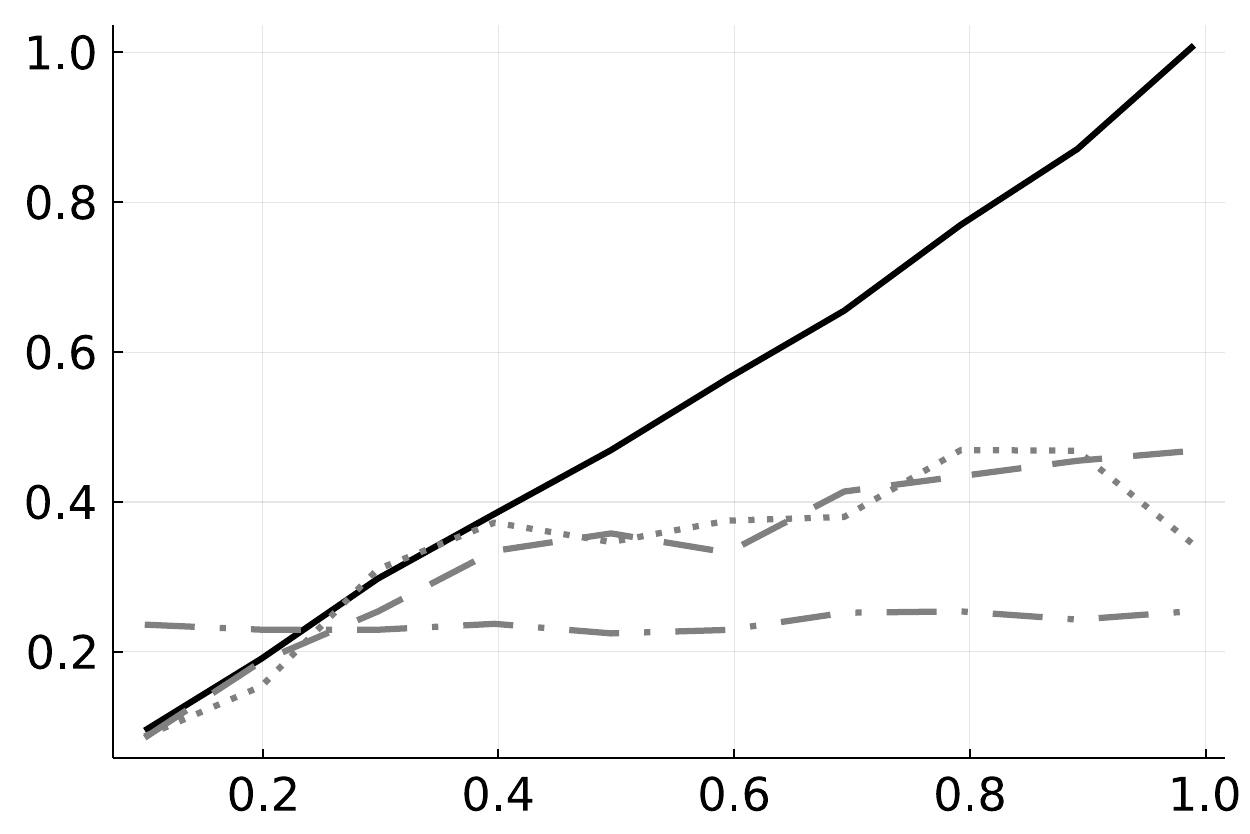}};
\node[below=of img2, node distance = 0, yshift = 1cm] {$\lambda$};
  \node[left=of img2, node distance = 0, rotate=90, anchor = center, yshift = -0.8cm] {$\mse[\textrm{var}]$};
  \end{tikzpicture}
\caption{Accuracy of solutions as $\lambda$ increases on a toy Gaussian model. We compare FE2kind-WGF and RJ-MCMC with similar cost through the $\ise$ and the $\mse$ of the estimated variance.}
\label{fig:toy_gaussian_comparison}
\end{figure}

\subsection{Karhunen--Lo{\`e}ve Expansions}
\label{sec:kl}
The Karhunen--Lo{\`e}ve expansion of a given stochastic process with covariance function $\kker(x, y)$  is given by the pairs of eigenvalues and eigenfunctions $(\lambda, \pi)$ of $\kker$. For any eigenvalue $\lambda$, the corresponding eigenfunction can be found by solving~\eqref{eq:fe} with $\varphi \equiv 0$.
We consider two examples commonly used in the Gaussian process literature \cite[Sec. 4.2]{williams2006gaussian}: the squared exponential kernel, $\kker(x, y) = \exp(-(y-x)^2)$, which satisfies our assumptions, and the exponential kernel, $\kker(x, y) = \exp(-\vert y-x\vert)$, which is not differentiable for $y=x$ and therefore does not satisfy Assumption~\ref{assum:general_kker}.

The Karhunen--Lo{\`e}ve expansion of the exponential kernel is known analytically \cite[Sec. 2.3.3]{ghanem2003stochastic}; the largest eigenvalue is $\lambda = 2/(1+\omega^2)$, where $\omega$ is the largest positive root of $f(\omega)= 1 - \omega\tan(\omega)$, with corresponding eigenfunction $\pi(x) = \cos(\lambda x)/\sqrt{1+\sin(2\lambda)/(2\lambda)}$ over $[-1, 1]$.

We compare the results obtained with the Nystr\"{o}m method, i.e. solving the eigenvalue problem associated with the matrix obtained by discretizing $\kker$ over the interval $[-1, 1]$, with the results obtained using FE2kind-WGF.
We compare the results of the two methods for increasing precision, corresponding to the number of particles $N$ for FE2kind-WGF and the number of discretization intervals for the Nystr\"{o}m method. 
In particular, we consider $N$ and the number of discretization intervals between 50 and $10^3$; for FE2kind-WGF we set $\gamma = 1/N$ as suggested in Section \ref{sec:tuning}, $\alpha = 10^{-2}$ and iterate for $n_T=400$ which empirically seems sufficient to obtain convergence of the value of $\Fun_\alpha^\eta$ (approximated numerically as described in Section \ref{sec:tuning}). The initial distribution and the reference measure $\pi_0$ are both Gaussians centred at 0 and with small variance ($\sigma^2=0.05^2$), in this case the reference measure carries information on the support of $\pi$, whose mass is mostly concentrated on $[-1, 1]$. 

Since the Nystr\"{o}m method is a deterministic algorithm we take its $\ise$ as reference and investigate the gains obtained by using FE2kind-WGF, see Figure~\ref{fig:kl} first panel. 
FE2kind-WGF produces results up to 5 times more accurate than those of the Nystr\"{o}m method for large $N$, while still outperforming the latter for small $N$.

The second and third panel of Figure~\ref{fig:kl} show the eigenfunction associated with the largest eigenvalue for the exponential and squared exponential kernel obtained with $N=500$ for FE2kind-WGF and 500 discretization intervals for the Nystr\"{o}m method. In the case of the squared exponential kernel, the $\lambda$ used in Algorithm~\ref{alg:second_kind} is that returned by the Nystr\"{o}m method.
The two methods return coherent approximations of the eigenfunction, but our approach does not require a fixed space discretization and has considerable lower error.

\begin{figure}
\centering
\begin{tikzpicture}[every node/.append style={font=\normalsize}]
\node (img1) {\includegraphics[width=0.3\textwidth]{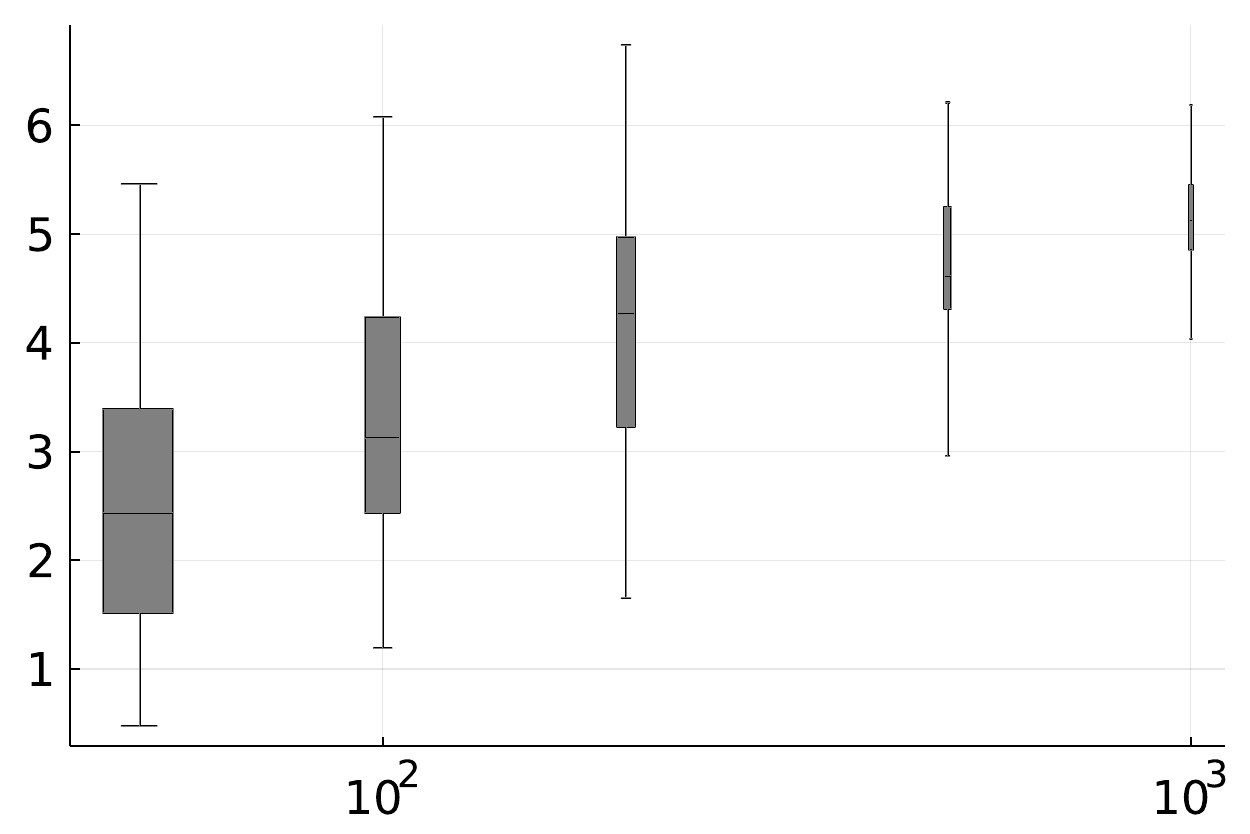}};
\node[below=of img1, node distance = 0, yshift = 1cm] (label1) {$N$};
  \node[left=of img1, node distance = 0, rotate=90, anchor = center, yshift = -0.8cm] {Gain in $\ise$};
\node[right=of img1, node distance = 0, xshift = -1cm] (img2) {\includegraphics[width=0.3\textwidth]{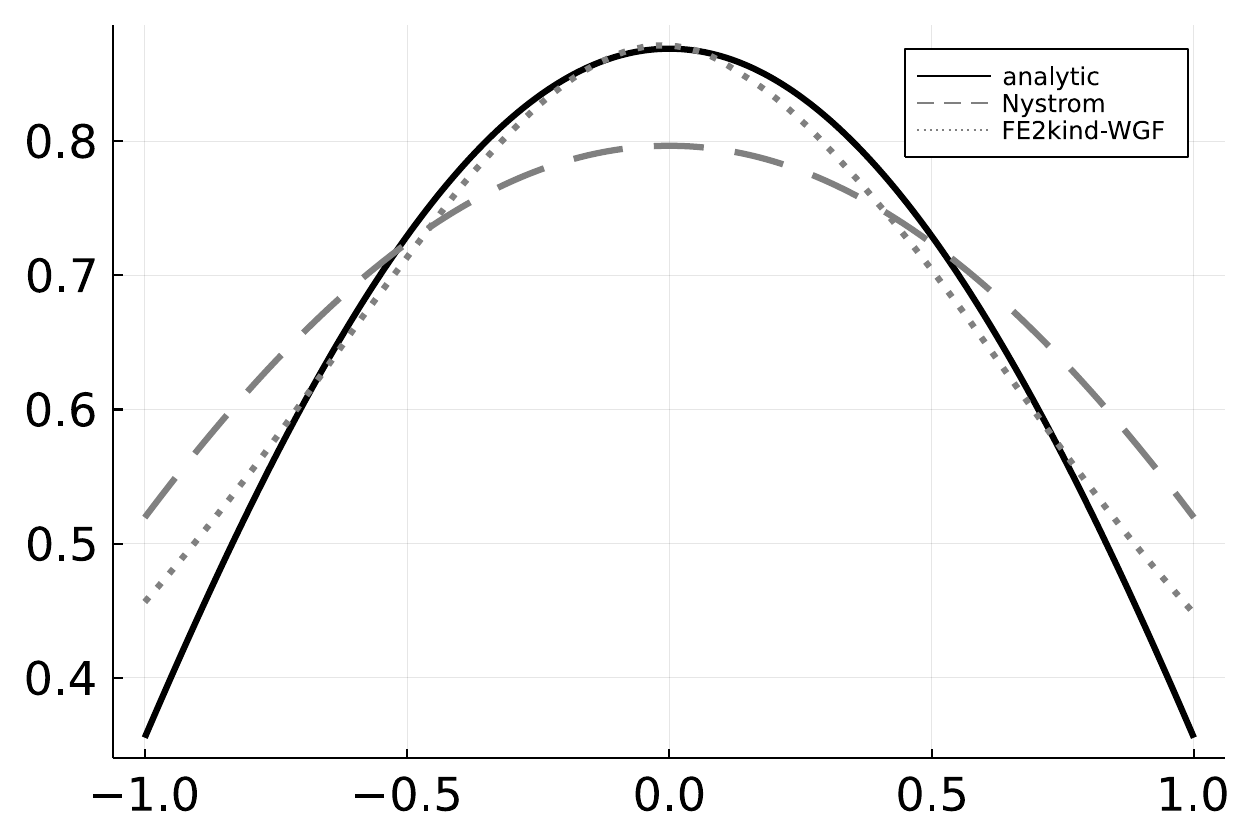}};
\node[below=of img2, node distance = 0, yshift = 1cm] {$x$};
\node[right=of img2, node distance = 0, xshift = -1cm] (img3) {\includegraphics[width=0.3\textwidth]{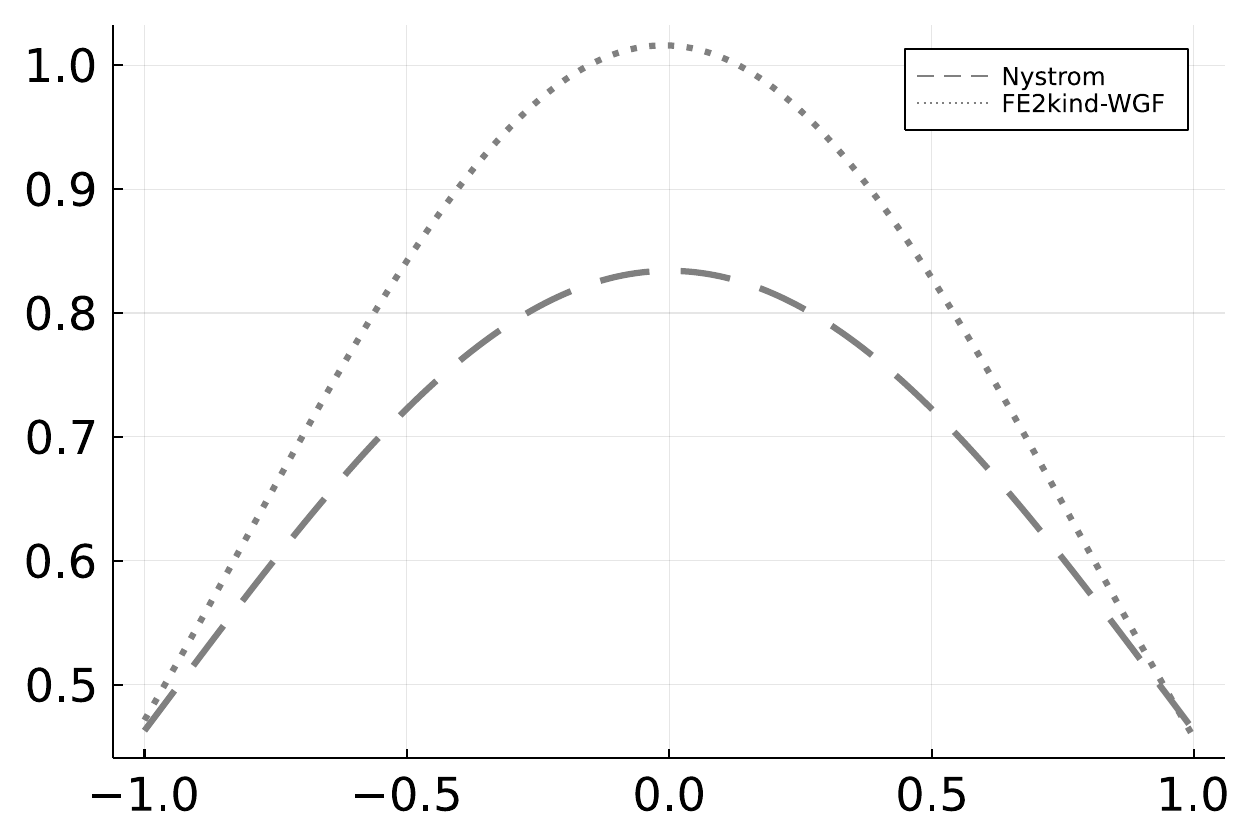}};
\node[below=of img3, node distance = 0, yshift = 1cm] {$x$};
\node[above=of img3, node distance = 0, yshift = -1.2cm] {squared exponential};
\node[above=of img2, node distance = 0, yshift = -1.2cm] {exponential};
  \end{tikzpicture}
\caption{Comparison of Nystr\"om method and FE2kind-WGF to reconstruct the first eigenfunction of common covariance kernels. Left panel: distribution of $\ise$ ratios ($\ise$ of Nystr\"{o}m method divide by $\ise$ of FE2kind-WGF). Middle and right panel: approximation of the eigenfunction of the largest eigenvalue for the exponential and squared exponential kernel.}
\label{fig:kl}
\end{figure}

\subsection{Equilibrium Distribution of Gaussian Process State Space Models}
\label{sec:gp}
Finally, we consider an example in which the forcing term $\varphi(x)\equiv 0$ and $\lambda=1$ and solving the integral equation~\eqref{eq:fe} is equivalent to finding the invariant measure of the kernel $\kker$ (or equivalently, the eigenmeasure corresponding to eigenvalue $\lambda=1$). In this case, the Von Neumann series does not provide a non-trivial solution, and the approach of \cite{doucet2010solving} cannot be applied.

\cite{beckers2016equilibrium} analyses the equilibrium distribution of Gaussian process state space models (GP-SSM) for a one-dimensional SSM with transition function $f(x) = 0.01x^3-0.2x^2+0.2x$. We fit a GP to learn $f$ using $m$ training pairs, $D_m:=(x_i, f(x_i)+\epsilon_i)_{i=1}^m$, where $\epsilon_i$ are independent Gaussians with mean 0 and standard deviation 5. We follow \cite{beckers2016equilibrium} and use $m=20$ input points $(x_i)_{i=1}^m$ uniformly distributed in $[-5, 5]$, and use a squared exponential covariance function $c(x, x') = \sigma_f^2\exp\left(-\norm{x-x'}^2/(2\ell^2)\right)$, with $\ell^2 = 3.59^2$ and $\sigma_f^2 = 4.21^2$.

The predictive distribution of the GP-SSM is given by a Fredholm integral equation~\eqref{eq:fe} with $\kker(x, y) = \N(x;\mu(y, D_m), \sigma^2(y, D_m))$, where 
\begin{align*}
    \mu(y, D_m) &= c(y, x_{1:m})^T(c(x_{1:m}, x_{1:m})+\Id)^{-1}z_{1:m}\\
    \sigma^2(y, D_m) &= c(y, y) - c(y, x_{1:m})^T(c(x_{1:m}, x_{1:m})+\Id)^{-1}c(y, x_{1:m}),
\end{align*}
with $x_{1:m} := (x_1, \dots, x_m)$ and $z_{1:m} := (f(x_1)+\epsilon_1, \dots, f(x_m)+\epsilon_m)$.
We compare the results obtained with FE2kind-WGF with that given by the Nystr\"{o}m method described in \cite{beckers2016equilibrium}. To ensure that the solution is non-trivial, i.e. $\pi(x)\neq 0$ for some $x$, we solve the linear system given by the Nystr\"{o}m method using least-squares with the additional constraint that the solution should be a probability density. FE2kind-WGF automatically enforces this constraint.

For the Nystr\"{o}m method we use 500 nodes in $[-20, 10]$. For FE2kind-WGF we use $N=200$ and $\gamma = 1/N$ as suggested in Section \ref{sec:tuning} and iterate for $n_T=100$ iterations which empirically seems sufficient to obtain convergence of the value of $\Fun_\alpha^\eta$ (approximated numerically as described in Section \ref{sec:tuning}). The reference measure is a Gaussian with mean 0 and standard deviation 1 (this is motivated by the fact that $\pi$ assigns positive mass to 0 and our empirical studies show that concentrated distributions tend to lead to faster convergence than diffuse ones), we set $\alpha=0.001$. Figure~\ref{fig:gp_ssm} shows the predictive distributions obtained with the two methods. To assess the quality of the results obtained we sample from the predictive distributions $\pi$ obtained with the two algorithms $n=2\cdot 10^4$ states $X_1, \dots, X_n$ and apply the transformation $Y_i = \mu(X_i, D_m) + \sigma(X_i, D_m)\xi$, where $\xi$ is a standard Gaussian random variable. To verify that the obtained $\hat{\pi}$ is indeed an invariant distribution we compare our approximations of $\pi$ with the histogram of the $Y_i$s, the fit is good for both algorithms (right panel).

\begin{figure}
\centering
\begin{tikzpicture}[every node/.append style={font=\normalsize}]
\node (img1) {\includegraphics[width=0.4\textwidth]{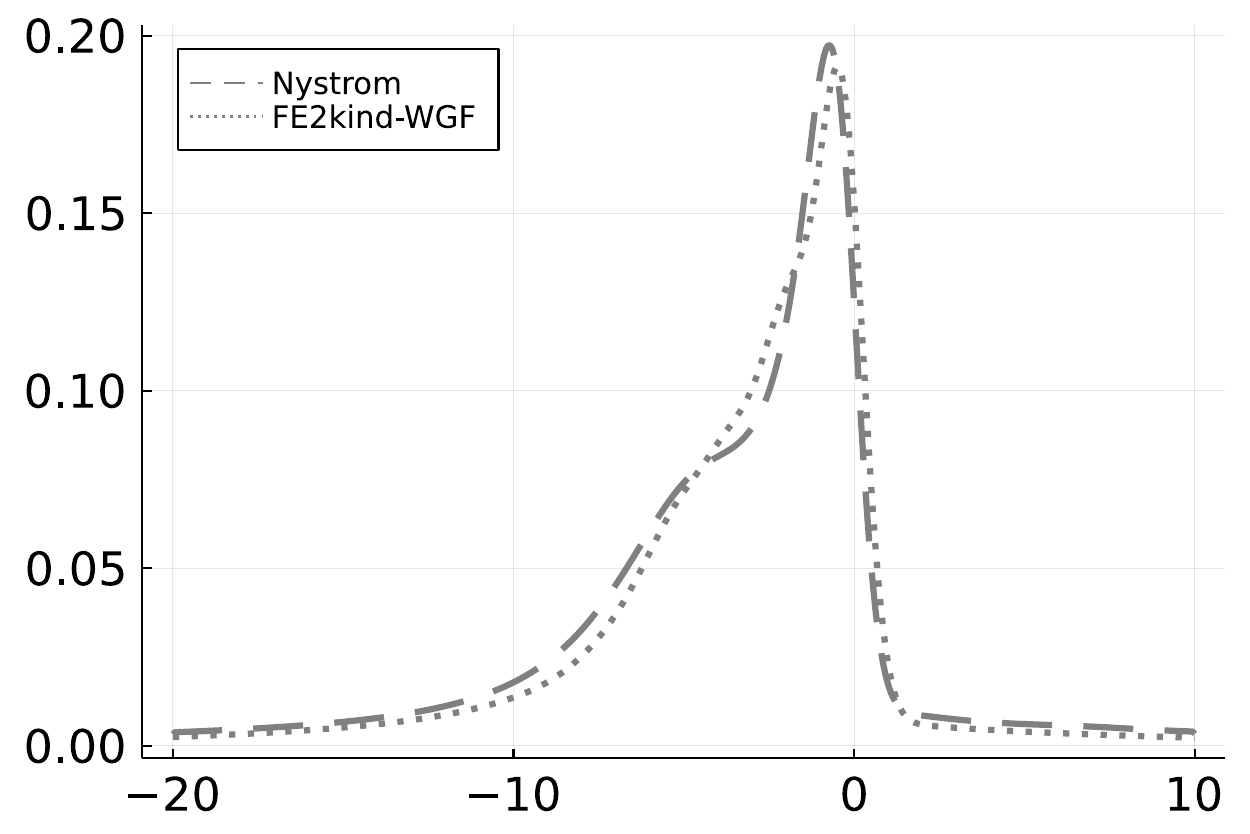}};
\node[below=of img1, node distance = 0, yshift = 1cm] (label1) {$x$};
  \node[left=of img1, node distance = 0, rotate=90, anchor = center, yshift = -0.8cm] {$\pi(x)$};
\node[right=of img1, node distance = 0, xshift = -0.5cm] (img2) {\includegraphics[width=0.4\textwidth]{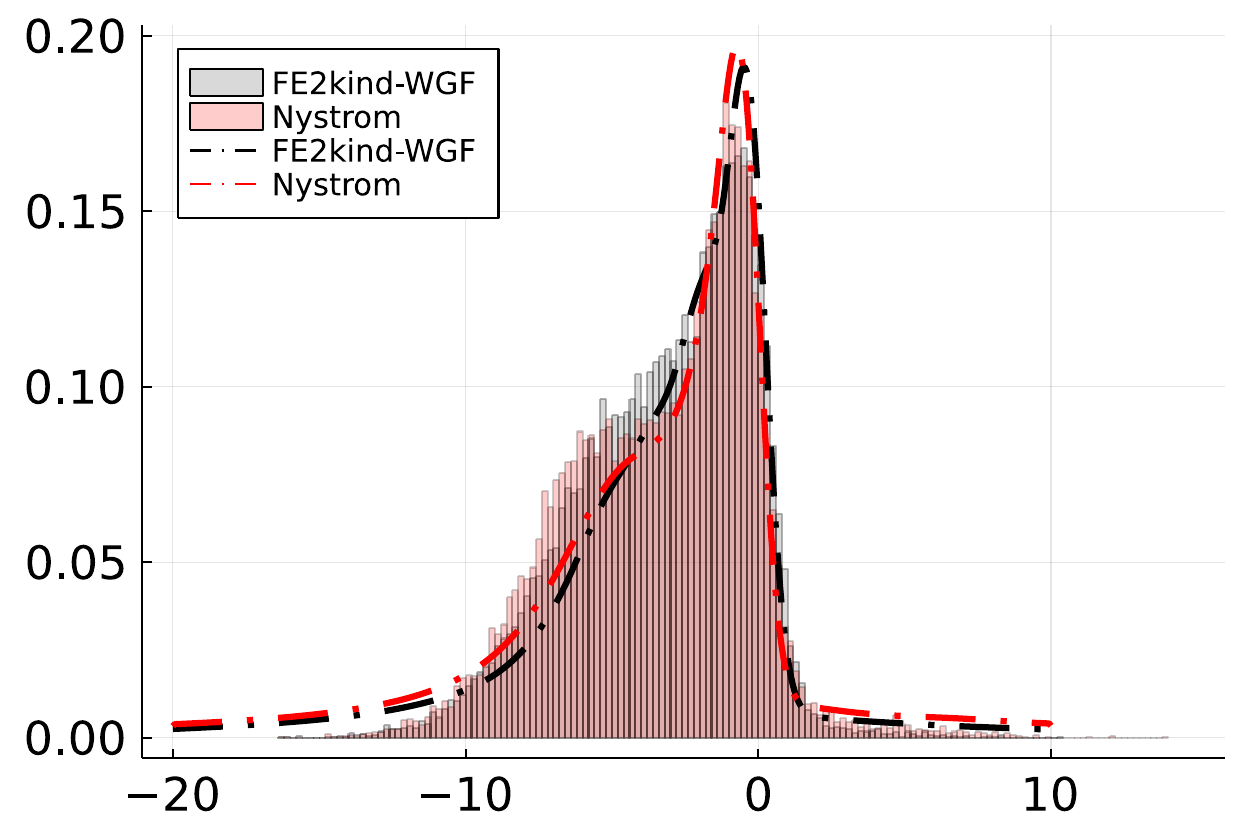}};
\node[below=of img2, node distance = 0, yshift = 1cm] {$x_k$};
  \node[left=of img2, node distance = 0, rotate=90, anchor = center, yshift = -0.8cm] {$x_{k+1}$};
  \end{tikzpicture}
\caption{Predictive distribution for a 1-dimensional GP-SSM. Left: distribution recovered by FE2kind-WGF and Nystr\"{o}m. Right: Histogram of $x_{k+1}$ obtained by sampling $x_k$ from the predictive distribution provided by FE2kind-WGF.}
\label{fig:gp_ssm}
\end{figure}

\section{Discussion}

In this paper we extended the approach of \cite{crucinio2022solving} to Fredholm integral equations of the second kind.
Under similar assumptions to those in \cite{crucinio2022solving} on the kernel $\kker$, we show that the regularized functional~\eqref{eq:G_eta} admits a unique minimizer which is the limiting distribution of the McKean--Vlasov SDE~\eqref{eq:mckean_sde}. In addition, we study the limiting behaviour of the minimizer when the regularization parameter tends to 0 (Proposition~\ref{prop:alphaeta0}).

We employ an interacting particle system~\eqref{eq:particle} to approximate this SDE and derive ergodicity and propagation of chaos results.
Combining the latter with strong convergence results of Euler--Maruyama schemes we obtain the bound~\eqref{eq:guideline} which provides a practical guideline for the selection of the time discretization step $\gamma$ and the number of particles $N$.

Our method addresses a wide variety of equations of the second kind and is competitive with standard methods based on deterministic discretizations. As in the case of \cite{crucinio2022solving}, the presence of the reference measure $\pi_0$ proves beneficial when the problem difficulty increases (Section~\ref{sec:rj}).
In addition, FE2kind-WGF is robust to small deviations from Assumption~\ref{assum:general_kker}, as shown in Section~\ref{sec:kl} where we consider a kernel $\kker$ which is not continuously differentiable.

While we focused on one dimensional examples to benchmark our method against other algorithms proposed in the literature, we believe that FE2kind-WGF has the potential to tackle higher dimensional problems and, as shown in \cite[Section 5.4]{crucinio2022solving} for equations of the first kind, to significantly outperform methods based on deterministic discretizations with the same storage cost. This is also confirmed in the one-dimensional case by Figure~\ref{fig:kl} left panel. One area in which this could prove particularly beneficial is spatial statistics, in which multidimensional spatial Karhunen--Lo{\`e}ve expansions are often employed \cite{fontanella2003dynamic, cressie2008fixed}.

One advantage of the Nystr\"om method against FE2kind-WGF is that the latter requires knowledge of the eigenvalue $\lambda$ for which we seek to find the corresponding eigenfunction. One possible way to relax this requirement is to modify the functional~\eqref{eq:funct_entropy} to $\Fun_\alpha(\pi, \lambda)$, allowing the definition of a minimization problem over the product space $\Pens(\rset^d) \times \rset^d$ as explored in \cite{kuntz2023particle} in a different setting.

%% file: subdifferentials.tex
\section{Proof of Proposition~\ref{prop:convergence_minimum}}
\label{app:functional}

\subsection{Auxiliary Result for the Proof of Proposition~\ref{prop:convergence_minimum}}

The following auxiliary result is adapted from \cite[Appendix A.1]{feydy2019interpolating}.
 \begin{lemma}   
 \label{lemma:kl_convex}
 Let $\nu, \mu$ be $\sigma$-finite and finite measures, respectively, on $\rset^d$, and denote by $\mathcal{M}^+(\rset^d)$ the set of $\sigma$-finite measures on $\rset^d$. Assume that $\mu\ll \nu$. Then, the Kullback--Leibler divergence $\KL{\mu}{\nu}=\int\limits_{\rset^d} \log ((\rmd \mu /\rmd \nu)(y)) \rmd 
 \mu(y)$ is convex in both arguments.
 \end{lemma}
\begin{proof}
The KL divergence above can be written as an $f$-divergence for $f(s)=s\log s$ with the convention $0\log 0=0$, $\KL{\mu}{\nu} = \int\limits_{\rset^d} f(\rmd \mu /\rmd \nu)\rmd \nu$.
The convex conjugate of $f$ is, for all $u\in\rset$,
\begin{align*}
f^\star(u):= \sup_{s>0}(su-f(s)) = e^{u-1},
\end{align*}
and satisfies $f(s)+f^\star(u)\geq su$ by Fenchel’s inequality.

We will now show that the KL divergence can be expressed as the solution of a dual concave problem:
\begin{align}
\label{eq:dual_kl}
\KL{\mu}{\nu}=\sup_{h\in \rmc_b(\rset^d)} \left[\int\limits_{\rset^d} h(x) \rmd\mu(x) - \int\limits_{\rset^d} e^{h(x)-1}\rmd\nu(x)\right],
\end{align}
where $\rmc_b(\rset^d)$ denotes the set of real-valued bounded continuous functions over $\rset^d$. Since the supremum in~\eqref{eq:dual_kl} is taken over a convex function of $\nu, \mu$ we will consequently obtain that $\KL{\mu}{\nu}$ is convex in both arguments.

We now prove~\eqref{eq:dual_kl}. If $\mu$ is absolutely continuous w.r.t. $\nu$ we have 
\begin{align*}
&\KL{\mu}{\nu} - \left[\int\limits_{\rset^d} h(x) \rmd\mu(x) - \int\limits_{\rset^d} e^{h(x)-1}\rmd\nu(x)\right] \\
&\qquad\qquad= \int\limits_{\rset^d} f\left(\frac{\rmd \mu }{\rmd \nu}\right)\rmd \nu+\int\limits_{\rset^d} f^\star(h(x))\rmd\nu(x)-\int\limits_{\rset^d} h(x) \rmd\mu(x)\geq 0,
\end{align*}
by Fenchel's inequality. This gives an upper bound on the $\sup$ in~\eqref{eq:dual_kl}. We will now establish that it is also bounded below by the KL divergence.

Let us define $h^\star := \log \rmd \mu /\rmd \nu+1$ and observe that
\begin{align*}
&\KL{\mu}{\nu} = \left[\int\limits_{\rset^d} h^\star(x) \rmd\mu(x) - \int\limits_{\rset^d} e^{h^\star(x)-1}\rmd\nu(x)\right].
\end{align*}
Then, if $h_n = \log \rmd \mu /\rmd \nu\1_{1/n \leq \log \rmd \mu /\rmd \nu\leq n}+1\in\rmc_b(\rset^d)$, the monotone and dominated convergence theorems
give
\begin{align*}
& \left[\int\limits_{\rset^d} h_n(x) \rmd\mu(x) - \int\limits_{\rset^d} e^{h_n(x)-1}\rmd\nu(x)\right]\to\KL{\mu}{\nu}
\end{align*}
as $n\to\infty$. Hence, we obtain a lower bound on the $\sup$ in~\eqref{eq:dual_kl}. 

It follows that the optimal value in~\eqref{eq:dual_kl} is bounded above and below by $\KL{\mu}{\nu}$ giving the result.
\end{proof}


\subsection{Proof of Proposition~\ref{prop:convergence_minimum}}
\begin{enumerate}[label=(\alph*)]

\item 

In the case $\eta = 0$ we have
\begin{align*}
\Fun(\pi) &= \KL{\pi}{\varphi+\lambda \int\limits_{\rset^d} \kker(\cdot, y) \pi( y)\rmd y}\geq 0
\end{align*}
because $\pi$ and $\nu(\rmd x) = \left(\varphi(x)+\lambda \int\limits_{\rset^d} \kker(x, y) \pi( y)\rmd y\right)\rmd x$ are probabilities.
To show that $\Fun(\pi)$ is convex, consider the functional $$\mathcal{M}:\pi\mapsto \left(\varphi+\lambda \int\limits_{\rset^d} \kker(\cdot, y) \rmd\pi( y)\right)$$ which is linear w.r.t $\pi$ and therefore convex. 
Thus, we have that $\Fun(\pi)$ can be written as $\KL{\pi}{\mathcal{M}(\pi)}$.
 In addition, \cite[Lemma 
 1.4.3-(b)]{dupuis1997weak}  guarantees that the KL between probability distributions is jointly convex in both arguments.
 Thus, for all $t\in[0, 1]$
 \begin{align*}
     &\KL{t\pi_1+(1-t)\pi_2}{\mathcal{M}(t\pi_1+(1-t)\pi_2)}\\ &\qquad=\KL{t\pi_1+(1-t)\pi_2}{t\mathcal{M}(\pi_1)+(1-t)\mathcal{M}(\pi_2)}\\
     &\qquad\leq t\KL{\pi_1}{\mathcal{M}(\pi_1)}+(1-t)\KL{\pi_2}{\mathcal{M}(\pi_2)},
 \end{align*}
 which shows that $\Fun^\eta$ is convex in $\pi$ when $\eta = 0$.

For the case $\eta>0$ recall that 
\begin{align*}
\Fun^\eta(\pi) &= \KL{\pi}{\varphi+\lambda \int\limits_{\rset^d} \kker(\cdot, y) \pi( y)\rmd y+\eta},
\end{align*} and take $f(s)=s\log s$ with the convention $0\log 0=0$. Then, $\Fun^\eta(\pi)$ can be written as $\int\limits_{\rset^d} f(\rmd \pi /\rmd \nu)\rmd \nu$ with $\nu(\rmd x) = \left(\varphi(x)+\lambda \int\limits_{\rset^d} \kker(x, y) \pi( y)\rmd y+\eta\right)\rmd x$ a $\sigma$-finite measure, while $\pi$ is a probability measure. Using Jensen's inequality on the convex function $f$ w.r.t. the probability measure $\nu-\eta$ we have
\begin{align*}
\int\limits_{\rset^d} f\left(\frac{\rmd \pi}{\rmd \nu}(x)\right) \nu(\rmd x)
& = \int\limits_{\rset^d} f\left(\frac{\rmd \pi}{\rmd \nu}(x)\right) (\nu(\rmd x)-\eta\rmd x)+\eta\int\limits_{\rset^d} f\left(\frac{\rmd \pi}{\rmd \nu}(x)\right)\rmd x \\
& \geq f\left(\int\limits_{\rset^d} \frac{\rmd \pi}{\rmd \nu}(x)(\nu(\rmd x)-\eta\rmd x)\right)+\eta\int\limits_{\rset^d} f\left(\frac{\rmd \pi}{\rmd \nu}(x)\right)\rmd x. 
\end{align*}
Since the integrand in the first term is always positive, the integral is in the support of $f$, and we can lower bound the first term with the minimum of $f$:
\begin{align}
\label{eq:lb_term1}
    \int\limits_{\rset^d} f\left(\frac{\rmd \pi}{\rmd \nu}(x)\right) \nu(\rmd x) &\geq -e^{-1}+\eta\int\limits_{\rset^d} f\left(\frac{\rmd \pi}{\rmd \nu}(x)\right)\rmd x. 
\end{align}
Now, we observe that

\resizebox{0.95\textwidth}{!}{\noindent
\hspace*{-0.5cm}\parbox{\textwidth}{\noindent
\begin{align}
\label{eq:lb_term2}
    \eta\!&\int\limits_{\rset^d} \!\! f\left(\frac{\rmd \pi}{\rmd \nu}(x)\right)\rmd x \\
    &= \int\limits_{\rset^d} \!\! \frac{\eta \pi(x)}{\varphi(x)+\lambda \int\limits_{\rset^d} \kker(x, y) \pi( y)\rmd y+\eta}\log \left(\frac{\pi(x)}{\varphi(x)+\lambda \int\limits_{\rset^d} \kker(x, y) \pi( y)\rmd y+\eta}\right)\rmd x\notag \\
    &\geq \eta((1+\lambda)\Mtt+\eta)^{-1}\Fun^\eta(\pi),\notag
\end{align}

}}
where we used Assumption~\ref{assum:general_kker} to obtain the lower bound.
Using~(\ref{eq:lb_term1},\ref{eq:lb_term2}) and the definition of $\Fun^\eta$ we find
\begin{align*}
  \Fun^\eta(\pi)\geq& -e^{-1} +\frac{\eta}{(1+\lambda)\Mtt+\eta} \Fun^\eta(\pi)
\end{align*}
and thus 
\begin{align*}
  \Fun^\eta(\pi)\geq -\left(1-\frac{\eta}{(1+\lambda)\Mtt+\eta}\right)^{-1} e^{-1}=:\Ctt.
\end{align*}
The fact that $\eta > 0$ guarantees that any probability measure with a Lebesgue density is absolutely continuous w.r.t. $\nu$ and so the convexity of $\Fun^\eta$ follows from Lemma~\ref{lemma:kl_convex} with $\mu = \pi$ and $$\nu(\rmd x) = \left(\varphi(x)+\lambda \int\limits_{\rset^d} \kker(x, y) \pi( y)\rmd y+\eta\right)\rmd x.$$

\item
To see that $\Fun_\alpha^\eta$ is proper observe that, since $\KL{\pi}{\pi_0}\geq 0$, $\Fun_{\alpha}^\eta(\pi)\geq \Ctt$
  for all $\pi \in \Pens(\rset^d)$. Take a reference measure $\pi_0$ with finite entropy, i.e. $|\Hent(\pi_0)| < +\infty$. We have 
\begin{align*}
    \Fun_\alpha^\eta(\pi_0) = -\Hent(\pi_0) - \int_{\rset^d} \pi_0(x)\log(\varphi(x)+\lambda\pi_0[\kker(x, \cdot)] + \eta )\rmd x.
\end{align*}
In addition, under Assumption~\ref{assum:general_kker}, for any $x \in \rset^d$ and $\pi \in \Pens(\rset^d)$
\begin{equation*}
\vert \log(\varphi(x)+\lambda\pi[\kker(x, \cdot)] + \eta )\vert \leq \max(\vert \log((1+\lambda)\Mtt+\eta)\vert, \vert \log(\eta)\vert),
\end{equation*}
showing that
\begin{align*}
- \int_{\rset^d} \pi_0(x)\log(\varphi(x)+\lambda\pi_0[\kker(x, \cdot)] + \eta )\rmd x \leq     \max(\vert \log((1+\lambda)\Mtt+\eta)\vert, \vert \log(\eta)\vert).
\end{align*}
Hence, there exists $\pi=\pi_0$ for which $\Fun_\alpha^\eta(\pi)<+\infty$ and the functional is proper.

Let $(\pi_n)_{n\geq 1} \in (\Pens(\rset^d))^\nset$ be such that $\lim_{n \to +\infty} \pi_n = \pi \in \Pens(\rset^d)$. Since for any $(x, y )\in \rset^d\times \rset^d$, $\vert \kker(x,y)\vert \leq \Mtt$ and $\kker$ is continuous, we have that $\lim_{n \to +\infty} \pi_n[\kker(x, \cdot)] = \pi[\kker(x, \cdot)]$ for each $x$. This and the fact that the Kullback--Leibler divergence is lower semi-continuous in both arguments \cite[Lemma
1.4.3-(b)]{dupuis1997weak} guarantees that $\Fun^\eta$ is lower semi-continuous.
\cite[Lemma
1.4.3-(b)]{dupuis1997weak} also guarantees that $\KL{\pi}{\pi_0}$ is strictly convex and lower semi-continuous.
It follows that $\Fun^\eta_\alpha$ is strictly convex and lower semi-continuous.

To see that $\Fun^\eta_\alpha$ is coercive observe that $\Fun^\eta_\alpha$ is the sum of the lower bounded lower semi-continuous functional $\Fun^\eta$ and the coercive functional \cite[Lemma
 1.4.3-(c)]{dupuis1997weak} $\KL{\pi}{\pi_0}$, then for any $\beta\in\rset$
\begin{align*}
    S:=&\left\lbrace\pi\in \Pens(\rset^d): \Fun^\eta(\pi)+\alpha\KL{\pi}{\pi_0}\leq \beta\right\rbrace\\
    \subseteq& \left\lbrace\pi\in \Pens(\rset^d): \alpha\KL{\pi}{\pi_0}\leq \beta-\Ctt\right\rbrace=:\tilde{S},
\end{align*}
since $\Fun^\eta(\pi)\geq \Ctt$.
$\tilde{S}$ is relatively compact since $\KL{\pi}{\pi_0}$ is coercive and thus $S$ is also relatively compact, showing that $\Fun^\eta_\alpha$ is coercive.
\end{enumerate}

\section{Convergence of Minimizers}
\label{app:gamma}
\subsection{Basics on $\Gamma$-convergence}

We start this section by recalling some basic facts on $\Gamma$-convergence. We refer to \cite{dalmaso1993introduction} for a more complete introduction to $\Gamma$-convergence.

First, we recall that a function $f: \ \msx \to \rset$ (where $\msx$ is a metric
space) is coercive if for any $t \in \rset$, $f^{-1}(\ocint{-\infty, t})$ is
relatively compact. This definition can be extended to the case where $\msx$ is
only a topological space, see \cite[Definition 1.12]{dalmaso1993introduction}.

\begin{definition}
\label{def:gammas}
Let $F = \{f_\alpha: \, \alpha \in \msa\}$
where $\msa$ is a topological space and for any $\alpha \in \msa$,
$f_\alpha: \msx \to \rset \cup \{+\infty\}$ with $\msx$ a metric space.  We say
that $F$ is a $\Gamma$-system if the following hold:
\begin{enumerate}[wide, labelindent=0pt, label=(\alph*)]
\item \label{def:a} For any $x \in \msx$, $\alpha^\star \in \msa$, 
  $(\alpha_n)_{n \in \nset} \in \msa^\nset$ with
  $\lim_{n \to +\infty} \alpha_n = \alpha^\star$ and 
  $(x_n)_{n \in \nset} \in \msx^\nset$ with $\lim_{n \to +\infty} x_n =x$
  we have $\liminf_{n \to +\infty} f_{\alpha_n}(x_n) \geq f_{\alpha^\star}(x)$.
\item \label{def:b} For any $x \in \msx$, $\alpha^\star \in \msa$,
  $(\alpha_n)_{n \in \nset} \in \msa^\nset$ with
  $\lim_{n \to+\infty} \alpha_n = \alpha^\star$ there exists
  $(x_n)_{n \in \nset} \in \msx^\nset$ with $\lim_{n \to +\infty} x_n =x$ such
  that $\lim_{n \to +\infty} f_{\alpha_n}(x_n) = f_{\alpha^\star}(x)$.
\end{enumerate}
\end{definition}

\begin{prop}[Corollary 7.20, \cite{dalmaso1993introduction}.]
  \label{prop:gamma_cv_without}
  Let $F = \{f_\alpha:\, \alpha \in \msa\}$ where $\msa$ is a
  topological space and for any $\alpha \in \msa$,
  $f_\alpha: \msx \to \rset \cup \{+\infty\}$ with $\msx$ a metric space. Assume
  that $F$ is a $\Gamma$-system  and is such that for any
  $\alpha \in \msa$, $f_\alpha$ admits a unique minimizer. Then
  for any $(\alpha_n)_{n \in \nset} \in \msa^\nset$ and
  $(x_n)_{n \in \nset} \in \msx^\nset$ such that for any $n \in \nset$, $x_n$ is
  the minimizer of $f_{\alpha_n}$ and
  $\lim_{n \to +\infty} \alpha_n =\alpha^\star$, if there exists
  $x^\star \in \msx$ such that $\lim_{n \to +\infty} x_n =x^\star$ then
  $x^\star$ is a minimizer of $f_{\alpha^\star}$.
\end{prop}

\subsection{Proof of Proposition~\ref{prop:alphaeta0}}

We start by showing an auxiliary result.
\begin{lemma}
\label{lem:min2}
Under Assumptions~\ref{assum:general_kker} and \ref{assum:pi0_second}, for
  any $\alpha > 0$, $\eta \geq 0$,
  $\pi_{\alpha, \eta}^\star \in \Pens_2^{ac}(\rset^d)$. 
\end{lemma}
\begin{proof}
    Take $\alpha > 0$ and  $\eta \geq 0$. Since by Proposition~\ref{prop:convergence_minimum},
  $\Fun_{\alpha}^\eta$ is proper,
  $\Fun_{\alpha}^\eta(\pi_{\alpha, \eta}^\star) < +\infty$ and therefore,
  $\KL{\pi_{\alpha, \eta}^{\star}}{\pi_0} < +\infty$. Hence,
  $\pi_{\alpha, \eta}^\star$ admits a density w.r.t to $\pi_0$. Since $\pi_0$ is
  dominated by the Lebesgue measure, we get that $\pi_{\alpha, \eta}^\star$
  admits a density with respect to the Lebesgue measure.
  
  In addition, let
  $\pi_1 \in \Pens(\rset^d)$ be such that for any $x \in \rset^d$ we have for some $\tau > 0$
  \begin{equation*}
    (\rmd \pi_1 / \rmd \pi_0)(x) = f(x) / \pi_0[f], \qquad f(x) = \exp[(\tau / 2) \norm{x}^2]. 
  \end{equation*}
  Since $\KL{\pi_{\alpha, \eta}^\star}{\pi_0} < +\infty$ we have
  $\int_{\rset^d} |\log((\rmd \pi_{\alpha, \eta}^\star / \rmd \pi_0)(x))| \rmd
  \pi_{\alpha, \eta}^\star(x) < +\infty$ and
  \begin{align}
    \label{eq:kl1_max}
    &\int_{\rset^d} \log((\rmd \pi_{\alpha, \eta}^\star / \rmd \pi_1)(x)) \1_{[1,+\infty)}((\rmd \pi_{\alpha, \eta}^\star / \rmd \pi_1)(x)) \rmd
    \pi_{\alpha, \eta}^\star(x) \\
    & \qquad \quad = \int_{\rset^d}\left[ \log((\rmd \pi_{\alpha, \eta}^\star / \rmd \pi_0)(x)) - (\tau/2)\norm{x}^2 +\log(\pi_0[f]) \right]\notag\\
    &\qquad\qquad\times \1_{[1,+\infty)}((\rmd \pi_{\alpha, \eta}^\star / \rmd \pi_1)(x)) \rmd
      \pi_{\alpha, \eta}^\star(x)\notag \\
    & \qquad \quad \leq \int_{\rset^d} \left[|\log((\rmd \pi_{\alpha, \eta}^\star / \rmd \pi_0)(x))|  +\log(\pi_0[f]) \right]\notag\\
    &\qquad\qquad\times \1_{[1,+\infty)}((\rmd \pi_{\alpha, \eta}^\star / \rmd \pi_1)(x)) \rmd
  \pi_{\alpha, \eta}^\star(x) < +\infty.\notag
  \end{align}
  Since $t \mapsto t |\log(t)|$  is bounded on $[0,1]$ we have that
  \begin{equation}
    \label{eq:kl1_min}
    - \int_{\rset^d} \log((\rmd \pi_{\alpha, \eta}^\star / \rmd \pi_1)(x)) \1_{[0,1]}((\rmd \pi_{\alpha, \eta}^\star / \rmd \pi_1)(x)) \rmd
    \pi_{\alpha, \eta}^\star(x) < +\infty. 
  \end{equation}
  Combining~\eqref{eq:kl1_min} and~\eqref{eq:kl1_max} we get that
  $\KL{\pi_{\alpha, \eta}^\star}{\pi_1} < +\infty$. Therefore, using \cite[Equation
  2.6]{csiszar1975divergence}, we get that
  \begin{align*}
   &(\tau/2) \int_{\rset^d} \norm{x}^2 \rmd \pi_{\alpha, \eta}^\star(x) - \log(\pi_0[f]) \\ 
    & \qquad  = \int_{\rset^d} \log((\rmd \pi_1 / \rmd \pi_0)(x)) \rmd \pi_{\alpha, \eta}^\star(x) = \KL{\pi_{\alpha, \eta}^\star}{\pi_0} - \KL{\pi_{\alpha, \eta}^\star}{\pi_1} < +\infty. 
  \end{align*}
  Therefore, $\pi_{\alpha, \eta}^\star \in \Pens_2(\rset^d)$.

\end{proof}

We can now show that the results in Proposition~\ref{prop:alphaeta0} hold.

\begin{enumerate}
    \item

Consider for fixed $\alpha$ the family $\{\Fun_{\alpha}^\eta:\mathcal{Q}_2^{ac}(\rset^d)\to \rset^+,\ s.t.\ \eta \geq 0\}$.
 Let  $(\eta_n)_{n \in \nset^\star} \in
  [0, +\infty)^{\nset^\star}$ be such that $\underset{n \to +\infty}{\lim} \eta_n = 0$. 
For any $(\pi_n)_{n \in \nset^\star} \in \Pens_2^{ac}(\rset^d)^{\nset^\star}$ such that
  $\lim_{n \to +\infty} \pi_n =\pi$ we now show that
  \begin{equation*}
    \liminf_{n \to +\infty} \Fun_{\alpha}^{\eta_n}(\pi_n) \geq \Fun_\alpha(\pi).
  \end{equation*}
  Using the properties of $\liminf$ and the lower semi-continuity of the KL divergence \cite[Lemma 1.4.3-(b)]{dupuis1997weak} we have 
  \begin{align*}
     \liminf_{n \to +\infty} \Fun_{\alpha}^{\eta_n}(\pi_n)&\geq \liminf_{n \to +\infty} \Fun^{\eta_n}(\pi_n) +\alpha\liminf_{n \to +\infty}\KL{\pi_n}{\pi_0}\\
     &\geq \liminf_{n \to +\infty} \Fun^{\eta_n}(\pi_n) +\alpha\KL{\pi}{\pi_0}.
  \end{align*}
  Let us denote $\bar{\eta}:= \sup_{n \in \nset}\eta_n$, then
  \begin{align*}
      \Fun^{\eta_n}(\pi_n) &= \Fun^{\bar{\eta}}(\pi_n)+\int_{\rset^d}\pi_n(x)\log\frac{\varphi(x)+\lambda\pi_n[\kker(x, \cdot)] + \bar{\eta}}{\varphi(x)+\lambda\pi_n[\kker(x, \cdot)] + \eta_n}\rmd x.
  \end{align*}
  Using again the properties of $\liminf$ and the lower semi-continuity of $\Fun^{\eta}$ established in Proposition~\ref{prop:convergence_minimum}
  we have
  \begin{align*}
     \liminf_{n \to +\infty} \Fun^{\eta_n}(\pi_n)
     &\geq \Fun^{\bar{\eta}}(\pi)+\liminf_{n \to +\infty} \int_{\rset^d}\pi_n(x)\log\frac{\varphi(x)+\lambda\pi_n[\kker(x, \cdot)] + \bar{\eta}}{\varphi(x)+\lambda\pi_n[\kker(x, \cdot)] + \eta_n}\rmd x.
  \end{align*}
  Moreover, the function
  \begin{align*}
      f_n(x):= \pi_n(x)\log\frac{\varphi(x)+\lambda\pi_n[\kker(x, \cdot)] + \bar{\eta}}{\varphi(x)+\lambda\pi_n[\kker(x, \cdot)] + \eta_n}\geq 0
  \end{align*}
  by construction and therefore Fatou's Lemma allows us to conclude 
  \begin{align*}
     \liminf_{n \to +\infty} \Fun^{\eta_n}(\pi_n)
     &\geq \Fun^{\bar{\eta}}(\pi)\\
     &+\int_{\rset^d}\pi(x)\log\frac{\varphi(x)+\lambda\pi_n[\kker(x, \cdot)] + \bar{\eta}}{\varphi(x)+\lambda\pi_n[\kker(x, \cdot)]}\rmd x \\
     &=\Fun(\pi).
  \end{align*}

  This guarantees condition~\ref{def:a} in Definition~\ref{def:gammas} holds for any
  $(\pi_n)_{n \in \nset^\star} \in \mathcal{Q}_2^{ac}(\rset^d)^{\nset^\star}\subset \Pens_2^{ac}(\rset^d)^{\nset^\star}$.

    To check condition~\ref{def:b} in Definition~\ref{def:gammas} consider w.l.o.g. 
  $\pi \in \mathcal{Q}_2^{ac}(\rset^d)$ which assigns positive mass to $[-1, 1]^d$. For any $n \in \nset^\star$, we consider
  $\pi_n $ with density w.r.t. the Lebesgue measure given for any $x \in \rset^d$ by
  \begin{align}
  \label{eq:pin}
   \pi_n(x) = \frac{\pi(x)\ind_{[-n, n]^d}(x)}{Z_n},\qquad \textrm{where } Z_n:= \int_{[-n, n]^d} \pi(x)\rmd x .  
  \end{align}
We observe that by construction $0<Z_1\leq Z_2 \leq \dots \leq 1$ and $|\pi_n(x)|\leq \pi(x)/Z_1$ and so $\pi_n \in \mathcal{Q}_2^{ac}(\rset^d)$ for all $n\geq 1$.
Therefore, we have
    $\lim_{n \to +\infty} \wassersteinD[2](\pi_n, \pi)^2  =0$ as $\pi_n \to \pi$ weakly and $\int_{\rset^d}\norm{x}^2\pi_n(x)\rmd x \to \int_{\rset^d}\norm{x}^2\pi(x)\rmd x $. 

  We now show that $\lim_{n \to +\infty} \Fun^{\eta_n}(\pi_n) = \Fun(\pi)$.
  Consider the function $x\mapsto \pi_n(x)\log(\varphi(x)+\lambda\pi_n[\kker(x, \cdot)] + \eta_n)$.
For any $n \in \nset^\star$ and $x \in \rset^d$, $\log(\varphi(x)+\lambda\pi_n[\kker(x, \cdot)] + \eta_n) \leq \log((1+\lambda)\Mtt + \sup_{n \in \nset}
   \eta_n)$ and thus
   \begin{align}
   \label{eq:integrable1}
       \pi_n(x)\log(\varphi(x)+\lambda\pi_n[\kker(x, \cdot)] + \eta_n) \leq \log((1+\lambda)\Mtt + \sup_{n \in \nset} \eta_n)\frac{\pi(x)}{Z_1}.
   \end{align}
   Define $\Phi: \ \rset^d \to \rset$, $\Phi: x \mapsto \norm{x}^2$.
   Using Jensen's inequality and the fact that under Assumption~\ref{assum:pi0_second} there
   exists $C_2 \geq 0$ such that for any $(x, y) \in \rset^d\times\rset^d$,
   $\kker(x,y)\geq C_2^{-1}\exp[-C_2(1+\Phi(x) + \Phi(y))]$, we have for any
   $y \in \rset^d$ and $n \in \nset^\star$ that
   \begin{align*}
     \log(\varphi(x)+\lambda\pi_n[\kker(x, \cdot)] + \eta_n) &\geq \log(\varphi(x)+\lambda\pi_n[\kker(x, \cdot)]) \\
     &\geq \lambda\left(-C_2\Phi(x) - C_2\sup_{n \in \nset^\star} \pi_n[\Phi] - C_2 -\log C_2\right),
   \end{align*}
   where we used the fact that $\varphi(x)\geq 0$ under Assumption~\ref{assum:general_kker}.
   Since, $(\pi_n)_{n \in \nset^\star}$ is relatively compact in $\Pens_2(\rset^d)$ \cite[Proposition 7.1.5]{ambrosio2008gradient},  there exists $C \geq 0$ such that for any
   $n \in \nset^\star$, $\pi_n[\Phi] \leq C$, see \cite[Definition
   6.8]{villani2009optimal}.  Hence,
    \begin{align}
   \label{eq:integrable2}
       \pi_n(x)\log(\varphi(x)+\lambda\pi_n[\kker(x, \cdot)] + \eta_n) \geq \lambda\left(-C_2\Phi(x) -C_2(C+1) -\log C_2\right)\frac{\pi(x)}{Z_1}.
   \end{align}
   Combining~\eqref{eq:integrable1}--\eqref{eq:integrable2} and the fact that $\pi\in\Pens_2^{ac}(\rset^d)$, the function $x\mapsto \pi_n(x)\log(\varphi(x)+\lambda\pi_n[\kker(x, \cdot)] + \eta_n)$ is bounded by an integrable function.
   
   Consider now $x\mapsto \pi_n(x)\log \pi_n(x)$ and observe that
    \begin{align*}
       \pi_n(x) \log \pi_n(x) = \pi_n(x) \log \pi(x) + \log(1/Z_n) \pi_n(x)  
     \end{align*}

      As $\pi\in\mathcal{Q}_2^{ac}(\rset^d)\subset \Pens_2^{ac}(\rset^d)$ it has finite second moment. This and the fact that $\pi$ has  essentially bounded density imply that its entropy $\Hent(\pi)$ is finite by \cite[Theorem 1]{ghourchian2017existence}. From the definition of the Lebesgue integral we have that $\int_{\rset^d} \pi(x) |\log \pi(x) | \rmd x < \infty$.
    
     We can further bound
     \begin{align}
       |\pi_n(x) \log \pi_n(x)| \leq& \pi_n(x) |\log \pi(x)| + \log(1/Z_n) \pi_n(x)   \nonumber\\
                          \leq& \pi_n(x) (|\log \pi(x)| + \log(1/Z_1)) \nonumber\\
                          \leq& \pi(x) (|\log \pi(x)| + \log(1/Z_1)|)/Z_1              \label{eq:hent_bound}              
     \end{align}
     which is an integrable function whose integral is bounded.

     Hence, noting that $\pi_n(x) \log \pi_n(x)$ converges pointwise to $\pi(x) \log \pi(x)$, we have that  $ \int_{\rset^d} \pi_n(x) \log \pi_n(x) \rmd x \to \int_{\rset^d} \pi(x) \log \pi(x) \rmd x$ by the Dominated Convergence Theorem, and that $\lim_{n \to +\infty} \Fun^{\eta_n}(\pi_n) = \Fun(\pi)$.

    In what follows, we show that $\lim_{n \to +\infty} \KL{\pi_n}{\pi_0} = \KL{\pi}{\pi_0}$.
     Using Assumption~\ref{assum:pi0_second} there exists $C_1 \geq 0$ and $\tau > 0$ such that
   for any $x \in \rset^d$, $| U(x)|\leq C_1 + \tau \norm{x}^2$.
    Since $\pi_n$ is upper bounded by an integrable function, so is $x\mapsto U(x)\pi_n(x)$ and the dominated convergence theorem guarantees $\int_{\rset^d}U(x)\pi_n(x)\rmd x\to \int_{\rset^d}U(x)\pi(x)\rmd x$ as $n\to \infty$.
    Then using the boundedness of $x\mapsto \pi_n(x)\log\pi_n(x)$ established above and the dominated convergence theorem we immediately have $\lim_{n \to +\infty} \KL{\pi_n}{\pi_0} = \KL{\pi}{\pi_0}$.

    Therefore, we have that
   $\Fun_\alpha(\pi) = \lim_{n \to +\infty} \Fun_{\alpha}^{\eta_n}(\pi_n)$, and $\{\Fun_{\alpha}^\eta:\mathcal{Q}_2^{ac}(\rset^d)\to \rset^+,\ s.t.\, \eta \geq 0\}$ is a
   $\Gamma$-system. 
   In addition, for any $\eta>0$, $\Fun_\alpha^\eta$ admits a unique minimizer $\pi_{\alpha, \eta}^\star$ by Proposition~\ref{prop:convergence_minimum},  which under our assumptions, satisfies $\pi_{\alpha, \eta}^{\star}\in\mathcal{Q}_2^{ac}(\rset^d)$.
   Using Proposition~\ref{prop:gamma_cv_without}, we get that $\pi_\alpha^\star$ is a
   minimizer of $\Fun_\alpha$.

\item 
Consider $\{\Fun_{\alpha}^\eta:\mathcal{Q}_2^{ac}(\rset^d)\to \rset^+,\ s.t.\ \alpha, \eta \geq 0\}$.
  Let
  $(\alpha_n)_{n \in \nset} \in (0, +\infty)^\nset, (\eta_n)_{n \in \nset} \in
  [0, +\infty)^\nset$ be such that $\underset{n \to +\infty}{\lim} \alpha_n = 0$,
  $\underset{n \to +\infty}{\lim} \eta_n = 0$. For any
  $(\pi_n)_{n \in \nset} \in \Pens^{ac}(\rset^d)^{\nset}$ such that
  $\lim_{n \to +\infty} \pi_n =\pi$, using the monotonicity in $\alpha$
  of $\Fun_\alpha^\eta$ and the result in point 1.  we have 
  \begin{equation*}
    \liminf_{n \to +\infty} \Fun_{\alpha_n}^{\eta_n}(\pi_n) \geq \liminf_{n \to +\infty} \Fun^{\eta_n}(\pi_n) \geq \Fun(\pi).
  \end{equation*}
  This guarantees condition~\ref{def:a} in Definition~\ref{def:gammas} holds since $\mathcal{Q}_2^{ac}(\rset^d)\subset \Pens^{ac}(\rset^d)$.

  To check condition~\ref{def:b} in Definition~\ref{def:gammas} consider for
  $\pi \in \mathcal{Q}_2^{ac}(\rset^d)$ and for any $n \in \nset$,
  $\pi_n \in \mathcal{Q}_2^{ac}(\rset^d)$ defined as in~\eqref{eq:pin}. The proof of point 1 shows that $\Fun^{\eta_n}(\pi_n)\to \Fun(\pi)$ as $n\to \infty$.

  We now show that $\lim_{n \to +\infty} \alpha_n\KL{\pi_n}{\pi_0} = 0$.
     Using Assumption~\ref{assum:pi0_second} there exists $C_1 \geq 0$ and $\tau > 0$ such that
   for any $x \in \rset^d$, $| U(x)|\leq C_1 + \tau \norm{x}^2$.
    Since $|\pi_n(x)|\leq \pi(x)/Z_1$ we have
    \begin{align}
    \label{eq:U_alpha}
        \int_{\rset^d}|U(x)|\pi_n(x)\rmd x &\leq \frac{C_1}{Z_1} + \tau \int_{\rset^d}\norm{x}^2\frac{\pi(x)}{Z_1}\rmd x <+\infty 
    \end{align}
    since $\pi\in \Pens_2^{ac}(\rset^d)$.
      Using~\eqref{eq:hent_bound} 
    we also have that
    \begin{align}
       |\pi_n(x)\log \pi_n(x)| &\leq \pi(x)\frac{|\log\pi(x)|+\log 1/Z_1}{Z_1}:=g(x) \label{eq:ent_alpha}
   \end{align}
   with $g(x)$ integrable since $\pi\in \mathcal{Q}_2^{ac}(\rset^d)$  implies $\Hent(\pi)<\infty$ by \cite[Theorem 1]{ghourchian2017existence}.

Combining~\eqref{eq:U_alpha} and~\eqref{eq:ent_alpha} we have
   $\lim_{n \to +\infty} \alpha_n \KL{\pi_n}{\pi_0} = 0$.
   
    Therefore, we have that
   $\Fun(\pi) = \lim_{n \to +\infty} \Fun_{\alpha_n}^{\eta_n}(\pi_n)$, and $\{\Fun_{\alpha}^\eta:\mathcal{Q}_2^{ac}(\rset^d)\to \rset^+,\ s.t.\ \alpha, \eta \geq 0\}$ is a
   $\Gamma$-system. 
   In addition, for any $\alpha, \eta>0$, $\Fun_\alpha^\eta$ admits a unique minimizer $\pi_{\alpha, \eta}^\star$ by Proposition~\ref{prop:convergence_minimum}, which under our assumptions, satisfies $\pi_{\alpha, \eta}^{\star}\in\mathcal{Q}_2^{ac}(\rset^d)$.
   Using Proposition~\ref{prop:gamma_cv_without}, we get that $\pi^\star$ is a
   minimizer of $\Fun$.
\end{enumerate}

\section{Subdifferential of $\Fun_\alpha^\eta$}
\label{app:subdifferential}

We start by recalling the definition of
subdifferentiability in Wasserstein spaces, see \cite[Definition
10.1.1]{ambrosio2008gradient}. We denote by $\Pens_{2}^{ac}(\rset^d)$ the space of
probability measures in $\Pens_2(\rset^d)$ which are absolutely continuous
w.r.t the Lebesgue measure and equip the space $\rmL^2(\rset^d,\pi):=\{ f:\rset^d\to\rset^d ; \pi(\Vert f\Vert^2) < \infty\}$ with the norm $\norm{f}^2_{\rmL^2(\rset^d, \pi)}:=\pi(\Vert f\Vert^2)$.
\begin{definition}[Fr\'echet subdifferential]
  Let $\Phi: \ \Pens_2(\rset^d) \to \rset$ and
  $\pi \in \Pens_{2}^{ac}(\rset^d)$, then
  $\xi \in \rmL^2(\rset^d, \pi)$ belongs to the strong Fr\'{e}chet
    subdifferential $\partial_{\rms} \Phi(\pi)$ of $\Phi$ at $\pi$ if for any  sequence $(t_n)_{n\geq 1} \in \rmL^2(\rset^d, \pi)$ such that $\norm{t_n - \Id}_{\rmL^2(\rset^d, \pi)}\to 0$  as $n \to \infty$ we have
    \begin{equation*}
        \liminf_{n \to \infty}\left. \left\lbrace\Phi(t_{n\#} \pi) - \Phi(\pi) - \int\limits_{\rset^d} \langle \xi(x), t_n(x) - x \rangle \rmd \pi(x)\right\rbrace\middle/ \norm{t_n - \Id}_{\rmL^2(\rset^d, \pi)} \right. \geq 0.
  \end{equation*}
\end{definition}
We are now ready to derive the subdifferential of $\Fun_\alpha^\eta$.
First, let us denote by $\Gun^\eta(\pi):=-\int\limits_{\rset^d} \log(\varphi(x)+\lambda\pi[\kker(x, \cdot)] + \eta) \rmd \pi(x)$, so that $\Fun^\eta(\pi)=-\Hent(\pi)+\Gun^\eta(\pi)$.
\begin{prop}[Subdifferential of $\Gun^\eta$]
\label{prop:subdifferential_middle}
Under Assumption~\ref{assum:general_kker}, $\Gun^\eta$ has strong subdifferential 
\begin{align*}
   \partial_{\rms}\Gun^\eta(\pi) = \left\lbrace x\to -\int\left[\frac{\lambda\nabla_2 \kker(z, x)}{\lambda\pi\left[ \kker(z, \cdot)\right]+\varphi(z)+\eta}+\frac{\lambda\nabla_1 \kker(x, z)+\nabla \varphi(x)}{\lambda\pi\left[ \kker(x, \cdot)\right]+\varphi(x)+\eta}\right]\rmd \pi\left(z\right)\right\rbrace. 
\end{align*}
\end{prop}
\begin{proof}
Let $\Lun: \ \Pens_2(\rset^d) \times \rset^d \to \left[0, +\infty\right)$ be given for
any $\pi \in \Pens_2(\rset^d)$ and $x \in \rset^d$ by
$\Lun(\pi, x) = \lambda\int\limits_{\rset^d} \kker(x, z) \rmd \pi(z) +\varphi(x)$,
and denote by $\xi_1(x, z) := \lambda\nabla_1\kker(x, z)  +\nabla\varphi(x)$ so that $\nabla \Lun(\pi, x) = \int\limits_{\rset^d} \xi_1(x, z) \rmd \pi(z) $ for all fixed $\pi \in\Pens(\rset^d)$, where Leibniz integral rule (e.g. \cite[Theorem 16.8]{billingsley1995measure}) guarantees that we can swap the derivative w.r.t. $x$ with the integral w.r.t. $z$ since $\norm{\nabla_1\kker}\leq (1+\lambda)\Mtt$ and $\kker$ is bounded and thus integrable.

Further, define  $\xi_2(x, z) :=\lambda \nabla_2 \kker(x, z)$, let $g : \ \left[0, +\infty\right) \to \rset$ be given for any $t \geq 0$ by $g(t) = -\log(t + \eta)$ and consider
\begin{align}
\label{eq:subdifferential_decomposition}
& \Gun^\eta(t_{n\#} \pi) -\Gun^\eta(\pi)  \\
&- \int\limits_{\rset^d} \left\langle \int\limits_{\rset^d} \left[g'(\Lun(\pi, z)) \xi_2(z, x) + g'(\Lun(\pi, x)) \xi_1(x, z)\right]\rmd \pi(z) , t_n(x) - x \right\rangle \rmd \pi(x)\notag\\
=& \int\limits_{\rset^d}g(\Lun(t_{n\#} \pi, x))\rmd t_{n\#}\pi( x) - \int\limits_{\rset^d}g(\Lun(\pi, x))\rmd t_{n\#}\pi( x)\notag\\
&- \int\limits_{\rset^d} \left\langle \int\limits_{\rset^d} g'(\Lun(\pi, z)) \xi_2(z, x) \rmd \pi(z) , t_n(x) - x \right\rangle \rmd \pi(x)\notag\\
    &+\int\limits_{\rset^d}g(\Lun(\pi, x))\rmd t_{n\#}\pi( x) - \int\limits_{\rset^d}g(\Lun(\pi, x))\rmd\pi( x)\notag\\
    &- \int\limits_{\rset^d} \left\langle \int\limits_{\rset^d}  g'(\Lun(\pi, x)) \xi_1(x, z)\rmd \pi(z) , t_n(x) - x \right\rangle \rmd \pi(x).\notag
\end{align}
We consider two groups of terms separately.

We further decompose the first group of three terms in~\eqref{eq:subdifferential_decomposition} into

\resizebox{0.95\textwidth}{!}{\noindent
\hspace*{-0.5cm}\parbox{\textwidth}{\noindent
\begin{align}
\label{eq:subdifferential_decomposition1}
    & \int\limits_{\rset^d}\!\!\left[g(\Lun(t_{n\#} \pi, x)) - g(\Lun(\pi, x)) - g'(\Lun(\pi, x)) \!\left(\int\limits_{\rset^d}\! \langle \xi_2(x, z), t_n(z) - z \rangle \rmd \pi(z) \right)\!\right]\!\rmd t_{n\#}\pi( x)\\
    & + \int\limits_{\rset^d} \int\limits_{\rset^d} \langle g'(\Lun(\pi, x))\xi_2(x, z), t_n(z) - z \rangle \rmd \pi(z) \rmd t_{n\#}\pi( x)\notag\\
    & - \int\limits_{\rset^d}\int\limits_{\rset^d}  \left\langle g'(\Lun(\pi, z)) \xi_2(z, x)  , t_n(x) - x \right\rangle \rmd \pi(z)\rmd \pi(x)\notag\\
    & \geq \int\limits_{\rset^d}g'(\Lun(\pi, x)) \left(\Lun(t_{n\#} \pi, x) - \Lun(\pi, x) - \int\limits_{\rset^d} \langle \xi_2(x, z), t_n(z) - y \rangle \rmd \pi(z) \right)\rmd t_{n\#}\pi( x)\notag\\
    & \quad + \int\limits_{\rset^d} \int\limits_{\rset^d} \langle g'(\Lun(\pi, x))\xi_2(x, z), t_n(z) - z \rangle \rmd \pi(z) \rmd t_{n\#}\pi( x)\notag\\
    & \quad - \int\limits_{\rset^d}\int\limits_{\rset^d}  \left\langle g'(\Lun(\pi, z)) \xi_2(z, x)  , t_n(x) - x \right\rangle \rmd \pi(z)\rmd \pi(x)\notag
\end{align}
}}

where the inequality follows since the convexity of $g$ implies
\begin{align*}
&g(\Lun(t_{n\#} \pi, x)) - g(\Lun(\pi, x)) - g'(\Lun(\pi, x)) \int\limits_{\rset^d} \langle \xi_2(x,z), t_n(z) - z \rangle \rmd \pi(z) \\
&\qquad\geq g'(\Lun(\pi, x)) \left(\Lun(t_{n\#} \pi, x) - \Lun(\pi, x) - \int\limits_{\rset^d} \langle \xi_2(x, z), t_n(z) - z \rangle \rmd \pi(z) \right).
\end{align*}
For the second group of terms in~\eqref{eq:subdifferential_decomposition}, using Assumption \ref{assum:general_kker} and that $\Lun(\pi, x)$ is positive for all $\pi\in\Pens(\rset^d), x\in\rset^d$, we have, using the spectral norm induced by the Euclidean norm,
\begin{align}
\label{eq:gsecond_bounded}
&\norm{\nabla^2 g(\Lun(\pi, x))} \\
\leq&\norm{ \frac{(\Lun(\pi, x)+\eta)\left(\lambda\int\limits_{\rset^d}\nabla_1^2\kker(x, z)\rmd \pi(z)+\nabla^2\varphi(x)\right)}{(\Lun(\pi, x)+\eta)^2}}\notag\\
    &+\norm{\frac{(\lambda\int\limits_{\rset^d}\nabla_1\kker(x, z)\rmd \pi(z)+\nabla\varphi(x))(\lambda\int\limits_{\rset^d}\nabla_1\kker(x, z)\rmd \pi(z)+\nabla\varphi(x))^\top}{(\Lun(\pi, x)+\eta)^2}}\notag\\
    \leq& \frac{((\lambda+1)\Mtt+\eta)(\lambda+1)\Mtt + (\lambda+1)^2\Mtt^2}{\eta^2}=:\Ctt_0,\notag
\end{align}
where $\nabla\int\limits_{\rset^d}\nabla_1 \kker(x, z)\rmd \pi(z)=\int\limits_{\rset^d}\nabla^2_1\kker(x, z)\rmd \pi(z)$ using Leibniz integral rule for differentiation under the
integral sign (e.g. \cite[Theorem 16.8]{billingsley1995measure}) since $\norm{\nabla_1^2\kker}\leq \Mtt$ and $\xi_1$ is bounded and thus integrable.
Using that $\norm{\nabla^2 g(\Lun(\pi, x))} \leq \Ctt_0$ we get that, $ \forall x_1, x_2 \in \rset^d$
\begin{equation*}
g(\Lun(\pi, x_2)) \geq g(\Lun(\pi, x_1)) + \langle g'(\Lun(\pi, x_1))\!\!\int\limits_{\rset^d}\!\!   \xi_1(x_1, z)\rmd \pi(z), x_2 -x_1 \rangle - \Ctt_0 \norm{x_1 - x_2}^2,
\end{equation*}
from which follows

\resizebox{0.95\textwidth}{!}{\noindent
\hspace*{-0.5cm}\parbox{\textwidth}{\noindent
\begin{align}
\label{eq:subdifferential_decomposition2}
&\int\limits_{\rset^d}g(\Lun(\pi, x))\rmd (t_{n\#}\pi-\pi)( x) -  \int\limits_{\rset^d} \left\langle g'(\Lun(\pi, x)) \int\limits_{\rset^d}  \xi_1(x, z)\rmd \pi(z) , t_n(x) - x \right\rangle \rmd \pi(x)\\
&\geq -\Ctt_0 \norm{t_n - \Id}^2_{\rmL^2(\rset^d, \pi)}.\notag
\end{align}
}}

Combining~\eqref{eq:subdifferential_decomposition}--\eqref{eq:subdifferential_decomposition2} we obtain

\resizebox{0.95\textwidth}{!}{\noindent
\hspace*{-0.5cm}\parbox{\textwidth}{\noindent
\begin{align}
\label{eq:subdifferential_last_term}
& \Psi_n(\pi):=\Gun^\eta(t_{n\#} \pi) -\Gun^\eta(\pi)  \\
& - \int\limits_{\rset^d} \left\langle \int\limits_{\rset^d} \left[g'(\Lun(\pi, z)) \xi_2(z, x) + g'(\Lun(\pi, x)) \xi_1(x, z)\right]\rmd \pi(z) , t_n(x) - x \right\rangle \rmd \pi(x) \notag\\
&\; \geq \int\limits_{\rset^d}g'(\Lun(\pi, x)) \left(\Lun(t_{n\#} \pi, x) - \Lun(\pi, x) - \int\limits_{\rset^d} \langle \xi_2(x, z), t_n(z) - z \rangle \rmd \pi(z) \right)\rmd t_{n\#}\pi( x)\notag\\
&\;\phantom{\geq} + \int\limits_{\rset^d} \int\limits_{\rset^d} \langle g'(\Lun(\pi, x))\xi_2(x, z), t_n(z) - z \rangle \rmd \pi(z) \rmd t_{n\#}\pi( x)\notag\\
&\;\phantom{\geq} - \int\limits_{\rset^d}\int\limits_{\rset^d}  \left\langle g'(\Lun(\pi, z)) \xi_2(z, x)  , t_n(x) - x \right\rangle \rmd \pi(z)\rmd \pi(x)\notag\\
&\;\phantom{\geq} -\Ctt_0 \norm{t_n - \Id}^2_{\rmL^2(\rset^d, \pi)}.\notag
\end{align}
}}

We now show that 
\begin{align*}
    \liminf_{n\to\infty}\frac{\Psi_n(\pi)}{\norm{t_n - \Id}_{\rmL^2(\rset^d, \pi)}} \geq 0.
\end{align*}

For the last term in~\eqref{eq:subdifferential_last_term}, we use the definition of $t_n$: as $n\to \infty$ we have
\begin{align*}
     -\Ctt_0\frac{\norm{t_n - \Id}^2_{\rmL^2(\rset^d, \pi)}}{\norm{t_n - \Id}_{\rmL^2(\rset^d, \pi)}}=-\Ctt_0\norm{t_n - \Id}_{\rmL^2(\rset^d, \pi)} \to 0.
\end{align*}

Let us consider then
\begin{align*}
&\int\limits_{\rset^d} \int\limits_{\rset^d} \langle g'(\Lun(\pi, x))\xi_2(x, z), t_n(z) - z \rangle \rmd \pi(z) \rmd t_{n\#}\pi( x)\\
&\qquad- \int\limits_{\rset^d}\int\limits_{\rset^d}  \left\langle g'(\Lun(\pi, z)) \xi_2(z, x)  , t_n(x) - x \right\rangle \rmd \pi(z)\rmd \pi(x)\\
=& \int\limits_{\rset^d} \int\limits_{\rset^d} \langle g'(\Lun(\pi, x))\xi_2(x, z), t_n(z) - z \rangle \rmd \pi(z) \rmd t_{n\#}\pi( x)\\
&\qquad- \int\limits_{\rset^d}\int\limits_{\rset^d}  \left\langle g'(\Lun(\pi, x)) \xi_2(x, z)  , t_n(z) - z \right\rangle \rmd \pi(x)\rmd \pi(z)\\
=& \int\limits_{\rset^d} \int\limits_{\rset^d} \left\langle g'(\Lun(\pi, x))\xi_2(x, z)\rmd (t_{n\#}\pi - \pi)( x), t_n(z) - z \right\rangle \rmd \pi(z).
\end{align*}
Under Assumption~\ref{assum:general_kker}, the function $x\mapsto g'(\Lun(\pi, x)) \xi_2(x, z)$ is continuous and bounded for any fixed $z$; in addition, $n \to \infty$ implies $t_{n\#}\pi \to \pi$ weakly, since $\wassersteinD[2]$ metrizes weak convergence. It follows that
\begin{align*}
\int\limits_{\rset^d} g'(\Lun(\pi, x))\xi_2(x, z)\rmd t_{n\#}\pi( x) - \int\limits_{\rset^d} g'(\Lun(\pi, x)) \xi_2(x, z) \rmd \pi(x) \to 0.
\end{align*}
The dominated convergence theorem then guarantees that, as $n \to \infty$:
\begin{align*}
\int\limits_{\rset^d} \left\langle \int\limits_{\rset^d} g'(\Lun(\pi, x))\xi_2(x, z)\rmd (t_{n\#}\pi - \pi)( x), t_n(z) - z \right\rangle \rmd \pi(z) \to 0.
\end{align*}

Under Assumption~\ref{assum:general_kker},
\begin{align*}
    \int\limits_{\rset^d} g'(\Lun(\pi, x))\xi_2(x, z)\rmd (t_{n\#}\pi - \pi)( x) \leq \eta^{-1}\lambda\Mtt \wassersteinD[1](t_{n\#}\pi, \pi)
\end{align*}
using the dual representation of $\wassersteinD[1]$.
Then, using the Cauchy-Schwarz inequality,
\begin{align*}
    &\frac{\int\limits_{\rset^d}  \left\langle \int\limits_{\rset^d} g'(\Lun(\pi, x))\xi_2(x, z)\rmd (t_{n\#}\pi - \pi)( x), t_n(z) - z \right\rangle \rmd \pi(z)}{\norm{t_n - \Id}_{\rmL^2(\rset^d, \pi)}} \\
    &\leq \eta^{-1}\lambda\Mtt \wassersteinD[1](t_{n\#}\pi, \pi)\frac{\int\limits_{\rset^d} \norm{t_n(z) - z}  \rmd \pi(z)}{\norm{t_n - \Id}_{\rmL^2(\rset^d, \pi)}}\\
    &\leq \eta^{-1}\lambda\Mtt \wassersteinD[1](t_{n\#}\pi, \pi)\to 0
\end{align*}
by Jensen's inequality.

For the remaining term in~\eqref{eq:subdifferential_last_term}, using that $\norm{\nabla^2 \kker} \leq \Mtt$ we get that:
\begin{equation*}
  \forall z_1, z_2 \in \rset^d, x\in\rset^d: \qquad
\kker(x, z_2) \geq \kker(x, z_1) + \langle \nabla_2 \kker(x, z_1), z_2 -z_1 \rangle - \Mtt \norm{z_1 - z_2}^2.
\end{equation*}
Combining this and \cite[Proposition
10.4.2]{ambrosio2008gradient} we get that for any $x \in \rset^d$ and
$\pi \in \Pens_2(\rset^d)$,
the strong subdifferential of $ \Lun(\pi, x)$ is given by $\xi_2(x, z) =\lambda \nabla_2 \kker(x, z)$, thus
\begin{equation*}
\label{eq:fn_convergence}
\liminf_{n \to \infty} \{\Lun(t_{n\#} \pi, x) - \Lun(\pi, x) - \int\limits_{\rset^d} \langle \xi_2(x, z), t_n(z) - z \rangle \rmd \pi(z) \} / \norm{t_n - \Id}_{\rmL^2(\rset^d, \pi)}  \geq 0.
\end{equation*}
Similarly, we also have that $-\lambda \nabla_2 \kker(x, z)$ is the strong subdifferential of $ -\Lun(\pi, x)$ and thus
\begin{align*}
    \liminf_{n \to \infty} \{-\Lun(t_{n\#} \pi, x) + \Lun(\pi, x) + \int\limits_{\rset^d} \langle \xi_2(x, z), t_n(z) - z \rangle \rmd \pi(z) \} / \norm{t_n - \Id}_{\rmL^2(\rset^d, \pi)}  \geq 0.
\end{align*}
Or, equivalently, 
\begin{align*}
    \limsup_{n \to \infty} \{\Lun(t_{n\#} \pi, x) - \Lun(\pi, x) - \int\limits_{\rset^d} \langle \xi_2(x, z), t_n(z) - z \rangle \rmd \pi(z) \} / \norm{t_n - \Id}_{\rmL^2(\rset^d, \pi)}  \leq 0.
\end{align*}
Thus 
\begin{align}
\label{eq:liminf_limsup}
    \lim_{n \to \infty} \{\Lun(t_{n\#} \pi, x) - \Lun(\pi, x) - \int\limits_{\rset^d} \langle \xi_2(x, z), t_n(z) - z \rangle \rmd \pi(z) \} / \norm{t_n - \Id}_{\rmL^2(\rset^d, \pi)}  = 0.
\end{align}

In addition, since $\kker$ is Lipschitz continuous with constant $\Mtt$ we get that for any $\pi_1, \pi_2 \in \Pens_2(\rset^d)$ and $x \in \rset^d$,
$\vert\Lun(\pi_1, x) - \Lun(\pi_2, x)
\vert\leq \lambda\Mtt \wassersteinD[1](\pi_1, \pi_2),$
and therefore
\begin{align}
\label{eq:upper_bound_xi2}
&\left\vert\Lun(t_{n\#} \pi, x) - \Lun(\pi, x) - \int\limits_{\rset^d} \langle \xi_2(x, z), t_n(z) - z \rangle \rmd \pi(z) \right\vert / \norm{t_n - \Id}_{\rmL^2(\rset^d, \pi)} \\
&\qquad\qquad\qquad \leq \lambda\Mtt + \int\limits_{\rset^d} \norm{\xi_2(x, z)} \rmd \pi(z)\notag
\leq 2\lambda\Mtt.\notag
\end{align}
In particular, we have the following lower bound
\begin{align*}
    &g'(\Lun(\pi, x)) \left(\Lun(t_{n\#} \pi, x) - \Lun(\pi, x) - \int\limits_{\rset^d} \langle \xi_2(x, z), t_n(z) - z \rangle \rmd \pi(z) \right)\\
    &\qquad\qquad\geq -\eta^{-1}2\lambda\Mtt\norm{t_n - \Id}_{\rmL^2(\rset^d, \pi)} > -\infty.
\end{align*}

Let us consider the family of functions
\begin{align}
\label{eq:family_fun}
    x \mapsto f_n(x):=g'(\Lun(\pi, x))h_n(x),
\end{align}
where we defined
\begin{align*}
    h_n(x):=  \frac{\Lun(t_{n\#} \pi, x) - \Lun(\pi, x) - \int\limits_{\rset^d} \langle \xi_2(x, z), t_n(z) - z \rangle \rmd \pi(z) }{\norm{t_n - \Id}_{\rmL^2(\rset^d, \pi)}}.
\end{align*}
Consider the second derivative of~\eqref{eq:family_fun}
\begin{align*}
    \nabla f_n(x) = g'(\Lun(\pi, x)) \nabla h_n(x) + h_n(x)g''(\Lun(\pi, x)) \nabla \Lun(\pi, x).
\end{align*}
Then, $g'(\Lun(\pi, x))$ is upper bounded by $\eta^{-1}$, $g''(\Lun(\pi, x))$ is upper bounded by $\eta^{-2}$ , $h_n$ is upper bounded by~\eqref{eq:upper_bound_xi2}. Moreover,
\begin{align*}
    \norm{ \nabla \Lun(\pi, x)}=\norm{\int\limits_{\rset^d} \xi_1(x, z) \rmd \pi(z)}\leq \Mtt(1+\lambda)
\end{align*}
and, using the Lipschitz continuity of $\nabla_1 \kker$ guaranteed by Assumption~\ref{assum:general_kker} and the Cauchy--Schwarz inequality,
\begin{align*}
   \norm{ \nabla h_n(x) }&= \frac{\norm{\lambda\int_{\rset^d}\nabla_1\kker(x, z)\rmd (t_{n\#}\pi- \pi)(z)- \int\limits_{\rset^d} \langle\nabla_1 \xi_2(x, z), t_n(z) - z \rangle \rmd \pi(z)} }{\norm{t_n - \Id}_{\rmL^2(\rset^d, \pi)}}\\
   &\leq \frac{\lambda\int_{\rset^d}\norm{\nabla_1\kker(x, t_n(z))-\nabla_1\kker(x, z)}\rmd \pi(z)} {\norm{t_n - \Id}_{\rmL^2(\rset^d, \pi)}}\\
   &+\frac{\int\limits_{\rset^d} \norm{\nabla_1 \xi_2(x, z)}\norm{ t_n(z) - z }\rmd \pi(z)}{\norm{t_n - \Id}_{\rmL^2(\rset^d, \pi)}}\\
   &\leq 2\lambda \Mtt\frac{\left(\int_{\rset^d}\norm{ t_n(z)-z}^2\rmd \pi(z)\right)^{1/2}} {\norm{t_n - \Id}_{\rmL^2(\rset^d, \pi)}}\\
   &\leq 2\lambda\Mtt,
\end{align*}
where we used $\xi_2(x, z) =\lambda \nabla_2 \kker(x, z)$, Assumption~\ref{assum:general_kker} and Jensen's inequality in the second-to-last inequality.

Hence, the sequence of functions is equicontinuous and uniformly bounded and we can apply the dominated convergence Theorem for weakly converging measures \cite[Theorem 5.2]{feinberg2020fatou2} to obtain 

\resizebox{0.95\textwidth}{!}{\noindent
\hspace*{-0.5cm}\parbox{\textwidth}{\noindent
\begin{align*}
   \lim_{n \to \infty}\frac{\int\limits_{\rset^d}\! g'(\Lun(\pi, x)) \!\left(\!\Lun(t_{n\#} \pi, x) \!-\! \Lun(\pi, x) \!-\! \int\limits_{\rset^d}\! \langle \xi_2(x, z), t_n(z) \!-\! z \rangle \rmd \pi(z) \!\right)\!\rmd t_{n\#}\pi( x)}{ \norm{t_n - \Id}_{\rmL^2(\rset^d, \pi)} } = 0,
\end{align*}
}}

which gives the result.
\end{proof}

\cite[Theorem 10.4.9]{ambrosio2008gradient} shows that the subdifferential of $\Mun(\pi) =
\KL{\pi}{\pi_0}$ is 
\begin{align}
\label{eq:subdifferential_pi0}
\partial_{\rms}\Mun(\pi) = \{x \mapsto \nabla \log(\rmd \pi / \rmd \pi_0)(x)\}.
\end{align}
Similarly, the subdifferential of $\Hent(\pi) :=-\int\limits_{\rset^d}\log(\pi(x)) \rmd \pi(x)$ is 
\begin{equation}
\label{eq:subdifferential entropy}
\partial_{\rms}\Hent(\pi) = \{x \mapsto -\nabla \log\pi(x)\}.
\end{equation}
Combining the above with Proposition~\ref{prop:subdifferential_middle} we obtain the subdifferential for $\Fun_\alpha^\eta$
\begin{align*}
\partial_{\rms}\Fun_\alpha^\eta(\pi)  &=  \left\lbrace x\to -\int\left[\frac{\lambda\nabla_2 \kker(z, x)}{\lambda\pi\left[ \kker(z, \cdot)\right]+\varphi(z)+\eta}+\frac{\lambda\nabla_1 \kker(x, z)+\nabla \varphi(x)}{\lambda\pi\left[ \kker(x, \cdot)\right]+\varphi(x)+\eta}\right]\rmd \pi\left(z\right)\right.\\
&\left.\qquad\qquad+ \nabla \log\pi(x)+\alpha\nabla \log(\rmd \pi / \rmd \pi_0)(x)\right\rbrace.
\end{align*}

We can then apply the definition of a Wasserstein gradient flow \cite[Definition 11.1.1]{ambrosio2008gradient} to establish that if there exists $(\pi_t)_{t \geq 0}$ such that for any $t \in (0, +\infty)$, $\pi_t$ admits a density w.r.t. the Lebesgue measure and~\eqref{eq:PDE_BM} holds in the sense of distributions then $(\pi_t)_{t \geq 0}$ is a Wasserstein gradient flow associated with $\Fun_{\alpha}^\eta$.

%% file: simods_mkvsde.tex
\section{Proof on MKVSDE}

\subsection{Proof of Proposition~\ref{prop:existence_uniqueness}}
\label{app:proof_eu}

To prove  existence and uniqueness of the solution of~\eqref{eq:mckean_sde} we use standard tools for McKean--Vlasov processes. In particular, we only need to show that the drift of~\eqref{eq:mckean_sde}, denoted by $b:\rset^d\times\Pens_1(\rset^d) \to \rset^d$, given for any $x \in \rset^d$ and $\pi \in \Pens_1(\rset^d)$ by~\eqref{eq:b}
is Lipschitz continuous.
The remainder of the proof is classical and is omitted, see for instance \cite[Theorem  1.1]{sznitman1991topics} or \cite[Appendix B.2]{crucinio2022solving} for a recent result on a similar scheme.

Let $\alpha, \eta > 0$. First, we show that $b^\eta$ in~\eqref{eq:drift} is Lipschitz continuous. Under Assumption~\ref{assum:general_kker}, the forcing term $\varphi$, $\kker$ and their gradients are Lipschitz continuous, hence, we have for any $x_1, x_2\in \rset^d$, $z\in \rset^d$ and $\pi_1, \pi_2 \in \Pens_1(\rset^d)$
\begin{align*}
\norm{b^\eta(x_1, z, \pi_1) - b^\eta(x_2, z,\pi_2)} \hspace*{-2cm}&\\
\leq& \norm{\frac{\lambda\nabla_2 \kker(z, x_1)}{\lambda\pi_1\left[ \kker(z, \cdot)\right]+\varphi(z)+\eta}-\frac{\lambda\nabla_2 \kker(z, x_2)}{\lambda\pi_2\left[ \kker(z, \cdot)\right]+\varphi(z)+\eta}}\\
&+\norm{\frac{\lambda\nabla_1 \kker(x_1, z)+\nabla \varphi(x_1)}{\lambda\pi_1\left[ \kker(x_1, \cdot)\right]+\varphi(x_1)+\eta} - \frac{\lambda\nabla_1 \kker(x_2, z)+\nabla \varphi(x_2)}{\lambda\pi_2\left[ \kker(x_2, \cdot)\right]+\varphi(x_2)+\eta}}\notag\\
\leq& \norm{\frac{\lambda\nabla_2 \kker(z, x_1) - \lambda\nabla_2 \kker(z, x_2)}{\lambda\pi_1\left[ \kker(z, \cdot)\right]+\varphi(z)+\eta}}\notag\\
&+\norm{\frac{\lambda\nabla_2 \kker(z, x_2)}{\lambda\pi_1\left[ \kker(z, \cdot)\right]+\varphi(z)+\eta}-\frac{\lambda\nabla_2 \kker(z, x_2)}{\lambda\pi_2\left[ \kker(z, \cdot)\right]+\varphi(z)+\eta}}\notag\\
&+\norm{\frac{\lambda\nabla_1 \kker(x_1, z)+\nabla \varphi(x_1) - \lambda\nabla_1 \kker(x_2, z)-\nabla \varphi(x_2)}{\lambda\pi_1\left[ \kker(x_1, \cdot)\right]+\varphi(x_1)+\eta}}\notag\\
&+\norm{\frac{\lambda\nabla_1 \kker(x_2, z)+\nabla \varphi(x_2)}{\lambda\pi_1\left[ \kker(x_1, \cdot)\right]+\varphi(x_1)+\eta}-\frac{\lambda\nabla_1 \kker(x_2, z)+\nabla \varphi(x_2)}{\lambda\pi_2\left[ \kker(x_2, \cdot)\right]+\varphi(x_2)+\eta}}\notag\\
\leq& \lambda(\Mtt/\eta)\norm{x_1-x_2}+\lambda^2(\Mtt/\eta^2)\vert \pi_2\left[ \kker(z, \cdot)\right]-\pi_1\left[ \kker(z, \cdot)\right]\vert\notag\\
&+(\lambda+1)(\Mtt/\eta)\norm{x_1-x_2}\notag \\
&+ (\lambda+1)(\Mtt/\eta^2)(\Mtt\norm{x_1-x_2}+\vert \pi_2\left[ \kker(x_2, \cdot)\right]-\pi_1\left[ \kker(x_1, \cdot)\right]\vert).\notag
\end{align*}
Using the Lipschitz continuity of $\kker$ we also have that for any $z\in \rset^d$
\begin{align*}\hfill
    \vert \pi_2\left[ \kker(z, \cdot)\right]-\pi_1\left[ \kker(z, \cdot)\right]\vert  & \leq  \Mtt \wassersteinD[1](\pi_1, \pi_2).
\end{align*}
Hence, we have
\begin{align}
\label{eq:bound_beta_one}
    &\norm{b^\eta(x_1, z, \pi_1) - b^\eta(x_2, z,\pi_2)} \\
    &\qquad\qquad\leq  \Ctt_1 ( \norm{x_1-x_2} +\vert \pi_2\left[ \kker(x_2, \cdot)\right]-\pi_1\left[ \kker(x_1, \cdot)\right]\vert+\wassersteinD[1](\pi_1, \pi_2)),  \notag  
\end{align}
with $\Ctt_1$ depending upon $\lambda, \eta$ and $\Mtt$. Following a similar approach, we can show that $b^\eta$ is Lipschitz continuous in $z$: 
for any $z_1, z_2\in \rset^d$, $x\in \rset^d$ and $\nu \in \Pens_1(\rset^d)$
\begin{align*}
&\norm{b^\eta(x, z_1, \nu) - b^\eta(x, z_2,\nu)}\\
&\qquad\qquad\leq \norm{\frac{\lambda\nabla_2 \kker(z_1, x)}{\lambda\nu\left[ \kker(z_1, \cdot)\right]+\varphi(z_1)+\eta}-\frac{\lambda\nabla_2 \kker(z_2, x)}{\lambda\nu\left[ \kker(z_2, \cdot)\right]+\varphi(z_2)+\eta}}\\
&\qquad\qquad+\norm{\frac{\lambda\nabla_1 \kker(x, z_1)+\nabla \varphi(x)}{\lambda\nu\left[ \kker(x, \cdot)\right]+\varphi(x)+\eta} - \frac{\lambda\nabla_1 \kker(x, z_2)+\nabla \varphi(x)}{\lambda\nu\left[ \kker(x, \cdot)\right]+\varphi(x)+\eta}}\notag\\
&\qquad\qquad\leq  \norm{\frac{\lambda\nabla_2 \kker(z_1, x)-\lambda\nabla_2 \kker(z_2, x)}{\lambda\nu\left[ \kker(z_1, \cdot)\right]+\varphi(z_1)+\eta}}\notag\\
&\qquad\qquad+\norm{\frac{\lambda\nabla_2 \kker(z_2, x)}{\lambda\nu\left[ \kker(z_1, \cdot)\right]+\varphi(z_1)+\eta}-\frac{\lambda\nabla_2 \kker(z_2, x)}{\lambda\nu\left[ \kker(z_2, \cdot)\right]+\varphi(z_2)+\eta}}\notag\\
&\qquad\qquad+\lambda(\Mtt/\eta)\norm{z_1-z_2}\notag\\
&\qquad\qquad\leq 2\lambda(\Mtt/\eta)\norm{z_1-z_2} \notag\\
&\qquad\qquad+ \lambda(\Mtt/\eta^2)\vert\lambda\nu\left[ \kker(z_2, \cdot)\right]+\varphi(z_2)-\lambda\nu\left[ \kker(z_1, \cdot)\right]-\varphi(z_1)\vert\notag\\
&\qquad\qquad\leq 2\lambda(\Mtt/\eta)\norm{z_1-z_2}+ \lambda(\Mtt^2/\eta^2)(1+\lambda)\norm{z_1-z_2} \notag\\
&\qquad\qquad\leq \Ctt_2\norm{z_1-z_2},\notag
\end{align*}
for $\Ctt_2$ which depends on $\lambda, \eta$ and $\Mtt$. 

Using~\eqref{eq:bound_beta_one}, for any $x_1, x_2 \in \rset^d$ and $\pi_1, \pi_2 \in \Pens_1(\rset^d)$ we have
\begin{align}
\label{eq:bound_beta_two}
    &\norm{\int b^\eta(x_1, z, \pi_1)\rmd \pi_1(z) - \int b^\eta(x_2, z, \pi_2)\rmd \pi_2(z)} \\
    &\qquad\leq  \int \norm{b^\eta(x_1, z, \pi_1) - b^\eta(x_2, z,\pi_2)}\rmd \pi_1(z)+\vert \int b^\eta(x_2, z,\pi_2) \rmd(\pi_1-\pi_2)(z)\vert\notag\\
    &\qquad\leq \Ctt_1 ( \norm{x_1-x_2} +\vert \pi_2\left[ \kker(x_2, \cdot)\right]-\pi_1\left[ \kker(x_1, \cdot)\right]\vert+\wassersteinD[1](\pi_1, \pi_2))\notag \\
    &\qquad+ \vert \pi_2\left[ b^\eta(x_2, \cdot,\pi_2)\right]-\pi_1\left[b^\eta(x_2, \cdot,\pi_2)\right]\vert. \notag
\end{align}
Using the dual representation of $\wassersteinD[1]$ and the Lipschitz continuity of $\kker$ we also have that for any $x_1, x_2\in \rset^d$
\begin{align*}\hfill
    \vert \pi_2\left[ \kker(x_2, \cdot)\right]-\pi_1\left[ \kker(x_1, \cdot)\right]\vert  & \leq  \vert \pi_2\left[ \kker(x_2, \cdot)\right]-\pi_1\left[ \kker(x_2, \cdot)\right]\vert \notag\\
    &+ \vert \pi_1\left[ \kker(x_2, \cdot)\right]-\pi_1\left[ \kker(x_1, \cdot)\right]\vert\notag\\
   &\leq \Mtt \wassersteinD[1](\pi_1, \pi_2) + \Mtt\norm{x_1-x_2},\notag
\end{align*}
and
\begin{align}
\label{eq:w1_dual_beta}
\vert (\pi_2 - \pi_1)\left[ b^\eta(x_2, \cdot,\pi_2)\right]\vert &\leq \Ctt_2 \wassersteinD[1](\pi_1, \pi_2).
\end{align}
Using the above, for any $x_1, x_2 \in \rset^d$ and $\pi_1, \pi_2 \in \Pens_1(\rset^d)$ we have
\begin{align}
\label{eq:part1_drift_lipschitz}
    &\norm{\int b^\eta(x_1, z, \pi_1)\rmd \pi_1(z) - \int b^\eta(x_2, z, \pi_2)\rmd \pi_2(z)} \\
    &\qquad\leq  \Ctt_1(1+\Mtt)( \norm{x_1-x_2} +\wassersteinD[1](\pi_1, \pi_2))+\Ctt_2\wassersteinD[1](\pi_1, \pi_2).\notag
\end{align}
Using Assumption~\ref{assum:pi0}, it then follows that
\begin{align}
\label{eq:drift_lipschitz}
\norm{b(x_1, \pi_1) - b(x_2, \pi_2)} &\leq \Ctt_1(1+\Mtt) ( \norm{x_1-x_2} +\wassersteinD[1](\pi_1, \pi_2))\\
&\phantom{\leq} \quad + \Ctt_2\wassersteinD[1](\pi_1, \pi_2)+\alpha\Ltt \norm{x_1-x_2} \notag\\
    &\leq (\Ctt_1+\Ctt_1\Mtt+\Ctt_2+\alpha\Ltt) ( \norm{x_1-x_2} +\wassersteinD[1](\pi_1, \pi_2)).\notag
\end{align}

\subsection{Proof of Proposition~\ref{prop:la_convergence_star}}
\label{app:proof_invariant}

Under Assumption~\ref{assum:general_kker}, Proposition~\ref{prop:convergence_minimum}--\ref{item:a} ensures that the functional $\Fun^\eta$ is convex and lower bounded.
In addition, Assumption~\ref{assum:pi0} guarantees that $U$ is sufficiently regular to satisfy \cite[Assumption 2.2]{hu2019mean}.
To see this, we can use Assumption~\ref{assum:pi0}--\ref{item:dissipativity} and the Cauchy--Schwarz inequality to obtain
\begin{align*}
    \langle \nabla U(x) , x \rangle & = \langle \nabla U(x) -\nabla U(0), x \rangle + \langle \nabla U(0) , x \rangle\geq \mtt\norm{x}^2-\ctt - \norm{x}\norm{\nabla U(0)}.
\end{align*}
Then, using Young's inequality, $\alpha\beta \leq \alpha^2/(2\epsilon)+\beta^2\epsilon/2$ for all $\epsilon>0$, 
\begin{align*}
    -\norm{x}\norm{\nabla U(0)} \geq -\norm{x}^2/(2\epsilon)-\epsilon \norm{\nabla U(0)}^2/2.
\end{align*}
It follows that for all $\epsilon > 1/(2\mtt)$ we have, for any $x \in \rset^d$, $\langle \nabla U(x) , x \rangle \geq \mathtt{a} \norm{x}^2 + \mathtt{b}$, where $\mathtt{a}:= \mtt-1/(2\epsilon) >0, \mathtt{b}=-\ctt-\epsilon \norm{\nabla U(0)}^2/2$.

The intrinsic derivative of $\Fun^\eta$ is given by the first component of the drift $b$ in~\eqref{eq:b}:
\begin{align*}
   D_\nu \Fun^\eta(x, \nu):=\int b^\eta(x, z, \nu) \rmd \nu(z)
\end{align*}
with $b^\eta$ in~\eqref{eq:drift}.
As shown in~\eqref{eq:part1_drift_lipschitz}, $D_\nu \Fun^\eta$ is Lipschitz continuous and thus $\Fun^\eta$ is continuously differentiable.
In addition, using the fact that under Assumption~\ref{assum:general_kker} $\varphi, \kker\in\rmc^\infty(\rset^d \times \rset^d, [0, +\infty))$ and Leibniz integral rule for differentiation under the
integral sign (e.g. \cite[Theorem 16.8]{billingsley1995measure}), we have that $D_\nu \Fun^\eta$ is $\rmc^\infty(\rset^d , \rset^d)$ for all fixed $\nu\in\Pens(\rset^d)$.
Finally, we need to show that
\begin{align*}
\nabla D_\nu \Fun^\eta(x, \nu) &= 
\int\limits_{\rset^d} \nabla_1 b^\eta(x, z, \nu) \rmd \nu(z)=: \nabla b^1(x, \nu) +\nabla b^2(x, \nu),
\end{align*}
where we defined
\begin{align*}
   \nabla b^1(x, \nu):= &\int\limits_{\rset^d} \frac{\lambda\nabla_2^2 \kker(z, x)}{\lambda\nu\left[ \kker(z, \cdot)\right]+\varphi(z)+\eta}\rmd \nu(z),\\
\nabla b^2(x, \nu):=&\int\limits_{\rset^d} \frac{(\lambda\nabla_1^2 \kker(x, z)+\nabla^2 \varphi(x))(\lambda\nu\left[ \kker(x, \cdot)\right]+\varphi(x) +\eta) }{(\lambda\nu\left[ \kker(x, \cdot)\right]+\varphi(x)+\eta)^2}\rmd \nu(z)\\
&-\int\limits_{\rset^d}\frac{(\lambda\nabla_1 \kker(x, z)+\nabla \varphi(x))(\lambda\nu\left[ \nabla_1\kker(x, \cdot)\right]+\nabla\varphi(x))}{(\lambda\nu\left[ \kker(x, \cdot)\right]+\varphi(x)+\eta)^2}\rmd \nu(z),
\end{align*}
is jointly continuous in $(x, \nu)$. To do so, consider a sequence $(x_n, \nu_n)_{n\geq 0}\in (\rset^d\times\Pens_2(\rset^d))^\nset$ such that such that $\lim_{n \to +\infty} (x_n, \nu_n) = (x, \nu) \in \rset^d\times \Pens_2(\rset^d)$.
Then,
\begin{align*}
&\norm{\nabla b^1(x_n, \nu_n)-\nabla b^1(x, \nu)} \\
\leq& \norm{\int\limits_{\rset^d} \frac{\lambda\nabla_2^2 \kker(z, x_n)}{\lambda\nu_n\left[ \kker(z, \cdot)\right]+\varphi(z)+\eta}\rmd \nu_n(z)-\int\limits_{\rset^d} \frac{\lambda\nabla_2^2 \kker(z, x_n)}{\lambda\nu_n\left[ \kker(z, \cdot)\right]+\varphi(z)+\eta}\rmd \nu(z)} \\
&+ \norm{\int\limits_{\rset^d} \frac{\lambda\nabla_2^2 \kker(z, x_n)}{\lambda\nu_n\left[ \kker(z, \cdot)\right]+\varphi(z)+\eta}\rmd \nu(z)-\int\limits_{\rset^d} \frac{\lambda\nabla_2^2 \kker(z, x)}{\lambda\nu_n\left[ \kker(z, \cdot)\right]+\varphi(z)+\eta}\rmd \nu(z)}\\
&+ \norm{\int\limits_{\rset^d} \frac{\lambda\nabla_2^2 \kker(z, x)}{\lambda\nu_n\left[ \kker(z, \cdot)\right]+\varphi(z)+\eta}\rmd \nu(z)-\int\limits_{\rset^d} \frac{\lambda\nabla_2^2 \kker(z, x)}{\lambda\nu\left[ \kker(z, \cdot)\right]+\varphi(z)+\eta}\rmd \nu(z)}\\
\leq&  \norm{\int\limits_{\rset^d} \frac{\lambda\nabla_2^2 \kker(z, x_n)}{\lambda\nu_n\left[ \kker(z, \cdot)\right]+\varphi(z)+\eta}\rmd \nu_n(z)-\int\limits_{\rset^d} \frac{\lambda\nabla_2^2 \kker(z, x_n)}{\lambda\nu_n\left[ \kker(z, \cdot)\right]+\varphi(z)+\eta}\rmd \nu(z)}\\
&+\frac{\lambda\Mtt}{\eta}\norm{x-x_n}+\frac{\lambda^2\Mtt^2}{\eta^2}\wassersteinD[1](\nu_n, \nu),
\end{align*}
where we used the fact that $\kker$ is Lipschitz continuous with Lipschitz constant $\Mtt$ and, for the last term, the dual representation of the $\wassersteinD[1]$ distance.
In addition, since $\kker\in\rmc^\infty(\rset^d \times \rset^d, [0, +\infty))$ all third derivatives w.r.t. both arguments are bounded, i.e. $\norm{\partial^3_{ij\ell} \kker(x,y)} \leq \Mtt$, the function
\begin{align*}
z\mapsto \lambda\nabla_2^2 \kker(z, x_n)/(\lambda\nu_n\left[ \kker(z, \cdot)\right]+\varphi(z)+\eta)
\end{align*}
is Lipschitz continuous with Lipschitz constant $\Mtt_1$ for all $n$ uniformly in $x_n$. Hence, using again the dual representation of the $\wassersteinD[1]$ distance, we find 
\begin{align*}
&\norm{\int\limits_{\rset^d} \frac{\lambda\nabla_2^2 \kker(z, x_n)}{\lambda\nu_n\left[ \kker(z, \cdot)\right]+\varphi(z)+\eta}\rmd \nu_n(z)-\int\limits_{\rset^d} \frac{\lambda\nabla_2^2 \kker(z, x)}{\lambda\nu\left[ \kker(z, \cdot)\right]+\varphi(z)+\eta}\rmd \nu(z)} \\
&\qquad\leq \frac{\lambda\Mtt}{\eta}\norm{x-x_n}+\left(\Mtt_1+\frac{\lambda^2\Mtt^2}{\eta^2}\right)\wassersteinD[1](\nu_n, \nu).
\end{align*}
Similarly, we can show that the second term satisfies
\begin{align*}
\norm{\nabla b^2(x_n, \nu_n)-\nabla b^2(x, \nu)} \leq C(\norm{x-x_n}+\wassersteinD[1](\nu_n, \nu))
\end{align*}
for some finite $C$. Since $\lim_{n \to +\infty} (x_n, \nu_n) = (x, \nu)$ implies $\norm{x-x_n}\to 0$ and $\wassersteinD[1](\nu_n, \nu)\to 0$, we have that $\nabla b(x, \nu)$ is jointly continuous.
Then, the result follows from \cite[Theorem
2.11]{hu2019mean}.

%% file: siam_supplement.tex
\section{Particle System and Time Discretization}
\subsection{Proof of Proposition~\ref{prop:propagation_chaos}}
\label{app:poc}
For $\{x_1^{k,N}\}_{k=1}^N, \{x_2^{k,N}\}_{k=1}^N \in (\rset^d)^N$ and $\ell \in \{1, \dots, N\}$, denote the empirical measures $\pi_1^N := (1/N) \sum_{k=1}^N \delta_{x_1^{k,N}}$ and $\pi_2^N:=(1/N) \sum_{k=1}^N \delta_{x_2^{k,N}}$.
Using~\eqref{eq:drift_lipschitz} and the fact that for empirical measures
\begin{align*}
    \wassersteinD[1](\pi_1^N,\pi_2^N)\leq  \wassersteinD[2](\pi_1^N,\pi_2^N) \leq N^{-1/2} \norm{x_1^{1:N}-x_2^{1:N}}\leq \norm{x_1^{1:N} - x_2^{1:N}}
\end{align*}
where we used the ordering of Wasserstein distances (see, e.g. \cite[Remark 6.6]{villani2008optimal}) for the first inequality and the definition of $\wassersteinD[2]$ for the second, we have that
\begin{align*}
\norm{b(x_1^{\ell,N}, \pi_1^N) - b(x_2^{\ell,N}, \pi_2^N)}&\leq  2(\Ctt_1(1+\Mtt)+\Ctt_2+\alpha\Ltt)\norm{x_1^{1:N}-x_2^{1:N}}.
\end{align*}
The existence and strong uniqueness of a solution to~\eqref{eq:particle} is then a straightforward consequence of the above and \cite[Chapter 5, Theorem 2.9 and 2.5]{karatzas1991brownian}.

The propagation of chaos result is a straightforward consequence of the Lipschitz continuity of $b$.
Combining~\eqref{eq:bound_beta_two}, \eqref{eq:w1_dual_beta} and Assumption~\ref{assum:pi0} we have for any $x_1, x_2 \in \rset^d$ and
  $\pi_1, \pi_2 \in \Pens_1(\rset^d)$
\begin{align}
\label{eq:beta_lipschitz_poc}
    &\norm{b(x_1, \pi_1) - b(x_2, \pi_2)}\\ 
    &\ \ \leq\Ctt_1 ( \norm{x_1-x_2} +\vert \pi_2\left[ \kker(x_2, \cdot)\right]-\pi_1\left[ \kker(x_1, \cdot)\right]\vert+  \wassersteinD[1](\pi_1, \pi_2)) \notag\\
    &\ \ \ + \Ctt_2  \wassersteinD[1](\pi_1, \pi_2) + \alpha\Ltt \norm{x_1-x_2}\notag\\
    &\ \ \leq(\Ctt_1+\Ctt_2+\alpha\Ltt)( \norm{x_1-x_2} +\vert \pi_2\left[ \kker(x_2, \cdot)\right]-\pi_1\left[ \kker(x_1, \cdot)\right]\vert + \wassersteinD[1](\pi_1, \pi_2)\vert).\notag
\end{align}
The rest of the proof is classical \cite{sznitman1991topics} and given for completeness, see \cite[Appendix B.3]{crucinio2022solving} for a proof in a similar scenario.

Using~\eqref{eq:beta_lipschitz_poc}, we have for any $t \geq 0 $
\begin{align}
  \label{eq:poc_proof}
  &\Exp\left[ \sup_{s \in [0, t]}\norm{\bm{X}_s - \bm{X}_s^{1,N}} \right]\\
  &\qquad\leq \int\limits_{0}^t \norm{b(\bm{X}_s, \pi_s) - b(\bm{X}_s^{1, N}, \pi^N_s)} \rmd s\notag\\
    &\qquad\leq (\Ctt_1+\Ctt_2+\alpha\Ltt)\int\limits_{0}^t \Exp\left[\sup_{u \in [0, s]}\norm{\bm{X}_u - \bm{X}_u^{1,N}}\right]\rmd s \notag \\
    &\qquad+  (\Ctt_1+\Ctt_2+\alpha\Ltt)\int\limits_{0}^t \Exp\left[\left\vert\frac{1}{N}\sum_{i=1}^N\kker(\bm{X}_s^{1, N}, \bm{X}_s^{i, N})-\pi_s\left[\kker(\bm{X}_s, \cdot)\right] \right\vert\right] \rmd s\notag\\
    &\qquad +(\Ctt_1+\Ctt_2+\alpha\Ltt)\int\limits_{0}^t \Exp\left[\wassersteinD[1](\pi_s, \pi_s^N) \right]\rmd s\notag\\
       &\qquad\leq 2(\Ctt_1+\Ctt_2+\alpha\Ltt)\int\limits_{0}^t \Exp\left[\sup_{u \in [0, s]}\norm{\bm{X}_u - \bm{X}_u^{1,N}}\right]\rmd s \notag \\
    &\qquad+  (\Ctt_1+\Ctt_2+\alpha\Ltt)\int\limits_{0}^t \Exp\left[\left\vert\frac{1}{N}\sum_{i=1}^N\kker(\bm{X}_s^{1, N}, \bm{X}_s^{i, N})-\pi_s\left[\kker(\bm{X}_s, \cdot)\right] \right\vert\right] \rmd s\notag,
\end{align}
  where the last inequality follows using the convexity of $\wassersteinD[1](\pi_s, \cdot)$ and exchangeability, which implies
\begin{align}
\label{eq:wasserstein_ordering}
    \mathbb{E}[\wassersteinD[1](\pi_s, \pi_s^N)] &= \mathbb{E}\left[ \wassersteinD[1]\left(\pi_s, \frac{1}{N} \sum_{i=1}^N \delta_{\bm{X_s}^{i,N}}\right) \right]\\
    \leq&
    \frac{1}{N} \sum_{i=1}^N \mathbb{E}\left[ \wassersteinD[1]\left(\pi_s, \delta_{\bm{X_s}^{i,N}}\right)\right] \notag\\
    \leq& \frac{1}{N}  \sum_{i=1}^N \mathbb{E}[\Vert \bm{X_s} - \bm{X_s}^{i,N}\Vert] \notag\\
    \leq& \mathbb{E}[\Vert \bm{X_s} - \bm{X_s}^{1,N}\Vert].\notag
\end{align}

  Now, consider $N$ independent copies of the nonlinear process $\bm{X}_s$, $\{(\bm{X}_s^{k,\star})_{t \geq 0}\}_{k=1}^N$. We can bound the second term in the above with
\begin{align*}
    &\Exp\left[\left\vert\frac{1}{N}\sum_{i=1}^N\kker(\bm{X}_s^{1, N}, \bm{X}_s^{i, N})-\pi_s\left[\kker(\bm{X}_s, \cdot)\right] \right\vert\right] \\
    &\qquad\leq \frac{1}{N}\Exp\left[\left\vert\sum_{i=1}^N\kker(\bm{X}_s^{1, N}, \bm{X}_s^{i, N})-\kker(\bm{X}_s^{1, N}, \bm{X}_s^{i, \star})\right\vert\right]   \\
   &\qquad+\frac{1}{N}\Exp\left[\left\vert\sum_{i=1}^N\kker(\bm{X}_s^{1, N}, \bm{X}_s^{i, \star})-\pi_s\left[\kker(\bm{X}_s, \cdot)\right] \right\vert\right]\\
   &\qquad\leq \Mtt \Exp\left[\sup_{u \in [0, s]}\norm{\bm{X}_u - \bm{X}_u^{1,N}}\right]\\
   &\qquad+\frac{1}{N}\Exp\left[\left\vert\sum_{i=1}^N\kker(\bm{X}_s^{1, N}, \bm{X}_s^{i, \star})-\pi_s\left[\kker(\bm{X}_s, \cdot)\right] \right\vert\right]
\end{align*}
where we used the Lipschitz continuity of $\kker$, \eqref{eq:wasserstein_ordering} and the fact that $\{(\bm{X}_t^{k,N})_{t \geq 0}\}_{k=1}^N$ is exchangeable to obtain the last inequality.

Consider now the last term in the display above and let us denote $\Delta_i := \kker(\mathbf{X}_s^{1,N},\mathbf{X}_s^{i,\star}) - \pi_s[\kker(\mathbf{X}_s^{1, N},\cdot)]$, then
\begin{align*}
   & \frac{1}{N}\Exp\left[\left\vert\sum_{i=1}^N\kker(\bm{X}_s^{1, N}, \bm{X}_s^{i, \star})-\pi_s\left[\kker(\bm{X}_s, \cdot)\right] \right\vert\right]\\
   &\qquad\qquad\leq \frac{1}{N}\Exp\left[\left| \sum_{i=1}^N \Delta_i \right|\right]+\Exp\left[\left\vert\pi_s\left[\kker(\bm{X}_s^{1, N}, \cdot)\right] -\pi_s\left[\kker(\bm{X}_s, \cdot)\right] \right\vert\right]\\
    &\qquad\qquad\leq  \frac{1}{N}\Exp\left[\left| \sum_{i=1}^N \Delta_i \right|\right]+\Mtt \Exp\left[\sup_{u\in[0, s]}\norm{\bm{X}_u^{1, N}-\bm{X}_u}\right],
\end{align*}
using the Lipschitz continuity of $\kker$. For the remaining term, using Jensen's inequality and the independence of the $\Delta_i$s we find
\begin{align*}
       \frac{1}{N} \mathbb{E}\left[ \left| \sum_{i=1}^N \Delta_i \right| \right]
       =&\frac{1}{N} \mathbb{E}\left[ \left\{ \left( \sum_{i=1}^N \Delta_i \right)^2 \right\}^{1/2} \right]\\
       \leq& \frac{1}{N} \mathbb{E}\left[  \left( \sum_{i=1}^N \Delta_i \right)^2 \right]^{1/2} \qquad \textrm{(Jensen)}\\
       =& \frac{1}{N} \mathbb{E}\left[  \left( \sum_{i=1}^N \Delta_i^2 + \sum_{i\neq j} \Delta_i \Delta_j \right) \right]^{1/2} \\
       =& \frac{1}{N} \mathbb{E}\left[  \sum_{i=1}^N \Delta_i^2  \right]^{1/2} \qquad \textrm{(Independence)}.
\end{align*}

Plugging the above into~\eqref{eq:poc_proof} we obtain

\resizebox{0.95\textwidth}{!}{\noindent
\hspace*{-0.5cm}\parbox{\textwidth}{\noindent
\begin{align*}
 \Exp\left[ \sup_{s \in [0, t]}\norm{\bm{X}_s - \bm{X}_s^{1,N}} \right]&\leq 2(1+\Mtt)(\Ctt_1+\Ctt_2+\alpha\Ltt)\int\limits_{0}^t \Exp\left[\sup_{u \in [0, s]}\norm{\bm{X}_u - \bm{X}_u^{1,N}}\right]\rmd s \\
 &+\frac{\Ctt_1+\Ctt_2+\alpha\Ltt}{N}\int_0^t\Exp\left[\sum_{i=1}^N\vert\kker(\bm{X}_s^{1, N}, \bm{X}_s^{i, \star})-\pi_s\left[\kker(\bm{X}_s^{1, N}, \cdot)\right] \vert^2\right]^{1/2}\rmd s.
\end{align*}
}}

Using Popoviciu's inequality on variances \cite{popoviciu1935equations} and recalling that $0\leq \kker(x,y)\leq \Mtt$ for all $(x, y)\in\rset^d\times \rset^d$, we have
\begin{align*}
\Exp\left[\vert\kker(\bm{X}_s^{1, N}, \bm{X}_s^{i, \star})-\pi_s\left[\kker(\bm{X}^{1, N}_s, \cdot)\right] \vert^2\mid \bm{X}_s^{1, N}\right] &= \textrm{var}\left(\kker(\bm{X}_s^{1, N}, \bm{X}_s^{i, \star})\mid \bm{X}_s^{1, N}\right)\\
&\leq \Mtt^2/4.
\end{align*}
It follows that
\begin{align*}
 \Exp\left[ \sup_{s \in [0, t]}\norm{\bm{X}_s - \bm{X}_s^{1,N}} \right]&\leq (2+\Mtt)(\Ctt_1+\Ctt_2+\alpha\Ltt)\int\limits_{0}^t \Exp\left[\sup_{u \in [0, s]}\norm{\bm{X}_u - \bm{X}_u^{1,N}}\right]\rmd s \\
 &+\Mtt/2(\Ctt_1+\Ctt_2+\alpha\Ltt)N^{-1/2}t.
\end{align*}
  Using Gr\"{o}nwall's lemma we get that for any $T \geq 0$ there exists
  $C_N(T) \geq 0$ such that for any $N \in \mathbb{N}$
\begin{equation*}
  \Exp\left[ \sup_{t \in [0, T]}\norm{\bm{X}_t - \bm{X}_t^{1,N}} \right] \leq C_N(T) N^{-1/2},
\end{equation*}
which concludes the proof.

\subsection{Proof of Proposition~\ref{prop:particle_ergodic}}
\label{app:ergodic}

First, we recall that $b$ is given for any $x \in \rset^d$, 
$\nu \in \Pens(\rset^d)$ by
\begin{equation*}
    b(x, \nu) = \int b^\eta(x, z, \nu) \rmd \nu(z)-\alpha\nabla  U(x)
\end{equation*}
with $b^{\eta}$ defined in~\eqref{eq:drift}.
Recall that under Assumption~\ref{assum:general_kker} and~\ref{assum:pi0} we have~\eqref{eq:drift_lipschitz}, i.e., there exist $C_1 \geq 0$ such that for any  $x_1, x_2 \in \rset^d$ and $\pi_1, \pi_2 \in \Pens(\rset^d)$
\begin{equation}
  \label{eq:big1}
     \norm{b(x_1, \pi_1) - b(x_2, \pi_2)} \leq C_1 \left( \norm{x_1 - x_2} + \wassersteinD[1](\pi_1, \pi_2)\right). 
 \end{equation}
Furthermore, using Assumption~\ref{assum:general_kker} we have for any  $x_1, x_2 \in \rset^d$, $\pi_1, \pi_2 \in \Pens(\rset^d)$,
\begin{align}
  \label{eq:big2}
\norm{b^\eta(x, z, \nu)} &\leq \norm{\frac{\lambda\nabla_2 \kker(z, x)}{\lambda\nu\left[ \kker(z, \cdot)\right]+\varphi(z)+\eta}}+\norm{\frac{\lambda\nabla_1 \kker(x, z)+\nabla \varphi(x)}{\lambda\nu\left[ \kker(x, \cdot)\right]+\varphi(x)+\eta}} \\
& \leq \lambda\Mtt/\eta + (\lambda+1)\Mtt /\eta =: C_0.\notag 
 \end{align}
In addition, note that for any $x_1^{1:N},x_2^{1:N} \in (\rset^d)^N$ we have for any $i \in \{1,2\}$
 \begin{equation}
   \label{eq:w1_empi}
   \wassersteinD[1](\pi_1^N, \pi_2^N) \leq \frac{1}{N} \sum_{k=1}^N \norm{x_1^{k,N} - x_2^{k,N}}, \qquad \pi_i^N = \frac{1}{N} \sum_{k=1}^N \delta_{x^{k,N}}.
 \end{equation}
Let $N \in \nset$ and denote $B_N: (\rset^d)^N \to (\rset^d)^N$ given for any $x^{1:N} \in (\rset^d)^N$ by
\begin{equation*}
B_N(x^{1:N}) = \{b(x^{k,N}, \pi^N)\}_{k \in \{1, \dots, N\}}, \qquad \pi^N = \frac{1}{N} \sum_{k=1}^N \delta_{x^{k,N}}.
\end{equation*}
Therefore, using~\eqref{eq:big1}--\eqref{eq:w1_empi} we have for any $x^{1:N}_1, x_2^{1:N} \in (\rset^d)^N$
\begin{align}
  \label{eq:lip_N}
& \norm{B_N(x^{1:N}_1) - B_N(x^{1:N}_2)}\\
&\qquad\leq C_1\left(\sum_{k=1}^N\norm{x_1^{k,N}-x_2^{k,N}}^2+N\wassersteinD[1]\left(\frac{1}{N} \sum_{k=1}^N \delta_{x_1^{k,N}}, \frac{1}{N} \sum_{k=1}^N \delta_{x_2^{k,N}}\right)^2\right.\notag\\
&\qquad\qquad\left.+2\sum_{k=1}^N\norm{x_1^{k,N}-x_2^{k,N}}\wassersteinD[1]\left(\frac{1}{N} \sum_{k=1}^N \delta_{x_1^{k,N}}, \frac{1}{N} \sum_{k=1}^N \delta_{x_2^{k,N}}\right)\right)^{1/2}\notag\\
  &\qquad \leq C_1\left(\sum_{k=1}^N\norm{x_1^{k,N}-x_2^{k,N}}^2+ \frac{1}{N}\left(\sum_{k=1}^N \norm{x_1^{k,N} - x_2^{k,N}}\right)^2\right.\notag\\
&\qquad\qquad\left.+2/N\sum_{k=1}^N\norm{x_1^{k,N}-x_2^{k,N}}\sum_{k=1}^N \norm{x_1^{k,N} - x_2^{k,N}}\right)^{1/2}\notag\\
&\qquad\leq 2C_1\sum_{k=1}^N\norm{x_1^{k,N}-x_2^{k,N}}\notag\\
 &\qquad \leq 2C_1N^{1/2} \norm{x_1^{1:N} - x_2^{1:N}}\notag,
\end{align}
where the last line follows from H\"older's inequality.

Using~\eqref{eq:big1}--\eqref{eq:big2} and Assumption~\ref{assum:pi0}--\ref{item:dissipativity} we have for any $x^{1:N}_1, x_2^{1:N} \in (\rset^d)^N$
\begin{align*}
 & \langle B_N(x^{1:N}_1) - B_N(x^{1:N}_2), x_1^{1:N} - x_2^{1:N} \rangle\\ 
  &\qquad\leq -\alpha \mtt \norm{x_1^{1:N} - x_2^{1:N}}^2 + \alpha \ctt N + 2C_0 \sum_{k=1}^N \norm{x_1^{k,N} - x_2^{k,N}} \notag\\
 &\qquad\leq -\alpha \mtt \norm{x_1^{1:N} - x_2^{1:N}}^2 + \alpha \ctt N + 2C_0 N^{1/2} \norm{x_1^{1:N} - x_2^{1:N}}.\notag
\end{align*}
Let $R = \max(8C_0N^{1/2}/(\alpha \mtt), 2( \ctt N/ \mtt)^{1/2})$ with $\mtt$ and $\ctt$ as in Assumption~\ref{assum:pi0}--\ref{item:dissipativity}. Then
\begin{align*}
    \alpha\ctt N\leq \frac{\alpha R^2\mtt}{4}\qquad\qquad 2C_0 N^{1/2}\leq \frac{\alpha R\mtt}{4},
\end{align*}
and for any $x^{1:N}_1, x_2^{1:N} \in (\rset^d)^N$ with $\norm{x_1^{1:N} - x_2^{1:N}} \geq R$ we further have
\begin{align}
  \label{eq:diss_N2}
  &\langle B_N(x^{1:N}_1) - B_N(x^{1:N}_2), x_1^{1:N} - x_2^{1:N} \rangle \\
  &\qquad\qquad\leq -\alpha \mtt \norm{x_1^{1:N} - x_2^{1:N}}^2 +\frac{\alpha R^2\mtt}{4} + \frac{\alpha \mtt R}{4}\norm{x_1^{1:N} - x_2^{1:N}}\notag\\
  &\qquad\qquad\leq  -(\alpha \mtt/2)  \norm{x_1^{1:N} - x_2^{1:N}}^2.\notag
\end{align}
We conclude upon combining~\eqref{eq:lip_N},~\eqref{eq:diss_N2} and \cite[Corollary 22]{debortoli2019convergence} whose assumptions C1 and B3 are satisfied and we verify B2 in~\eqref{eq:lip_N} which in turns implies B5.

\subsection{Proof of Proposition~\ref{prop:particles_min}}
\label{app:invariant}

Let $\alpha, \eta >0$. First we show that
  $\{\pi^{N}\}_{N \in \nset}$ is relatively compact in
  $\Pens_1(\rset^d)$.  Let $N \in \nset$ and assume that $\bm{X}_0^{1:N} = 0$.
Define for any $t \geq 0$
  \begin{equation*}
    \bm{M}_t^N := \frac12\norm{\bm{X}_t^{1,N}}^2 -\int\limits_0^t \left\lbrace\langle \bm{X}_u^{1,N}, b(\bm{X}_u^{1,N}, \pi^N_u) \rangle + (1+\alpha) d\right\rbrace\rmd u,
  \end{equation*}
where $\{(\bm{X}_t^{k,N})_{t \geq 0}\}_{k=1}^N$ is given in \eqref{eq:particle}.
Using It\^{o}'s formula we have that, for any $s < t$: 
\begin{align*}
    \frac12\norm{\bm{X}_t^{1,N}}^2 &= \frac12\norm{\bm{X}_s^{1,N}}^2 +\int\limits_s^t \langle \bm{X}_u^{1,N}, b(\bm{X}_u^{1,N}, \pi^N_u) \rangle \rmd u\\
    &+(1+\alpha)d(t-s) +\sqrt{2(1+\alpha)}\int_s^t \langle \bm{X}_u^{1,N}, \rmd \bm{B}_{u}^{1}\rangle,
\end{align*}
so that 
\begin{align*}
    \bm{M}_t^N  &= \frac12\norm{\bm{X}_s^{1,N}}^2-\int\limits_0^s \langle \bm{X}_u^{1,N}, b(\bm{X}_u^{1,N}, \pi^N_u) \rangle\rmd u\\
    &-(1+\alpha)ds+\sqrt{2(1+\alpha)}\int_s^t \langle \bm{X}_u^{1,N}, \rmd \bm{B}_{u}^{1}\rangle\\
    &= \bm{M}_s^N+\sqrt{2(1+\alpha)}\int_s^t \langle \bm{X}_u^{1,N}, \rmd \bm{B}_{u}^{1}\rangle
\end{align*}
It follows that
\begin{align*}
    \Exp\left[\bm{M}_t^N -\bm{M}_s^N \mid \mathcal{F}_s^N\right]=\sqrt{2(1+\alpha)}\Exp\left[\int_s^t \langle \bm{X}_u^{1,N}, \rmd \bm{B}_{u}^{1}\rangle\mid \mathcal{F}_s^N\right].
\end{align*}
\cite[Theorem 3.2.1]{oksendalstochastic} guarantees that the expectation above is zero if $\bm{X}_u^{1,N}$ is such that $\mathbb{E}[\int_0^T ||X_u^{1,N}||^2du] < \infty$ for every $T>0$. 
We show that this is the case  under our assumptions using \cite[Theorem 5.2.1]{oksendalstochastic}.
First, ~\eqref{eq:drift_lipschitz} guarantees that the drift coefficient~\eqref{eq:b} is Lipschitz continuous.
We then need to show
\begin{align*}
    \norm{b(x, \nu)}\leq C(1+\norm{x}).
\end{align*}
Observe that
\begin{align*}
    \norm{b(x, \nu)} = \norm{b(x, \nu) - b(0, \nu)+b(0, \nu)}\leq C_1 \norm{x}+\norm{b(0, \nu)}
\end{align*}
due to Lipschitzness. Moreover, $\norm{b(0, \nu)}$ is finite :
\begin{align*}
    \norm{b(0, \nu)} &\leq \int \norm{b^\eta(0, z, \nu)} \rmd \nu(z)+\alpha\norm{\nabla  U(0)}\\
    &\leq C_0+\alpha\norm{\nabla  U(0)}
\end{align*}
with $C_0$ in~\eqref{eq:big2}. We can then apply \cite[Theorem 5.2.1]{oksendalstochastic} which guarantees $\mathbb{E}[\int_0^T ||X_u^{1,N}||^2du] < \infty$ for every $T>0$ and conclude that $\Exp\left[\bm{M}_t^N -\bm{M}_s^N \mid \mathcal{F}_s^N\right]=0$.

Since $\Exp\left[\bm{M}_t^N -\bm{M}_s^N \mid \mathcal{F}_s^N\right]=0$, we also have that 
\begin{align*}
&\frac12\Exp\left[\norm{\bm{X}_t^{1,N}}^2\right] - \frac12\Exp\left[\norm{\bm{X}_s^{1,N}}^2\right] \\
&\qquad=\Exp\left[\int\limits_s^t \left\lbrace\langle \bm{X}_u^{1,N}, b(\bm{X}_u^{1,N}, \pi^N_u) \rangle + (1+\alpha) d\right\rbrace\rmd u\right]\\
&\qquad=\Exp\left[\int\limits_s^t \langle \bm{X}_u^{1,N}, \int\limits_{\rset^d}b^\eta(\bm{X}_u^{1,N}, z, \pi^N_u) \rmd \pi_u^N(z)\rangle \rmd u\right]\\
&\qquad\phantom{=}-\Exp\left[\int\limits_s^t \left\lbrace\langle \bm{X}_u^{1,N}, \alpha\nabla U(\bm{X}_u^{1,N}) \rangle - (1+\alpha)  d\right\rbrace\rmd u\right],
\end{align*}
where we used the definition of $b$ in~\eqref{eq:b}.

Using Assumption~\ref{assum:pi0}--\ref{item:dissipativity} with $x_1 = x$ and $x_2 =0$, we have for any $x \in \rset^d$
\begin{equation*}
\langle \nabla U(x), x \rangle \geq \langle \nabla U(0), x \rangle + \mtt \norm{x}^2 - \ctt. 
\end{equation*}
Therefore, using this result and the Cauchy--Schwarz inequality we obtain that for any $t \geq 0$
  \begin{align*}
&\frac12\Exp\left[\norm{\bm{X}_t^{1,N}}^2\right] - \frac12\Exp\left[\norm{\bm{X}_s^{1,N}}^2\right] \\
&\qquad\leq \Exp\left[ \int\limits_s^t \left\lbrace \int\limits_{\rset^d} \norm{b^\eta(\bm{X}_u^{1,N}, z, \pi^N_u)}\norm{\bm{X}_u^{1,N}} \rmd \pi_u^N(z)\right\rbrace \rmd u\right]\\
&\qquad\phantom{=} +\Exp\left[\int\limits_s^t \left\lbrace
 -\alpha \mtt \norm{\bm{X}_u^{1,N}}^2 + \alpha \norm{\nabla U(0)}\norm{\bm{X}_u^{1,N}} + (1+\alpha) (\ctt +  d)\right\rbrace \rmd u\right]\\
&\qquad= \int\limits_s^t \Exp\left[  \int\limits_{\rset^d} \norm{b^\eta(\bm{X}_u^{1,N}, z, \pi^N_u)}\norm{\bm{X}_u^{1,N}} \rmd \pi_u^N(z) \right]\rmd u\\
&\qquad\phantom{=} +\int\limits_s^t \left\lbrace -\alpha \mtt\Exp\left[
  \norm{\bm{X}_u^{1,N}}^2\right] + \alpha \norm{\nabla U(0)}\Exp\left[\norm{\bm{X}_u^{1,N}}\right] + (1+\alpha)  (\ctt+ d)\right\rbrace \rmd u,
  \end{align*}
where the last line follows from Tonelli's Theorem since all integrated functions are positive (or always negative, in which case we can consider minus the integral itself).

Let $\mathcal{V}_t^N =   \Exp\left[\norm{\bm{X}_t^{1,N}}^2\right]$. Appealing to the fundamental theorem of calculus  we get that
\begin{align*}
\frac12\rmd \mathcal{V}_t^N / \rmd t \leq& \Exp\left[ \int\limits_{\rset^d} \norm{b^\eta(\bm{X}_t^{1,N}, z, \pi^N_t)}\norm{\bm{X}_t^{1,N}} \rmd \pi_t^N(z) \right]\\
  & -\alpha  \mtt \Exp\left[\norm{\bm{X}_t^{1,N}}^2\right] + \alpha\norm{\nabla U(0)}\Exp\left[\norm{\bm{X}_t^{1,N}}\right] \\
  &+ (1+\alpha) ( \ctt + d).
\end{align*}
Using Jensen's inequality we have $\Exp\left[\norm{\bm{X}_t^{1,N}}\right]\leq \Exp\left[\norm{\bm{X}_t^{1,N}}^2\right]^{1/2}$ and using that, as shown in \eqref{eq:big2}, for any $x \in \rset^d$ and $\pi \in \Pens(\rset^d)$ we have $\norm{b^\eta(x, z,\pi)} \leq C_0$ we have:
\begin{align*}
\frac12\rmd \mathcal{V}_t^N / \rmd t \leq& C_0 (\mathcal{V}_t^N)^{1/2}
    -\alpha  \mtt \mathcal{V}_t^N + \alpha \norm{\nabla U(0)}  (\mathcal{V}_t^N)^{1/2}  + (1+\alpha) ( \ctt +  d)\\
      =& (C_0 + \alpha  \norm{\nabla U(0)}) (\mathcal{V}_t^N)^{1/2} + (1+\alpha) ( \ctt +  d) - \alpha  \mtt \mathcal{V}_t^N  
\end{align*}
Noting that for any $\mathsf{a}, \mathsf{b}$ such that $\mathsf{ab} \geq 1/2$, for any $x \geq 0$ we have $\sqrt{x} \leq \mathsf{a} + \mathsf{b} x$, and setting $\mathsf{a} = (C_0 + \alpha  \norm{\nabla U(0)})/ (\alpha \mtt)$ and $\mathsf{b} = 1/2\mathsf{a}$ we have:
\begin{align*}
  \frac12\rmd \mathcal{V}_t^N / \rmd t 
     \leq& \frac{(C_0 + \alpha\norm{\nabla U(0)})}{\alpha  \mtt} + (1+\alpha) ( \ctt +  d )- \frac12 \alpha  \mtt \mathcal{V}_t^N 
\end{align*}
Hence, for any $t \geq 0$ and any $N \in \nset$ we get that 
$\mathcal{V}_t^N \leq C$ with
\begin{equation}
  C = \frac{2}{\alpha\mtt} \left[ \frac{(C_0 + \alpha  \norm{\nabla U(0)})}{\alpha  \mtt} + (1+\alpha) ( \ctt +  d)\right].
\end{equation}
Therefore, letting $t \to +\infty$ we get that for any $N \in \nset$,
$\int\limits_{\rset^d} \norm{x}^2 \rmd \pi^{N}(x) \leq C$. Hence,
$\{\pi^{N}\}_{N \in
  \nset}$ is relatively compact in
$\Pens_2(\rset^d)$ using \cite[Proposition 7.1.5]{ambrosio2008gradient}.

Let $\pi^\star$ be a cluster point of $\{\pi^{ N}\}_{N \in
  \nset}$. Le us denote by $\pi_t(\pi^\star), \pi_t^N(\pi^\star)$ the law of $\bm{X}_t$ following \eqref{eq:mckean_sde} with initial condition $\bm{X}_0\sim \pi^\star$ and $\bm{X}_t^{1,N}$ following \eqref{eq:particle} with initial condition $\bm{X}_0^{1,N}\sim \pi^\star$.
Let $(N_k)_{k \in \nset}$ be an increasing sequence such that
$\lim_{k \to +\infty} \wassersteinD[1](\pi^{N_k}, \pi^\star) = 0$. We have that for any $t \geq 0$
\begin{align}
  \label{eq:la_grosse_borne}
  \wassersteinD[1](\pi^\star, \pi_{\alpha, \eta}^\star) &\leq \wassersteinD[1](\pi^\star, \pi^{N_k})   \\ 
  & + \wassersteinD[1](\pi^{N_k}, \pi_t^{N_k}(\pi^{N_k}))+ \wassersteinD[1](\pi_t^{N_k}(\pi^{N_k}), \pi_t(\pi^{ N_k})) \notag\\
  &+ \wassersteinD[1](\pi_t(\pi^{N_k}), \pi_t(\pi^{\star})) + \wassersteinD[1](\pi_t(\pi^{\star}), \pi_{\alpha, \eta}^\star), \notag
\end{align}
where $\pi^{N_k}$ denotes the invariant measure of $(\bm{X}_t^{1:N_k})_{t \geq 0}$ which exists by Proposition~\ref{prop:particle_ergodic}.
We now control each of these terms. Let $\varepsilon >0$, for sufficiently large $t \geq 0$ Proposition \ref{prop:la_convergence_star} ensures
$\wassersteinD[1](\pi_t(\pi^{\star}), \pi_{\alpha, \eta}^\star)
\leq \varepsilon$. Using the Lipschitz continuity of the drift established in \eqref{eq:bound_beta_two} and Assumption \ref{assum:pi0} with \cite[Lemma 20]{crucinio2022solving}, there exists $k_0 \in \nset$ such that for any
$k \geq k_0$ and any $t\in[0, T]$ we have
$$\wassersteinD[1](\pi_t(\pi^{N_k}),
\pi_t(\pi^{\star})) \leq C_T\wassersteinD[1](\pi^{N_k},
\pi^{\star})\leq \varepsilon$$ since we assumed $\lim_{k \to +\infty} \wassersteinD[1](\pi^{N_k}, \pi^\star) = 0$. Using
Proposition \ref{prop:propagation_chaos}, there exists $k_1 \in \nset$ such that for any
$k \geq k_1$ we have that
$\wassersteinD[1](\pi_t^{N_k}(\pi^{N_k}), \pi_t(\pi^{
  N_k})) \leq \varepsilon$. Since $\pi^{N_k}$ is invariant for
$(\bm{X}_t^{1:N_k})_{t \geq 0}$ we get that
$\wassersteinD[1](\pi^{N_k}, \pi_t^{N_k}(\pi^{N_k})) = 0$. Finally,
there exists $k_2 \in \nset$ such that for any $k \geq k_2$ we have
$\wassersteinD[1](\pi^\star, \pi^{N_{k}}) \leq \varepsilon$ since we assumed $\lim_{k \to +\infty} \wassersteinD[1](\pi^{N_k}, \pi^\star) = 0$.  
Combining these
results in~\eqref{eq:la_grosse_borne}, we get that
$\wassersteinD[1](\pi^\star, \pi_{\alpha, \eta}^\star) \leq 4\varepsilon$. Therefore, since $\varepsilon > 0$ is arbitrary, we have that $\pi^\star = \pi_{\alpha, \eta}^\star$, which concludes the proof.

%% file: additional_expe.tex
\section{Additional Experiments}

We empirically validate our convergence results in Proposition~\ref{prop:la_convergence_star}--\ref{prop:particles_min} and~\eqref{eq:euler_error} using the simple Gaussian example in Section~\ref{sec:expe} with $\lambda = 1/2, \beta = 1/2$ so that we have
\begin{align*}
    \varphi(x)&=\frac{1}{3}\N\left(x;0,1\right)\\
    \kker(x, y) &=\N\left(y;xe^{-1/2},1-e^{-1}\right).
\end{align*}
The unique solution is $\pi(x)=\N\left(x;0,1\right)$. The reference measure is $\pi_0(x)=\pi(x)$ so that $\pi$ is the unique minimizer of $\Fun_\alpha^\eta$ and $\alpha=0.1$. We pick a small value of $\alpha$ to analyze the behavior of the algorithm with little influence from the reference measure.

We start by checking the validity of Proposition~\ref{prop:la_convergence_star}, \ref{prop:particle_ergodic} on this example and consider 4 initial distributions: a diffuse Gaussian $\mathcal{N}(0, 2^2)$, a concentrated Gaussian $\mathcal{N}(0, 0.1^2)$, the target distribution $\mathcal{N}(0, 1^2)$ and a Uniform over $[-1, 1]$.
We set $N=500$, $\gamma = 10^{-3}$ and iterate for $n_T=1000$ iterations to allow for the runs with initial distributions which are far from $\pi$ to converge.

Figure~\ref{fig:appendix} shows the evolution of $\wassersteinD[1](\hat{\pi}_n^N, \pi_{\alpha, \eta}^\star)$, where $\hat{\pi}_n^N$ is the output of Algorithm~\ref{alg:second_kind} and $\pi_{\alpha, \eta}^\star=\pi$ is the true minimizer, along iterations.  
The plot shows that $\wassersteinD[1]$ decreases as the number of iteration increases for all initial distributions. The exponential decay predicted by Proposition~\ref{prop:particle_ergodic} is particularly evident for the concentrated and uniform initial distributions.

\begin{figure}
\centering
\begin{tikzpicture}[every node/.append style={font=\normalsize}]
\node (img1) {\includegraphics[width=0.75\textwidth]{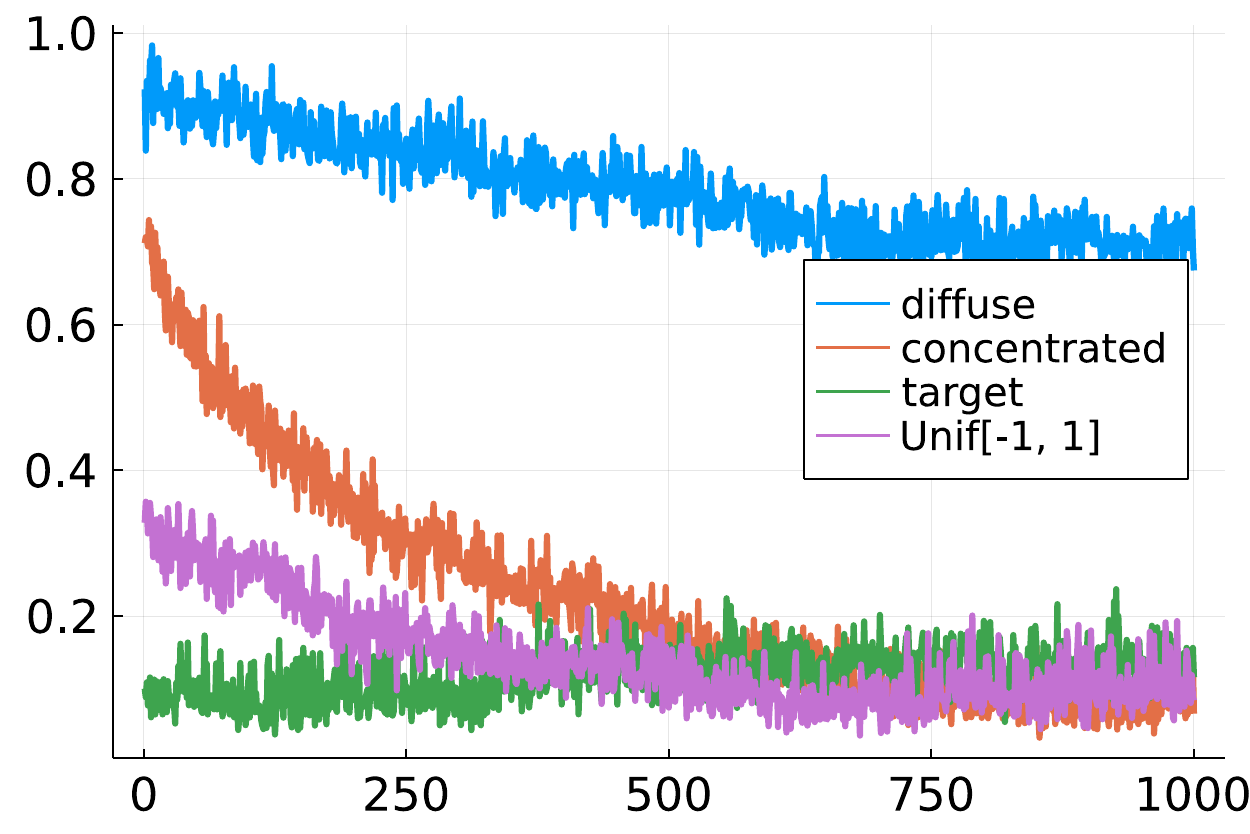}};
\node[below=of img1, node distance = 0, yshift = 1cm] (label1) {$n$};
  \node[left=of img1, node distance = 0, rotate=90, anchor = center, yshift = -0.8cm] {$\wassersteinD[1](\hat{\pi}_n^N, \pi_{\alpha, \eta}^\star)$};
  \end{tikzpicture}
\caption{Decay of $\wassersteinD[1]$ along iterations for 4 initial distributions: a more diffuse Gaussian $\mathcal{N}(0, 2^2)$, a more concentrated one $\mathcal{N}(0, 0.1^2)$, the target $\mathcal{N}(0, 1^2)$ and a Uniform distribution centered at 0, $\textrm{Unif}[-1, 1]$. }
\label{fig:appendix}
\end{figure}

We also check the validity of the rate in Proposition~\ref{prop:propagation_chaos} and of the convergence established in Proposition~\ref{prop:particles_min} on this example. We set $\gamma = 10^{-2}$ and $n_T=200$ steps. The initial distribution is a concentrated one ($\mathcal{N}(0, 0.1^2)$) and the reference measure is $\pi_0=\pi$ with $\alpha = 0.1$. We consider number of particles $N= 50, 100, 200, 500, 1000$.
Figure~\ref{fig:Ngamma} left panel shows the evolution of $\wassersteinD[1]$ between the solution $\pi$ and the approximation provided by Algorithm~\ref{alg:second_kind} at the final time $n_T$. As predicted by Proposition~\ref{prop:particles_min}, $\wassersteinD[1](\hat{\pi}_{n_T}^N, \pi_{\alpha, \eta}^\star)$ decays as $N\to \infty$.
The right panel shows the square root of the mean squared error ($\mse$) for the mean of $\pi$. We add the theoretical rate in Proposition~\ref{prop:propagation_chaos} for reference. As is often the case, for small $N$ the decay is faster than $\sqrt{N}$, but as $N\to \infty$ the $\mse$ decays at the predicted rate. The constant $c_1(T)$ is empirically estimated to be $c_1(T)\approx 0.075$.

\begin{figure}
\centering
\begin{tikzpicture}[every node/.append style={font=\normalsize}]
\node (img1) {\includegraphics[width=0.4\textwidth]{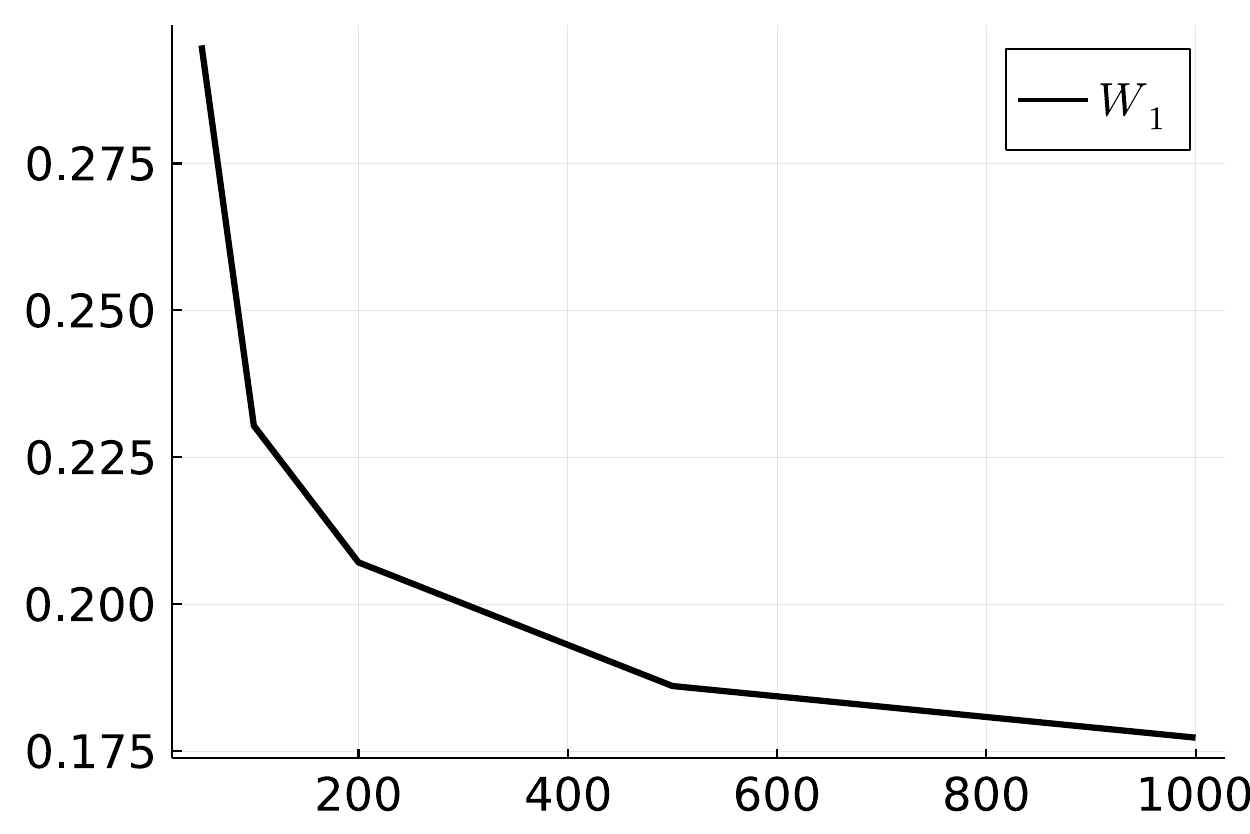}};
\node[below=of img1, node distance = 0, yshift = 1cm] (label1) {$N$};
  \node[left=of img1, node distance = 0, rotate=90, anchor = center, yshift = -0.8cm] {$\wassersteinD[1](\hat{\pi}_{n_T}^N, \pi_{\alpha, \eta}^\star)$};
  \node[right=of img1, node distance = 0, xshift = -0.5cm] (img2) {\includegraphics[width=0.4\textwidth]{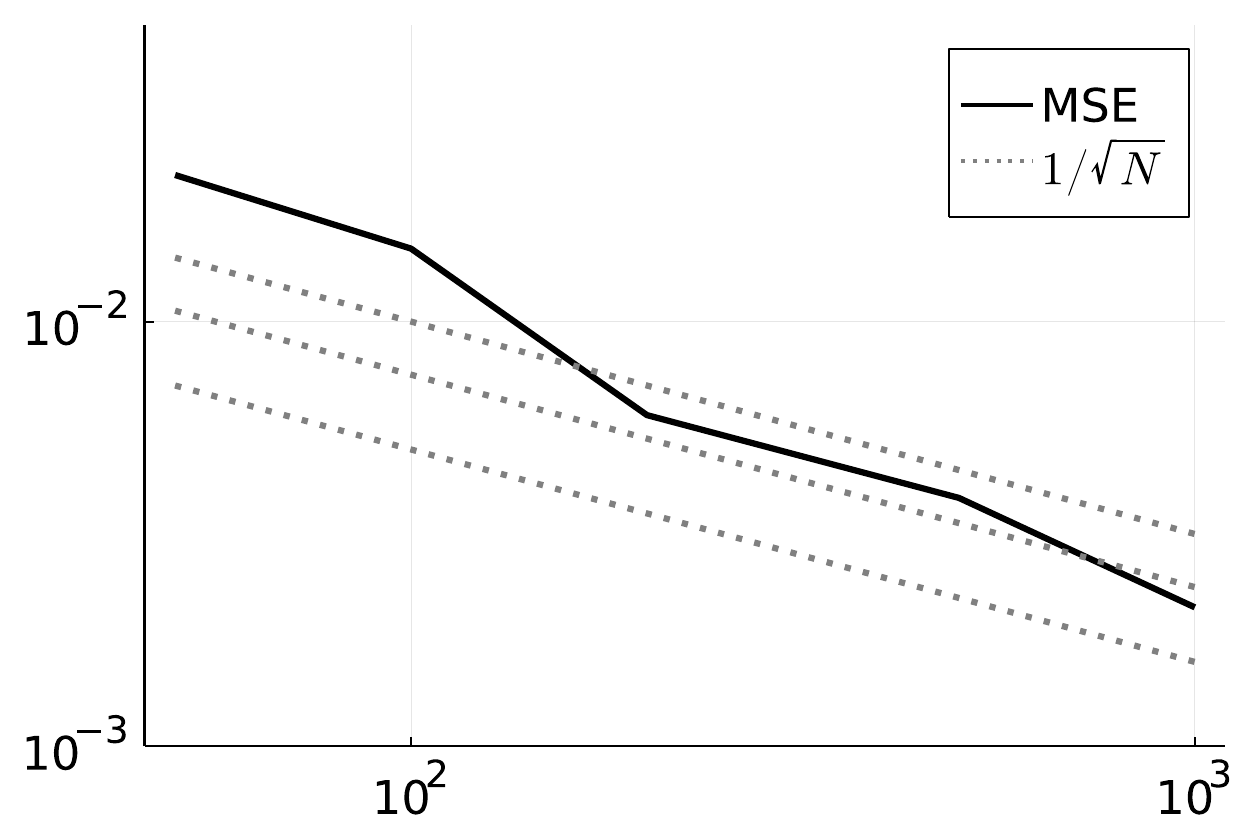}};
\node[below=of img2, node distance = 0, yshift = 1cm] {$N$};
  \node[left=of img2, node distance = 0, rotate=90, anchor = center, yshift = -0.8cm] {$\mse[\textrm{mean}]$};
  \end{tikzpicture}
\caption{Left: Decay of $\wassersteinD[1]$ at the final time step $n_T$ for increasing number of particles $N$. Right: Decay of $\mse^{1/2}$ for the mean of $\pi$ for increasing number of particles $N$; we add the theoretical rate $1/\sqrt{N}$ for reference. The results are averaged over 10 replicates.}
\label{fig:Ngamma}
\end{figure}

We repeat the same analysis for $\gamma$ to validate the rate in~\eqref{eq:euler_error}. We use the same set up as before with $N=500$ and consider $\gamma = $.
Figure~\ref{fig:Ngamma2} shows the mean squared error ($\mse$) for the mean of $\pi$. We add the theoretical rate in~\eqref{eq:euler_error} for reference. 

\begin{figure}
\centering
\begin{tikzpicture}[every node/.append style={font=\normalsize}]
\node (img1) {\includegraphics[width=0.6\textwidth]{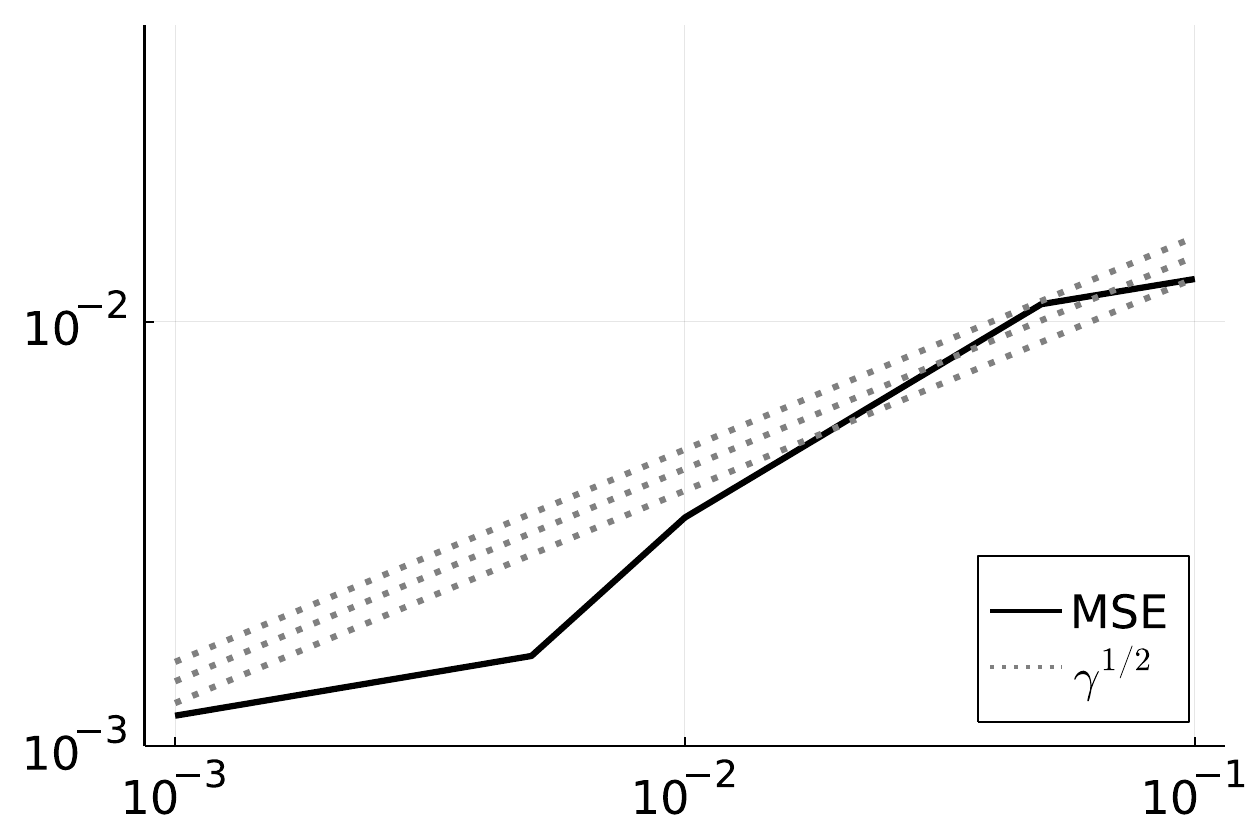}};
\node[below=of img1, node distance = 0, yshift = 1cm] (label1) {$\gamma$};
  \node[left=of img1, node distance = 0, rotate=90, anchor = center, yshift = -0.8cm] {$\mse[\textrm{mean}]$};
  \end{tikzpicture}
\caption{Behaviour of $\mse$ for the mean of $\pi$ for increasing $\gamma$; we add the theoretical rate $\sqrt{\gamma}$ for reference. The results are averaged over 10 replicates.}
\label{fig:Ngamma2}
\end{figure}